\documentclass[10pt,a4paper, fleqn]{article}
\usepackage{etex}
\usepackage[utf8]{inputenc}
\usepackage[TS1, T1]{fontenc}
\usepackage{xspace}
\usepackage{upgreek}
\usepackage[symbol]{footmisc}

\usepackage[a4paper, tmargin=3cm, bmargin=3cm, hmargin=2.5cm]{geometry}

\usepackage{ifthen}
\usepackage{xcolor}

\usepackage{tikz-cd}
\usepackage[english]{babel}

\usepackage[bitstream-charter, cal=cmcal]{mathdesign}
\usepackage{amsfonts}
\usepackage{lmodern}
\usepackage{enumitem}
\newlist{enum-hypothesis}{enumerate}{1}
\setlist[enum-hypothesis]{label=(Hyp.~\arabic*),itemsep=0pt, parsep=0pt, labelwidth=5em, leftmargin=5em}

\setlist[enumerate,1]{label=\arabic*., ref=\arabic*, topsep=1pt, itemsep=2pt, parsep=0pt, leftmargin=1.5em, itemindent=0em, labelsep=0.2em, labelwidth=1.3em}
\setlist[enumerate,2]{label=\alph*., ref=\theenumi.\alph*, topsep=1pt, itemsep=2pt, parsep=0pt, leftmargin=0.5em, itemindent=0em, labelsep=0.2em, labelwidth=1.5em}
\setlist[enumerate,3]{label=\roman*., ref=\theenumii.\roman*, topsep=1pt, itemsep=2pt, parsep=0pt, leftmargin=0.5em, itemindent=0em, labelsep=0.2em, labelwidth=1.2em}

\setlist[itemize,1]{topsep=1pt, itemsep=2pt, parsep=0pt}

\usepackage{array}   % for \newcolumntype macro
\newcolumntype{R}{>{\raggedleft\arraybackslash$}p{1.5em}<{$}} % math-mode version of "r" column type
\usepackage{booktabs, dcolumn}
\usepackage{graphicx, subfig}

\usepackage{pbox}

\usepackage{stmaryrd}
\usepackage{amsmath, mathtools, slashed, mathdots}  % remove amssymb with elsarticle
\usepackage[thmmarks, amsmath]{ntheorem}
\newtheorem{theorem}{Theorem}[section]
\newtheorem{proposition}[theorem]{Proposition}
\newtheorem{lemma}[theorem]{Lemma}
\newtheorem{corollary}[theorem]{Corollary}

\newtheorem{definition}[theorem]{Definition}

\theoremsymbol{$\square$}
\theoremstyle{plain}
\theorembodyfont{\upshape}
\newtheorem{remark}[theorem]{Remark}
\theoremsymbol{$\Diamond$}

\theoremsymbol{}
\theoremstyle{break}
\theorembodyfont{\slshape}

\theoremheaderfont{\scshape}
\theorembodyfont{\upshape} 
\theoremstyle{nonumberplain}
\theoremseparator{} 
\theoremsymbol{\rule{1.3ex}{1.3ex}} 
\newtheorem{proof}{Proof}

\usepackage{comment}

\newif\ifdraft
\drafttrue  % preliminary version
\ifdraft

\else
	\excludecomment{Draft}
	\excludecomment{Detail}
	\excludecomment{Important}
\fi

\DeclareSymbolFont{largesymbols}{OMX}{cmex}{m}{n}

\usepackage{hyperref}

\hypersetup{
plainpages=false,
colorlinks=true,% à mettre avant la suite : supprime l'encadrement des liens
linkcolor=black, 
anchorcolor=black, 
citecolor=black, 
urlcolor=black, 
menucolor=black, 
filecolor=black, 
bookmarksopen=true,
bookmarksnumbered=true,
bookmarksdepth=2
}

\hypersetup{
pdfauthor={Serge Lazzarini, Thierry Masson, Louis Usala},
pdftitle={How to Untwist Twisted Gauge Fields},
pdfsubject={},
pdfcreator={pdflatex},
pdfproducer={pdflatex},
pdfkeywords={}
}

\newcommand{\bbR}{\mathbb{R}}

\newcommand{\bbbone}{{\text{\usefont{U}{bbold}{m}{n}\char49}}}
\newcommand{\calA}{\mathcal{A}}

\newcommand{\calF}{\mathcal{F}}
\newcommand{\calG}{\mathcal{G}}
\newcommand{\calH}{\mathcal{H}}

\newcommand{\calP}{\mathcal{P}}
\newcommand{\calQ}{\mathcal{Q}}
\newcommand{\calR}{\mathcal{R}}
\newcommand{\calS}{\mathcal{S}}

\newcommand{\kg}{\mathfrak{g}}
\newcommand{\kh}{\mathfrak{h}}

\newcommand{\ka}{\mathfrak{a}}
\newcommand{\kb}{\mathfrak{b}}

\newcommand{\phyA}{\mathsf{A}}
\newcommand{\phya}{\mathsf{a}}

\newcommand{\phyF}{\mathsf{F}}

\newcommand{\locgamma}{{\underline{\gamma}}}
\newcommand{\locu}{{\underline{u}}}

\newcommand\cdotact{\mathord{\cdot}}

\newcommand{\defeq}{\vcentcolon=} % \defeq 
\newcommand{\rdefeq}{=\vcentcolon} % \rdefeq 

\DeclareMathOperator{\Ad}{Ad}
\DeclareMathOperator{\Aut}{Aut}
\DeclareMathOperator{\Autv}{Aut_\text{\textup{v}}}

\DeclareMathOperator{\Diff}{Diff} %% Imagesq

\DeclareMathOperator{\Gau}{Gauge} 
\DeclareMathOperator{\id}{id}

\DeclareMathOperator{\im}{im}

\DeclareMathOperator{\Ker}{Ker}	 %% tkernel
\DeclareMathOperator{\M}{M}

\DeclareMathOperator{\prp}{pr_\manP}	 %% Projection
\DeclareMathOperator{\prq}{pr_\manQ}	 %% Projection
\DeclareMathOperator{\rank}{rank}	 %% Endomorphism
\DeclareMathOperator{\Rep}{Rep}  %% Representation
\newcommand{\dr}{\text{\textup{d}}} % differential
\newcommand{\Dr}{\text{\textup{D}}} % covariant differential
\newcommand{\Lr}{\text{\textup{L}}} % Left multiplication
\newcommand{\Rr}{\text{\textup{R}}} % Right multiplication
\newcommand{\dert}[1]{\left(\tfrac{\dr}{\dr t} #1 \right)_{|t=0}} % time derivative at t=0
\newcommand{\inv}{^{-1}} % inverse
\newcommand{\Tg}{\text{\textup{T}}} % Tangent space
\newcommand{\Ver}{\text{\textup{V}}} % Vertical space
\newcommand{\ver}{\text{\textup{v}}} % vertical vector
\newcommand{\Hor}{\text{\textup{H}}} % Horizontal space
\newcommand{\hor}{\text{\textup{hor}}} % horizontal vector
\newcommand{\eqv}{\text{\textup{eq}}} % equivariant function
\newcommand{\Ceq}{\text{$C$\textup{-eq}}} % C-equivariant function
\newcommand{\tens}{\text{\textup{tens}}} % tensorial form
\newcommand{\Ctens}{\text{$C$\textup{-tens}}} % C-tensorial form
\newcommand{\invt}{\text{\textup{inv}}} % invariant form
\newcommand{\bas}{\text{\textup{bas}}} % basic form
\newcommand{\manP}{\calP} % manifold P (principal bundle)
\newcommand{\manQ}{\calQ} % manifold Q (principal bundle)
\newcommand{\manS}{\calS} % manifold S (principal bundle)

\newcommand{\Pz}{\calP^0} % manifold P0 (principal bundle)
\newcommand{\rra}{\rightrightarrows} % groupoid arrows
\newcommand{\ract}{\curvearrowright} % groupoid arrows
\newcommand{\lact}{\curvearrowleft} % groupoid arrows

\newcommand{\Rl}{\mathbb{R}} % Real numbers
\newcounter{mnotecount}[section]
\renewcommand{\themnotecount}{\thesection.\arabic{mnotecount}}
\newcommand{\mnote}[1]%
{\protect{\stepcounter{mnotecount}}${}^{\text{\footnotesize$\bullet$\themnotecount}}$%
\reversemarginpar%
\marginpar{\raggedleft\footnotesize$\bullet$\themnotecount: #1}}

\newlength{\mnotewidth}
\usepackage[shadow,backgroundcolor=red!20]{todonotes}

\usepackage{soul}
\usepackage{color}

\newcommand\keyword[1]{%
  \settowidth{\dimen0}{\textbf{Keywords:}}%
  \sbox0{{\textbf{Keywords:}}}
  \begin{list}{}{
    \setlength\parsep{0pt}
    \setlength\labelwidth{\dimen0}
    \setlength\labelsep{1em}
    \setlength\listparindent{15pt}
    \setlength\leftmargin{\dimen0+1em}
  }%
    \item[\usebox0]\textit{#1}\hfill\par
  \end{list}
}

\newcommand\MSC[1]{%
  \settowidth{\dimen0}{\textbf{2020 MSC:}}%
  \sbox0{{\textbf{2020 MSC:}}}
  \begin{list}{}{
    \setlength\parsep{0pt}
    \setlength\labelwidth{\dimen0}
    \setlength\labelsep{1em}
    \setlength\listparindent{15pt}
    \setlength\leftmargin{\dimen0+1em}
  }%
    \item[\usebox0]#1\hfill\par
  \end{list}
}

\numberwithin{equation}{section}
\allowdisplaybreaks

\begin{document}

\renewcommand*{\thefootnote}{\fnsymbol{footnote}}
\setcounter{footnote}{1}

%%%%%%%%%%%%%%%%%%%%%%%%%%%%%
%%% JOURNAL OF GEOMETRY AND PHYSICS
%
%\title{How to Untwist Twisted Gauge Fields}
%\author[CPT]{Serge Lazzarini}
%\ead{serge.lazzarini@cpt.univ-mrs.fr}
%\author[CPT]{Thierry Masson}
%\ead{thierry.masson@cpt.univ-mrs.fr}
%\author[CPT]{Louis Usala}
%\ead{louis.usala@cpt.univ-mrs.fr}
%
%\address[CPT]{%
%Centre de Physique Théorique\\
%%% WARNING: THE FOLLOWING LINE HAS TO APPEAR AS IT IS WRITTEN (VERBATIM, INCLUDING PUNCTUATIONS)
%%% OUR FUNDINGS DEPEND ON IT
%%% THIS IS WHY WE DO NOT USE THE \affiliation COMMAND PROVIDED IN THE LATEX CLASS
%Aix Marseille Univ, Université de Toulon, CNRS, CPT, Marseille, France
%}

%%%%%%%%%%%%%%%%%%%%%%%%%%%%
%% NON JOURNAL VERSION

{% beginning of local redefinition
%\makeatletter\def\@fnsymbol{\@arabic}\makeatother % redefinition of \@fnsymbol
\title{How to Untwist Twisted Gauge Fields}
\author{Serge Lazzarini, Thierry Masson, Louis Usala
\footnote{serge.lazzarini@cpt.univ-mrs.fr, thierry.masson@cpt.univ-mrs.fr, louis.usala@cpt.univ-mrs.fr}\\
{\small Centre de Physique Théorique}%
\\
\small{Aix Marseille Univ, Université de Toulon, CNRS, CPT, Marseille, France}\\[2ex]
}
\date{June, 24, 2026 (v1)\\ October, 20, 2026 (current version)}

\maketitle
}% end local redefinition
\renewcommand*{\thefootnote}{\arabic{footnote}}
\setcounter{footnote}{0}

\begin{center}
This version is identical to the article published in the Journal of Geometry and Physics.\\
To cite this version: \href{https://doi.org/10.1016/j.geomphys.2026.106028}{https://doi.org/10.1016/j.geomphys.2026.106028}
\end{center}

\begin{abstract}
This paper provides an isomorphism between the space of twisted gauge fields on a principal bundle $\manP(M,H)$ and the space of standard gauge fields on a different principal bundle $\manQ(M,G)$ associated to $\manP$. This isomorphism extends to local fields on the base manifold, which enables the use of local twisted fields in standard gauge theories (e.g. Yang-Mills-like theories). This allows one to deal with two symmetry groups, coming from $\manP$ and $\manQ$, respectively. The construction makes use of a larger principal bundle $\manS$ which has $\manP$ and $\manQ$ as quotient bundles. $\manS$ is a Hilsum-Skandalis bibundle between the action groupoids $H\ltimes\manP$ and $G\ltimes\manQ$. The gauge structure on $\manS$ encodes both standard and twisted gauge structures on $\manP$. In addition, the equivalence classes of bundles $\manS$ up to an isomorphism covering the identity map on $\manP$ are in 1:1 correspondence with the equivalence classes of cocycles (up to a coboundary). This paper also provides a new interpretation of (full) dressing fields as dynamic (or active) sections of a principal bundle.
\end{abstract}
\keyword{Twisted Gauge Fields, Dressing Field Method, Non-Abelian Group Cohomology, Gauge Theories, Extensions of Principal Bundles, Morita Equivalence of Lie Groupoids}
\MSC{53C05, 58A32, 70S15, 53C80, 20J06}
%2020 MSC: 
%53C05% Connections (general theory)
%, 
%58A32% Natural bundles
%, 
%70S15% Yang-Mills and other gauge theories in mechanics of particles and systems
%, 
%53C80% Applications of global differential geometry to the sciences 
%, 
%20J06% Cohomology of groups

\maketitle

%% END OF JOURNAL OF GEOMETRY AND PHYSICS
%%%%%%%%%%%%%%%%%%%%%%%%%%%%

\setcounter{tocdepth}{2}
\tableofcontents

%\newpage

%%%%%%%%%%%%%%%%
\section{Introduction}
\label{sec Introduction}
%%%%%%%%%%%%%%%%

Fiber bundles provide the natural geometric framework for describing physical fields with non-trivial global structure. Locally, a fiber bundle resembles a product $U \times F$ of an open subset $U\subset M$ of the base manifold with a typical fiber $F$, yet its global topology may be genuinely non-trivial, encoding phenomena that no local description can capture. A particularly important class is that of principal bundles $\manP(M,H)$, in which the fiber carries the structure of a Lie group $H$ acting freely and transitively on the total space $\manP$. The structure group $H$ plays the role of a symmetry group for a physical model. Principal bundles occupy a foundational role in differential geometry \cite{kobayashi_foundations_1963, kobayashi_foundations_1969, sharpe_differential_1997}: all other fiber bundles sharing the same symmetry group – the so-called associated bundles – are modeled on them, making the principal bundle the central object from which the full geometric structure of a theory is derived.

This framework finds its most prominent physical applications in gauge theories, where fields are defined so as to carry, at each point of spacetime, additional internal degrees of freedom – a phase, an amplitude, or more generally an element of some internal space. The Standard Model of particle physics \cite{glashow_partial-symmetries_1961, salam_weak_1994, weinberg_model_1967} is most naturally formulated in this language, with the gauge group based on the symmetry group $H=U(1) \times SU(2) \times SU(3)$ defining a principal bundle over spacetime whose sections $\varphi$ of associated bundles\footnote{$(W,\rho_H)$ denotes a representation space of $H$.} $\manP\times_HW$ encode matter fields and principal connection 1-forms $\omega\in\calA(\manP)$ encode interaction fields of the theory. In the same fashion, the curvature associated to $\omega$ models the field strength, the corresponding covariant derivatives models minimal coupling and vertical automorphisms model gauge transformations. Beyond local structure, the topological aspects of principal bundles prove equally indispensable: phenomena such as the Aharonov–Bohm effect \cite{aharonov_significance_1959, wu_concept_1975} and the Berry phase \cite{Berry_1984, simon_holonomy_1983} are fundamentally global in nature and find their precise mathematical expression in the bundle-theoretic setting.

A closely related and geometrically rich construction is that of Cartan geometry \cite{cap_parabolic_2009, sharpe_differential_1997}, which may be understood as a principal bundle equipped with an additional soldering form that constrains the geometry of the fibers to model the geometry of the tangent space to the base manifold. A Cartan bundle is a principal bundle $\manP(M,H)$ modeled on a pair of Lie groups $(G,H)$, $H\subset G$, with Lie algebras $\kg$ and $\kh$. The “model” consists in an isomorphism between $\kg$ and the tangent space $\Tg_p\manP$ at each point $p\in\manP$, which induces an isomorphism between $\kh$ and the vertical tangent space $\Ver_p\manP$. Cartan geometry thus generalizes Riemannian geometry in a natural and flexible way, and provides in particular the appropriate framework for gauge formulations of general relativity. The most notable instance of such formulation is the Cartan tetrad (or vierbein) formalism \cite{krasnov_formulations_2020}, in which the gravitational field is obtained from a connection on a Cartan bundle modeled on the Poincaré and Lorentz groups $(ISO(1,3),SO(1,3))$. This construction admits a natural and physically motivated extension to conformal Cartan geometry \cite{sharpe_differential_1997}, modeled on the group $SO(2,4)$, which captures the conformal symmetry structure of spacetime and underpins several approaches to conformal field theory and twistor theory.

A recurring theme across these geometric frameworks is that of gauge symmetry reduction: the process by which the structure group $H$ of a principal bundle is reduced to a proper subgroup $H' \subset H$. Such reductions arise both for practical computational reasons and for physical ones, when only a subgroup of the full gauge symmetry is deemed physically relevant.

The Dressing Field Method (DFM), developed since 2010, provides a systematic and geometric approach to such reductions, and has emerged as a compelling alternative to more ad hoc procedures \cite{attard_dressing_2018, masson_remark_2010}. In particular, it offers a natural substitute for the Higgs\footnote{Englert-Brout-Higgs-Guralnik-Hagen-Kibble mechanism for completeness} mechanism \cite{englert_broken_1964, guralnik_global_1964, higgs_broken_1964}, circumventing the need for auxiliary hypotheses such as the assumption of a symmetric phase, and avoiding the philosophical difficulties surrounding the physical interpretation of gauge symmetries – which, being redundancies of the mathematical description rather than physical symmetries in the strict sense, cannot be broken in any operationally meaningful way.

The DFM has found a broad range of applications: it yields a reduction of the Cartan tetrad formalism to standard general relativity \cite{fournel_gauge_2014}, a reduction of conformal Cartan geometry to tractor calculus \cite{attard_tractors_2017-1} or twistor geometry \cite{attard_tractors_2017}, a reinterpretation of the Faddeev–Popov gauge-fixing mechanism \cite{guillaud_gauge_2025}, and a foundation for the program of general-relativistic gauge field theory (gRGFT) \cite{berghofer_dressing_2024, francois_relational_2025, francois_geometric_2025, francois_unconventional_2025, francois_invariant_2026}.

The DFM aims at erasing a proper subgroup $K\subset H$ of the structure group. It relies on the existence of a local “dressing field” $u:\manP_U\to H$ satisfying the equivariance property $u(pk)=k\inv u(p)$ w.r.t. any $k\in K$. Provided such a field can be constructed in the model, one can compose it with standard gauge fields $\varphi$ and $\omega$ to get composite fields (or dressed fields) $\varphi^u$ and $\omega^u$ which turn out to be invariant with respect to $K$. This process shall be interpreted as a new arrangement of the degrees of freedom carried by $\varphi$ and $\omega$: those related to $K$, seen as redundancies, are carried by $u$ while the dressed fields carry only the remaining (physical) degrees of freedom. This method leads to a residual gauge theory with structure group $H/K$ (when $K$ is a normal subgroup).

It is worth insisting on the fact that a dressing field is not a gauge group element. Quite the opposite, $u$ reduces the symmetry while a gauge group element $\gamma$ merely shifts its reference. This paper also suggests that $u$ might be better interpreted as related to a local trivializing section of $\manP$, see Sec.~\ref{subsec The DFM}.

A notable feature emerging from the application of the DFM to conformal Cartan geometry (and later on, to gRGFT) is the appearance of residual symmetries of a new and unexpected type: rather than transforming in the standard way, the remaining gauge fields exhibit twisted transformation laws under the residual symmetry group. This observation motivates the development of a general theory of Twisted Gauge Fields, initiated in \cite{francois_twisted_2021}. This forms the starting point of the present article.

Let us sketch the origin of this more general mathematical framework for gauge theories. The usual approach for gauge theories is based on the differential geometry of fiber bundles and connections. In this context, a specific Lie group $H_0$ plays two roles. The first one, of common use in theoretical physics, consists in modeling matter fields as functions with values in some representation space $(W,\rho_0)$ of $H_0$. This role is of algebraic nature. The second role, of geometric nature, consists in using $H_0$ as the typical fiber and structure group of a principal fiber bundle $\manP_0(M,H_0)$ over a base manifold $M$. Working globally at the level of $\manP_0$ (as opposed to the local approach on open sets of $M$), these two roles are consistently related by an equivariance condition when one considers, for instance, matter fields as equivariant maps $\varphi_0:\manP_0\rightarrow W$ satisfying $\varphi_0(ph)=\rho_0(h\inv)\cdotact\varphi_0(p)$. The geometric action on the left hand side (LHS) is consistently related to the algebraic action on the right hand side (RHS).

The DFM does not necessarily preserve this relationship between the two roles. Indeed, a dressing field is crafted in order to satisfy the relevant equivariance property $u(pk)=k\inv u(p)$ for any $k\in K$, the subgroup of $H_0$ to be erased. Such a field may be valued in a different (e.g. larger) group than $H_0$ \cite{attard_tractors_2017-1, attard_tractors_2017}. This shall impact the algebraic part of the equivariance relation for the dressed field $\varphi_0^u$ defined by “$u\inv\cdotact\varphi_0$”. When this latter expression exists (in a mathematical sense), $\varphi^u$ does not necessarily belong to $(W,\rho_0)$ anymore. We refer to \cite{attard_tractors_2017-1, attard_tractors_2017} for further details. Concerning the geometric part, when it exists, the residual group $H=H_0/K$ defines a new principal fiber bundle $\manP(M,H)$, which can replace $\manP_0$ to describe the model at hand. However, the residual group $H$ may no longer be consistent with the algebraic part.

This mismatch is the reason why the residual gauge model – after application of the DFM – requires a more general framework to restore an equivariant relation between the geometric and algebraic parts. The minimal extension of the usual framework consists in considering
\begin{itemize}
\item the principal fiber bundle $\manP(M,H)$ as the geometric part,
\item a Lie group $G$ together with a representation space $(V,\rho_G)$, modeling the algebraic part for the dressed fields,
\item a map $C:\manP\times H\rightarrow G$ used to relate these two actions through an equivariance relation of the form $\varphi(ph)=\rho_G(C(p,h)\inv)\cdotact\varphi(p)$ for any $(p,h)\in\manP\times H$, where $\varphi:\manP\rightarrow V$. 
\end{itemize}
Note that the twisted field $\varphi$ is the field induced on $\manP$ by the $K$-invariant field $\varphi_0^u$ on $\manP_0$. One can forget about this lineage and consider twisted fields for themselves, which is the approach followed in this article.

Since, on the geometric part, one has $p(hh')=(ph)h'$ for any $(p,h,h')$, one gets from the equivariance relation\footnote{$\rho_G$ is omitted from now on.}
\begin{align*}
&\varphi(p(hh')) 
= C(p,hh')\inv\cdotact\varphi(p)
\\
&= \varphi((ph)h')
= C(ph,h')\inv\cdotact\varphi(ph)
= C(ph,h')\inv C(p,h)\inv\cdotact\varphi(p)
\end{align*}
which leads to the so-called cocycle relation $C(p,hh')=C(p,h)C(ph,h')$. This new standpoint is a generalization of the standard framework, which is recovered in the specific case $G=H$ and $C(p,h)=h$, see Sec.~\ref{subsec Identity Map}. We denote $(\manP(M,H),G,C)$ this new structure. The previous cocycle relation is nothing but the defining relation for cocycles appearing in the context of twisted gauge fields, \textit{cf.} \cite{francois_twisted_2021}.

The present work interprets twisted gauge fields on $(\manP(M,H),G,C)$ as standard gauge fields on an additional principal bundle $\manQ(M,G)$ whose structure group is the action group of the twisted framework. The relevant choice for $\manQ$ is a twisted associated bundle to $\manP$ of the form $\manP\times_{C(H)}V$. This construction relies on a correspondence space $\manS(M,H\times G)$ which has both $\manP$ and $\manQ$ as quotient bundles. Note that one can still define standard $H$-gauge fields on $\manP$ together with the previously mentioned twisted gauge fields. In the current framework, both standard and twisted fields on $\manP$ appear as pullbacks of standard $(H\times G)($-gauge fields on $\manS$. In particular, these fields on $\manS$ split naturally into $H$-fields (which induce the standard fields on $\manP$) and $G$-fields (which induce the twisted ones). Going from $\manS$ to $\manQ$ amounts to retaining the $G$-structure while forgetting the $H$ one.

From a structural point of view, $\manS$ appears as more than a mere correspondence space for twisted and standard gauge fields. This work highlights a deep link between $\manS$ and the cohomology of the cocycle map $C$. Note that the geometry of $\manS$ (as a $(H \times G)$-principal fiber bundle) depends on the cocycle considered. In fact, the equivalence classes of $\manS$ up to a principal bundle isomorphism covering the identity on $\manP$ are in 1:1 correspondence with the equivalence classes of $C$ up to a coboundary. This implies that two equivalent cocycles $C$ and $C'$ induce the same space $\manS$ in which two different submanifolds are canonically identified with $\manP$. On the contrary, two non-equivalent cocycles generate two non-isomorphic $\manS$.

\smallskip
The paper is organized as follows. Sec.~\ref{sec Standard and Twisted Gauge Fields} recalls the theory of standard and twisted gauge fields. It also contains a reminder on the DFM with a new interpretation. In Sec.~\ref{sec Isomorphism}, $\manS$ and $\manQ$ are introduced and an isomorphism between the twisted structure on $\manP$ and the standard one on $\manQ$ is displayed. Sec.~\ref{sec Connections GT on S} highlights a correspondence between connections on $\manS$ and pairs of standard and twisted connections on $\manP$. An equivalent correspondence is displayed for gauge transformations. Sec.~\ref{sec Groupoids} provides an alternative approach in terms of action groupoids and interprets $\manS$ as a Hilsum-Skandalis bibundle between $H\ltimes\manP$ and $G\ltimes\manQ$ inducing a Morita equivalence between them. Sec.~\ref{sec Cocycles} focuses on the behavior of this framework under a change of cocycle. In particular, it draws a link between the cohomology of $C$ and the geometry of $\manS$. Sec.~\ref{sec Examples of cocycles} details several particular cases when $C$ has specific properties. Finally, Sec.~\ref{sec Local POV} sets the local version of these constructions on $U\subset M$. It makes use of the DFM to parametrize sections of the principal bundles and it suggests a use of twisted fields in standard gauge theories.

%%%%%%%%%%%%%%%%
\section{Standard and Twisted Gauge Fields, a Reminder}
\label{sec Standard and Twisted Gauge Fields}
%%%%%%%%%%%%%%%%

%%%%%%%%%%%%%%%%
\subsection{Standard Gauge Fields}
\label{subsec Standard Gauge Fields}
%%%%%%%%%%%%%%%%

Let $\manP(M,H)$ be a principal bundle over a manifold $M$, with structure group a Lie group $H$ (whose Lie algebra reads $\kh$). $\calR_h p \rdefeq ph$ denotes the right action of $H$ on $\manP$ and $\pi_\manP:\manP\rightarrow M$ is the corresponding projection. Let $(W,\rho_H)$ be a representation space for $H$, \textit{i.e.} $W$ is a vector space and $\rho_H:H\rightarrow \Aut(W)$ is a group homomorphism. The right action of $H$ on $W$ reads $a_h(w) \defeq \rho_H(h\inv)\cdotact w=h\inv\cdotact w$, for all $h\in H$ and $w\in W$, when the use of $\rho_H$ is unambiguous. Indeed, $a_h$ satisfies $(a_h\circ a_{h'})(w) = a_h(a_{h'}(w)) = h\inv\cdotact a_{h'}(w) = h\inv\cdotact ({h'}\inv \cdotact w) = (h\inv {h'}\inv )\cdotact w = (h'h)\inv\cdotact w = a_{h'h}(w)\implies a_h\circ a_{h'} =  a_{h'h}$. Notice that $\rho_H$ itself induces a left action of $H$ on $W$. 

Given all the above, the (standard) associated vector bundle $E_\manP \defeq \manP\times_H W \defeq (\manP\times W)/H$ is the set of equivalence classes on $\manP\times W$ under the relation $(p,w)\sim(ph,h\inv \cdotact w),\,\forall h\in H$. Among the fields on $\manP$, \textit{i.e.} $f:\manP\to F$, some are said ($H$-)equivariant when their pullbacks through $\calR_h$ are given by a right action $\text{Act}$ of $H$ on their target space: $\calR_h^*f=\text{Act}_{h\inv }f$. For example, $\varphi\in\Omega^0(\manP,W)=\calF(\manP,W)$ such that $(\calR_h^*\varphi)(p)=\varphi(ph)=\rho_H(h\inv)\cdotact\varphi(p)$ is said $\rho_H$-equivariant. The corresponding space of equivariant maps will be denoted $\calF_\eqv(\manP,W)$. It is well known that such equivariant fields are in bijection with sections $\phi\in\Gamma(E)$ of the corresponding associated bundle. Explicitly, the correspondence is given by
\begin{align} 
\label{eq correspondence section to eq maps}
\begin{aligned}
\varphi &\in\calF_\eqv(\manP,W)
&\longleftrightarrow
&
&\phi : x \mapsto[p, \varphi(p)] = [ph,  h\inv \cdotact  \varphi(p)] = [ph,\varphi(ph)] \in \Gamma(E), %%%\manP\times_H W),
\end{aligned}
\end{align}
for any $p\in\pi_\manP\inv(x)$. An equivariant field is fully determined by its values on a section of $\manP$. Indeed, knowing $f$ at a given point $p$, one deduces its value at any other point in the fiber of $\manP$ using the equivariance relation. A particular case of an equivariant field is an invariant field, satisfying $\calR_h^*f=f$.

Any tangent linear space $\Tg_p\manP$ at $p\in\manP$ has a canonical “vertical” subspace $\Ver_p\manP$. This is the space of vertical vectors, \textit{i.e.} tangent vectors to curves which stay in the same fiber as $p$ locally. As is well-known, $\Ver_p\manP\simeq \kh$. Choosing a “horizontal” supplementary subspace $\Hor_p\manP$ such that $\Tg_p\manP=\Ver_p\manP\oplus \Hor_p\manP$ is equivalent to choosing an (Ehresmann) connection 1-form $\omega\in\Omega^1(\manP,\kh)$ satisfying the following vertical normalization and $H$-equivariance with respect to the adjoint action:
\begin{itemize}
\item $\omega_p(\ka_p^\ver)=\ka,\,\forall \ka\in\kh$,
\item $\calR_h^*\omega=\Ad_{h\inv }\omega,\,\forall h\in H$.
\end{itemize}
In the vertical normalization, $\ka_p^\ver \defeq \dert{p\exp(t\ka)}\in \Ver_p\manP$. The horizontal tangent space is then defined as $\Hor_p\manP \defeq \Ker\omega_p$. The set of connection 1-forms on $\manP$ will be denoted $\calA(\manP)$ and it has a structure of affine space modeled on the vector space of $\Ad$-tensorial $\kh$-valued 1-forms, see below. Let $\rho_*$ be the representation of $\kh$ induced\footnote{$\rho_*(\ka)\cdotact w=\dert{\rho(e^{t\ka})\cdotact w},\,\ka\in\kh$} on $W$ by $\rho \defeq \rho_H$. Any connection induces a covariant differential $\Dr=\dr+\rho_*(\omega)$ on ($\bbR$-valued, $W$-valued, $\kh$-valued...)\footnote{The implicit action of $\omega$ is the action of $\kh$ induced by the action of $H$. It depends on the space in which the form is valued.} differential forms and a curvature 2-form $\Omega=\dr\omega+\frac{1}{2}[\omega,\omega]\in\Omega^2(\manP,\kh)$. Like $\omega$, $\Omega$ is $\Ad$-equivariant.

A differential form is called horizontal when it vanishes on vertical vector fields. The space of $S$-valued horizontal forms (for a given vector space $S$) is $\Omega_\hor^\bullet(\manP,S)=\{\alpha\in\Omega^\bullet(\manP,S),\,\alpha_p(X_p^0,X_p^1,\ldots,X_p^k)=0\ \forall X_p^0\in \Ver_p\manP;\,X_p^1, \ldots, X_p^k\in \Tg_p\manP\}$. In particular, any 0-form is horizontal. An equivariant horizontal form is called tensorial and an $H$-invariant horizontal form is called basic. These spaces are denoted $\Omega_\hor^\bullet,\ \Omega_\eqv^\bullet,\ \Omega_\invt^\bullet,\ \Omega_\tens^\bullet,\ \Omega_\bas^\bullet$. Note that a $H$-tensorial form on $\manP$ induces a differential form on $M$ and a $H$-basic form on $\manP$ induces an invariant form on $M$. For instance, the curvature $\Omega$ is $\Ad$-tensorial but the connection $\omega$ is not.

The group of vertical automorphisms $\Autv(\manP)=\{\Phi\in\Diff(\manP) \mid\Phi(ph)=\Phi(p)h,\ \pi_\manP\circ\Phi=\pi_\manP\}$ acts naturally on the previous fields by pullback. $\Autv(\manP)$ is isomorphic to the gauge group $\calH^\manP=\{\gamma : \manP \rightarrow H \mid \gamma(ph)=h\inv \gamma(p)h\}$ with the correspondence $\Phi(p)=p\gamma(p)$. Note that both $\Phi$ and $\gamma$ are defined through their equivariance properties: $\calR_h^*\Phi=\calR_h\circ\Phi$ and $\calR_h^*\gamma=h\inv\gamma h$. To be more explicit, the gauge transformations of the previous fields are given by
\begin{align*}
\begin{aligned}
\omega^\gamma 
&\defeq \Phi^*\omega 
= \gamma\inv\omega\gamma+\gamma\inv\dr\gamma,
%\\
&
\Omega^\gamma 
&\defeq \Phi^*\Omega 
= \gamma\inv\Omega\gamma,
\\
\varphi^\gamma 
&\defeq \Phi^*\varphi 
= \rho(\gamma\inv)\cdotact\varphi,
%\\
&
\Dr^\gamma \varphi^\gamma 
&\defeq \Phi^*\Dr\varphi 
= \rho(\gamma\inv)\cdotact \Dr\varphi.
\end{aligned}
\end{align*}
In addition, the elements of the gauge group themselves are subject to a specific gauge transformation: $\forall \gamma,\gamma'\in\calH^\manP,\,\gamma'^\gamma \defeq \Phi^*\gamma'=\gamma\inv\gamma'\gamma$.

Any gauge field defined on $\manP$ induces a field locally on $M$. Let $U\subset M$ be an open subset and consider a (local) smooth trivializing section of $\manP$ denoted $\sigma:U\rightarrow\manP_U \defeq \pi_\manP\inv(U)$, which satisfies $\pi_\manP\circ\sigma=\id_U$. Given a local section, one can define fields on $U$ as pullbacks by $\sigma$ of fields on $\manP$. For example, consider $\varphi^\manP\in\calF_\eqv(\manP,W)$, $\omega^\manP\in\calA(\manP)$ and $\alpha^\manP\in\Omega_\tens^\bullet(\manP,W)$. These fields give rise to $\phi=\sigma^*\varphi^\manP\in\calF(U,W)$, $\phyA \defeq \sigma^*\omega^\manP\in\Omega^1(U,\kh)$ and $\phya \defeq \sigma^*\alpha^\manP\in\Omega^\bullet(U,W)$. Let $U,U'\subset M$ be open subsets such that $U\cap U'\neq\varnothing$ and consider the sections $\sigma$ and $\sigma'$ of $\manP$ on $U$ and $U'$ respectively. Let $g:U\cap U'\rightarrow H$ be the corresponding gluing function, \textit{i.e.} $\forall x\in U\cap U',\ \sigma'(x)=\sigma(x)g(x)$. With the notations $\phi=\sigma^*\varphi^\manP,\ \phi'=\sigma^{\prime*}\varphi^\manP$, one get that for any $x\in U\cap U'$,
\begin{align}
\label{eq loc phi 1}
\phi'(x)
&= \rho(g(x)\inv)\cdotact\phi(x).
\end{align}
With similar notations,
\begin{align}
\label{eq loc A a 1}
\begin{aligned}
\phyA'_x
&= g(x)\inv \phyA_x g(x)+ g(x)\inv\dr g_x,
%\\
&
\phya'_x
&= \rho(g(x)\inv)\cdotact \phya_x.
\end{aligned}
\end{align}
As said before, since the curvature $\Omega^\manP$ is $\Ad$-tensorial, it yields a globally defined object on the base manifold $M$, say $F\in \Omega^2(M,\kh)$. It can locally be identified with
the pull-back $\phyF \defeq \sigma^*\Omega^\manP\in\Omega^2(U,\kh)$ such that $\phyF=d\phyA + \phyA^2$ thanks to the pull-back of the Cartan structure equation. One has $\phyF'_x = g(x)\inv \phyF_x g(x)$.

A change of local trivializing section is often called “(local) passive gauge transformation” and is closely related to a so-called “(local) active gauge transformation”, namely the pullback of a vertical automorphism by a local trivializing section.

Given a local trivializing section $\sigma$ of $\manP$, let us define the local gauge group $\calH_U \defeq \{\locgamma=\sigma^*\gamma \mid \gamma\in\calH^\manP\}$. The defining equivariance property of $\calH^\manP$ induces a characterization of the local gauge group: $\calH_U=\{\locgamma : U \rightarrow H \mid \forall \locgamma'\in\calF(U,H),\ \locgamma^{\locgamma'}=\locgamma^{\prime-1}\locgamma\locgamma'\}$. This characterization is independent of both $\sigma$ and $\manP$. The action of $\calH^\manP$ on gauge fields induces an action of $\calH_U$ on their local counterparts. The local active gauge transformations of the previous local fields are
\begin{align}
\label{eq standard GT}
\begin{aligned}
\phi^\locgamma(x)
&= \rho(\locgamma(x)\inv)\cdotact\phi(x),
%\\
&
\phyA^\locgamma_x
&= \locgamma(x)\inv \phyA_x \locgamma(x)+\locgamma(x)\inv\dr\locgamma_x,
\\
 \phyF^\locgamma_x
& = \locgamma(x)\inv \phyF_x \locgamma(x),
%\\
&
\phya^\locgamma_x
&=\rho(\locgamma(x)\inv)\cdotact \phya_x.
\end{aligned}
\end{align}
These transformations are formally identical to passive gauge transformations. The compatibility between \eqref{eq loc A a 1} and \eqref{eq standard GT} yields $\locgamma' = g\inv \locgamma g$.

This framework is of great use in physics since it underlies gauge theories, in particular the Standard Model of Particle Physics. In this framework, $M$ models spacetime up to a diffeomorphism, $H$ carries internal degrees of freedom, $\phi\in\Gamma(E)$ and $\phyA \in\Omega^1(U,\kh)$ correspond to matter fields and interaction fields respectively and $\phyF\in\Omega^2(U,\kh)$ is the field strength.

%%%%%%%%%%%%%%%%
\subsection{The Dressing Field Method (DFM)}
\label{subsec The DFM}
%%%%%%%%%%%%%%%%

The previous section highlights two ways of changing the reference of local gauge fields. One can either change the section used to pullback fields on $\manP$ (“passive gauge transformation”) or apply a vertical automorphism on $\manP$ and pullback through the same section (“active gauge transformation”). In this sense, $\gamma\in\calH^\manP$ can be seen as the “active” (or “dynamic”) counterpart of a gluing function. Similarly, a dressing field $u$ (defined in the following) is a map on $\manP$ crafted in such a way that it is the active counterpart of a section $\sigma$ on $U$. The scheme of this approach is the following:
\begin{align}
\label{tab DFM vs section}
\text{\begin{tabular}{ccc}
 & ...passively & ...actively\\
changes the reference... & $g(x)$ & $\gamma(p)$\\
selects a reference... & $\sigma(x)$ & $u(p)$
\end{tabular}}
\end{align}
Recall that a principal bundle is canonically associated to itself. Indeed, $\manP\times_HH\simeq\manP$ with the two (mutually inverse) isomorphisms
\begin{align*}
%\begin{aligned}
\manP &\rightarrow \manP\times_HH
&
\manP\times_HH &\rightarrow \manP
\\
p &\mapsto[p,e]=[ph,h\inv]
&
[p,h] &\mapsto ph
%\end{aligned}
\end{align*}
These two isomorphisms together with the correspondence \eqref{eq correspondence section to eq maps} lead to the isomorphism $\calF_\eqv(\manP,H)\simeq\Gamma(\manP\times_HH)\simeq\Gamma(\manP)$ given by
\begin{align*}
\begin{aligned}
u &\in \calF_\eqv(\manP,H) 
&
\longleftrightarrow
&
&
\bar{\sigma} : x &\mapsto [p, u(p)] = [ph, h\inv u(p)] = [ph,u(ph)] \in\Gamma(\manP\times_HH)
\\
&
&
\longleftrightarrow
&
&
\sigma : x &\mapsto pu(p) = phh\inv u(p) = phu(ph) \in \Gamma(\manP)
\end{aligned}
\end{align*}
where $p\in\pi_\manP\inv(x)$ is arbitrary. Note that a global section of $\manP$ exists if and only if $\manP\simeq M\times H$ is a trivial bundle. However, the previous correspondences still hold locally: $\calF_\eqv(\manP_U,H)\simeq\Gamma(\manP_U)$. The global correspondence between dressing fields and sections is discussed in \cite{zajac_dressing_2023} and generalized to partial dressings, see below.
\begin{definition}\label{def dressing}
A \textbf{$H$-dressing field} over $U$ on $\manP$ is an equivariant map $u\in\calF_\eqv(\manP_U,H)$ with equivariance property $\calR_h^*u(p)=h\inv u(p)$.
\end{definition}
One checks that this definition induces the gauge transformation
\begin{align}\label{eq GT u}
u^\gamma(p)=\gamma(p)\inv u(p)
\end{align}
The above stated information imply that
\begin{itemize}
\item a dressing field $u$ is equivalent to a local section $\sigma$ of $\manP$ over $U$,
\item the existence of a global dressing field on $\manP$ is equivalent to the triviality of $\manP\simeq M\times H$.
\end{itemize}
Similarly to the isomorphism $\calH^\manP\rightarrow\Autv(\manP),\ \gamma\mapsto \Phi:p\mapsto p\gamma(p)$, dressing fields correspond to $H$-invariant maps. Let $u\in\calF_\eqv(\manP_U,H)$ and let us define
\begin{align*}
\begin{aligned}
\Sigma : \manP_U &\rightarrow \manP_U, 
&
p &\mapsto pu(p).
\end{aligned}
\end{align*}
This map trivially satisfies $\calR_h^{\manP*}\Sigma=\Sigma$ (contrary to vertical automorphisms $\Phi$ satisfying $\calR_h^{\manP*}\Phi(p)=\Phi(p)h$). Thus it induces a map on $U$ which is nothing else but $\sigma$. $\Sigma$ will play a key role in the DFM, see below. The dressing field is indeed the concept which consistently fills the bottom right cell of table \eqref{tab DFM vs section}.
\begin{proposition}
The passive and active gauge transformations are consistent. Given any sections $\sigma,\ \sigma'$ glued by $g$ over $U\cap U'\neq 0$, their active counterparts $u,\ u'$ satisfy $u'=u^\gamma=\gamma\inv u$ where $\gamma$ is the (active) gauge element corresponding to $g\inv$.
\end{proposition}
Note that $g=\sigma^*\gamma\inv$ instead of $g=\sigma^*\gamma$ is a matter of convention. It depends on the chosen conventions for gluing functions and for the action of the gauge group on $\manP$. The current conventions are commonly found in the literature.
\begin{proof}
$\sigma$ (resp. $\sigma'$) has active counterpart $u$ (resp. $u'$), which means that: $\forall x\in U$,
\begin{align*}
\begin{aligned}
\sigma(x) &= pu(p), 
& 
\sigma'(x) &= pu'(p),
\end{aligned}
\end{align*}
where $p\in\pi_\manP\inv(x)$ is arbitrary. On the one hand,
\begin{align*}
\sigma'(x) &= \sigma(x)g(x).
\end{align*}
On the other hand, suppose $u'=\gamma\inv u$. Then, using $\gamma(p u(p)) = u(p)\inv \gamma(p) u(p)$, one has
\begin{align*}
\sigma'(x) 
&= pu'(p)
= p\gamma\inv(p)u(p)
= pu(p)u(p)\inv\gamma\inv(p)u(p)
= pu(p)\gamma(pu(p))\inv
%\\
%&
= \sigma(x)\gamma(\sigma(x))\inv.
\end{align*}
Since the action of $H$ on $\manP$ is assumed to be free, then $g(x)=\gamma(\sigma(x))\inv=\sigma^*\gamma\inv(x)$.
\end{proof}
The DFM consists in composing a dressing field with the standard fields of a gauge theory in order to get gauge-invariant fields. Given $u$, and so $\Sigma$, one can pullback gauge fields $\omega,\ \Omega,\ \varphi$ and $\Dr\varphi$ through $\Sigma$ in order to define so-called dressed (or composite) fields on $\manP_U$
\begin{equation}
\label{eq:DFM/P}
\begin{aligned}
\omega^u &\defeq \Sigma^*\omega = u\inv\omega u+u\inv\dr u,
%\\
&
\Omega^u &\defeq \Sigma^*\Omega = u\inv\Omega u,
\\
\varphi^u &\defeq \Sigma^*\varphi = \rho(u\inv)\cdotact\varphi,
%\\
&
\Dr^u \varphi^u &\defeq \Sigma^*\Dr\varphi = \rho(u\inv)\cdotact \Dr\varphi.
\end{aligned}
\end{equation}

\begin{lemma}
The dressed fields are invariant under the action of $\calH^\manP$.
\end{lemma}
\begin{proof}
Let $u$ be a dressing field and $\Sigma:p\mapsto pu(p)$. As shown before, $\Sigma$ is an $H$-invariant map. Thus, for any $\gamma\in\calH^\manP\leftrightarrow\Phi\in\Autv(\manP)$,
\begin{align*}
\Sigma\circ\Phi(p)
&= \Sigma(p\gamma(p))
= p\gamma(p)u(p\gamma(p))
= pu(p)
= \Sigma(p).
\end{align*}
Let $\omega\in\calA(\manP)$. (The same proof holds for gauge fields of any other type.)
\begin{align*}
(\omega^u)^\gamma
&= \Phi^*\Sigma^*\omega 
= (\Sigma\circ\Phi)^*\omega
= \Sigma^*\omega
= \omega^u.
\end{align*}
Thus $\omega^u$ is a gauge-invariant field. 
\end{proof}
The use of $\Sigma$, already appearing in \cite{attard_tractors_2017-1,attard_tractors_2017}, provides a geometric interpretation of the dressing. The pullback $\varphi^u=\Sigma^*\varphi$ of a generic gauge field $\varphi$ through $\Sigma$ consists in prescribing the value $\varphi^u(p)=\varphi(pu(p)) =\varphi(\sigma(x))$ for the dressed field at any point $p$ in the fiber $\pi_\manP\inv(x)$, where the local section $\sigma$ is defined from $u$ as before. Then the equivariance of each gauge field promotes this geometric operation to an algebraic operation of the type $\varphi^u(p)=u(p)\inv\cdotact\varphi(p)$, as seen already in \cite{masson_remark_2010}. This recovers the approach based on field composers, of algebraic nature, used in \cite{guillaud_gauge_2025} for instance.

Since dressed fields are gauge-invariant, they induce gauge-invariant local fields $\phyA^\locu,\ \phyF^\locu,\ \phi^\locu,\ \Dr^\locu\phi^\locu$ on $U\subset M$ by taking the pullbacks of \eqref{eq:DFM/P} by the local trivializing section $\sigma_0$, with $\locu\defeq \sigma_0^*u$. Following the gauge transformation \eqref{eq GT u} one has $\locu^\locgamma=\locgamma\inv\locu$, which is similar to $\locu'=g\inv\locu$ for a change of local trivializing sections, see \eqref{eq loc A a 1}.

It is worth noticing an important feature of local dressed fields. Let $\sigma_0$ and $\sigma_0'$ be local trivializing sections on $U$ and $U'$ such that $U\cap U'\neq\varnothing$, with gluing function $g$. The local versions of the dressed fields are linked by
\begin{align*}
(\phyA^{\locu})'\defeq \sigma_0^{\prime *}\omega^u=\sigma_0^{\prime *}\Sigma^*\omega=(\Sigma\circ\sigma_0')^*\omega
\end{align*}
Let us remark that
\begin{align*}
\Sigma\circ\sigma_0'(x)=\sigma_0(x)g(x)u(\sigma_0(x)g(x)) =\sigma_0(x)u(\sigma_0(x))=\Sigma\circ\sigma_0(x)
\end{align*}
This implies that
\begin{align}\label{eq:AuInvarDress}
(\phyA^{\locu})'\defeq \sigma_0^{\prime *}\omega^u=\sigma_0^{\prime *}\Sigma^*\omega=\sigma_0^*\Sigma^*\omega=\sigma_0^*\omega^u =\phyA^{\locu}
\end{align}
In other words, local dressed fields are independent of the chosen local description. One readily checks that $(\phyA^\locu)'=(\phyA')^{\locu'}=\phyA^\locu$. The DFM enables to effectively reduce the structure group. Indeed, the dressed fields (interpreted as physically relevant quantities) are gauge-invariant. However, contrary to a mere gauge fixing, $H$ is still part of the theory since $u$ is $H$-equivariant. The correct way to interpret this trick is to see the pair $(\varphi^u,u)$ as carrying the same information as $\varphi$.\footnote{To be precise, there is the same information carried by all the gauge fields $\{\varphi,\ \omega,...\}$ of the model and by $(\{\varphi^u,\ \omega^u,...\},u)$.} The dressing field inherits the degrees of freedom related to the gauge group while the dressed field carries only the physically relevant degrees of freedom. The existence of an undressing process ensures that no information is lost, see \cite{guillaud_gauge_2025}.

This section introduces only “full” dressing procedures. Indeed, $u$ is required to satisfy the equivariance relation in Def.~\ref{def dressing} for any $h\in K=H$. One can equally define “partial” dressing fields satisfying $\calR_k^*u(p)=k\inv u(p)$ for any $k\in K$ where $K\subset H$ is a proper subgroup. If $K$ is normal, one may prescribe an equivariance relation for $u$ with respect to the residual structure group $H/K$. The details of this process is out of the scope of this paper, see \cite{attard_tractors_2017-1, attard_tractors_2017} for more details. Note that the correspondence between local trivializing sections and equivariant $H$-valued local maps already appears in \cite{daniel_geometrical_1980} with a field “$\phi_\alpha$” corresponding to $u\inv$. However, this was not used to generate (partially or fully) gauge-invariant composite fields as with the DFM, see below.

%%%%%%%%%%%%%%%%
\subsection{Twisted Gauge Fields}
\label{subsec Twisted Gauge Fields}
%%%%%%%%%%%%%%%%

As mentioned in the introduction, twisted structures are a generalization of the standard gauge structures described in Sec.~\ref{subsec Standard Gauge Fields}. We present here a short review of the concepts used in the following and we refer to \cite{francois_twisted_2021} for further details.

In the following, $\manP,M,H,\calR,\pi_\manP$ are defined as in Sec.~\ref{subsec Standard Gauge Fields}. Let $G$ be a Lie group (with Lie algebra $\kg$) and $(V,\rho_G)$ be a representation space for $G$.
\begin{definition}
A map $C:\manP\times H\rightarrow G$ is called \textbf{cocycle} (or \textbf{group action cocycle}, or \textbf{1-cocycle}) if it satisfies the cocycle relation
\begin{align}
\label{cocycle relation}
C(p,hh')
&= C(p,h)C(ph,h') \text{ for any $p\in\manP$ and $h,h'\in H$}.
\end{align}
\end{definition}
From this relation, one deduces that $C(p,e_H)=e_G$ and $C(p,h)\inv=C(ph,h\inv)$. As a smooth map w.r.t. both its arguments, $C$ can be differentiated on $\manP$ and $H$ respectively:
\begin{align*}
	\begin{array}{c}
		\dr_\manP C_{|(p,h)}:\\
		 \ 
	\end	{array}{}{}&
	\left|
	\begin{array}{l}
		\Tg_p\manP\rightarrow \Tg_{C(p,h)}G\\
		X_p^\manP\mapsto \dert{C(s(t),h)}\hspace{2.9cm} s:(-1,1)\rightarrow\manP,\,s(0)=p,\,\dot{s}(0)=X_p^\manP
	\end{array}
	\right.\\
	\begin{array}{c}
		\dr_H C_{|(p,h)}:\\
		 \ 
	\end	{array}{}{}&
	\left|
	\begin{array}{l}
		\Tg_hH\rightarrow \Tg_{C(p,h)}G\\
		X_h^H\mapsto \dert{C(p,\tilde{s}(t))}\hspace{2.9cm} \tilde{s}:(-1,1)\rightarrow H,\,\tilde{s}(0)=h,\,\dot{\tilde{s}}(0)=X_h^H
	\end{array}
	\right.\\
	\begin{array}{c}
		\dr C_{|(p,h)}:\\
		 \ 
	\end	{array}{}{}&
	\left|
	\begin{array}{l}
		\Tg_p\manP\oplus \Tg_hH\rightarrow \Tg_{C(p,h)}G\\
		X_p^\manP\oplus X_h^H\mapsto \dr_\manP C_{|(p,h)}(X_p^\manP)+\dr_HC_{|(p,h)}(X_h^H)
	\end{array}
	\right.
\end{align*}
Note that $C(p,h)\inv\dr_\manP C_{|(p,h)}(X_p^\manP),\,C(p,h)\inv\dr_H C_{|(p,h)}(X_h^H),\,C(p,h)\inv\dr C_{|(p,h)}(X_p^\manP\oplus X_h^H)\in \Tg_eG$ are seen as elements in the Lie algebra $\kg$. The following result will be useful in several calculations.
\begin{lemma}\label{lemma d_H C}
The differential of $C$ with respect to its second argument satisfies
\begin{align*}
\left(\dr_H C_{(p,h)}\inv \right) C(p,h)
&= -C(p,h)\inv \left( \dr_H C_{(p,h)} \right),
\end{align*}
where, in the left hand side, $\dr_HC_{(p,h)}\inv$ is the differential of $(p,h)\mapsto C(p,h)\inv$ with respect to its second variable. In particular,
\begin{align*}
\dr_H C_{(p,e)}\inv
&= -\dr_H C_{(p,e)}.
\end{align*}
\end{lemma}
\begin{proof}
This result is obtained by deriving $C(p,\tilde{h}(t))\inv C(p,\tilde{h}(t))=e$ at $t=0$, where $\tilde{h}:(-1,1)\to H$ is a smooth curve such that $\tilde{h}(0)=h$.
\end{proof}

\begin{definition}
The \textbf{twisted associated vector bundle} $E_\manP^C=\manP\times_{\rho_G\circ C}V=\manP\times_{C(H)}V=\manP\times_C V \defeq (\manP\times V)/H$ is the set of equivalence classes on $\manP\times V$ under the relation
\begin{align*}
(p,v)\sim(ph,\rho_G(C(p,h)\inv)\cdotact v)
&\rdefeq (ph,C(p,h)\inv\cdotact v).
\end{align*}
\end{definition}
The cocycle relation \eqref{cocycle relation} ensures that $\sim$ is an equivalence relation. The same procedure enables defining any twisted associated bundle, where the associated space is any space with an action of $G$. Any section $\phi\in\Gamma(\manP\times_{C(H)}V)$ is in 1:1 correspondence with a $C(H)$-equivariant map $\varphi\in\Omega_\Ceq^0(\manP,V)$, where the $C(H)$-equivariance relation is $\calR_h^*\varphi(p)=C(p,h)\inv\cdotact\varphi(p)$. Note that twisted fields, together with the twisted forms and structures defined in the following, are a generalization of the standard fields defined in Sec.~\ref{subsec Standard Gauge Fields}. Indeed, in the particular case of $G=H$ and $C(p,h)=h,\ \forall (p,h)$, one recovers standard $H$-fields, see Sec.~\ref{subsec Identity Map}. In the same way, one can generalize the notion of tensorial forms to $C$-tensorial ($V$-valued) forms. These are horizontal forms satisfying a $C(H)$-equivariance relation. The space of $C$-tensorial forms will be denoted $\Omega_\Ctens^\bullet(\manP,V)$.

The relevant process to generalize a horizontal structure to the twisted context is not to prescribe a horizontal supplementary subspace $\Hor_p\manP$ to $\Ver_p\manP$. Indeed, sections of twisted associated bundles carry an action of $G$, whose horizontal directions may not correspond to the ones defined by $H$. Instead, the connection 1-form is promoted to a twisted form.
\begin{definition}
The \textbf{twisted connection 1-form} $\omega\in\Omega^1(\manP,\kg)$ is the $\kg$-valued 1-form satisfying the following vertical normalization and $H$-equivariance:
\begin{itemize}
\item $\omega_p(\ka_p^\ver)=\dr_H C_{|(p,e)}(\ka),\,\forall \ka\in\kh$,
\item $\calR_h^*\omega_{ph}=\Ad_{C(p,h)\inv }\omega_p+C(p,h)\inv \dr_\manP C_{|(p,h)},\,\forall h\in H$.
\end{itemize}
The set of twisted connections on $\manP$ is denoted $\calA^C(\manP)$ and it is an affine space modeled on $\Omega_\Ctens^1(\manP,\kg)$.
\end{definition}
Note that $\dr_H C_{|(p,e)}(\ka)$ is an instance of $C(p,h)\inv\dr_H C_{|(p,h)}(\ka)$ at $h=e$, which highlights that this is an element of $\kg$. Let $\rho_*$ be the representation of $\kg$ induced on $V$ by $\rho \defeq \rho_G$. The exterior covariant differential associated to $\omega$ is $\Dr \defeq \dr+\rho_*(\omega)$ on equivariant maps and tensorial forms. This operator acts on $V$-valued forms on $\manP$ and preserves the horizontality and the $C(H)$-equivariance. In particular, the covariant differential of a $C(H)$-tensorial form is a $C(H)$-tensorial form. Since sections of $\manP\times_{C(H)}V$ are equivalent to $V$-valued $C(H)$-equivariant 0-forms on $\manP$, $\Dr$ provides a notion of covariant derivative of sections. Additionally, $\calR_h^*\Dr\phi=\rho(C(p,h)\inv )\cdotact \Dr\phi$. The twisted connection also induces a twisted curvature $\Omega\in\Omega_\Ctens^2(\manP,\kg)$ defined by $\Omega \defeq \dr\omega+\frac{1}{2}[\omega,\omega]$. Its equivariance condition is $(\calR_h^*\Omega)_p =\Omega_{p h} \circ \Tg_p \calR_h = \Ad_{C(p,h)\inv }\Omega_p$. Moreover, $\Omega$ satisfies the Bianchi identity $\dr\Omega+[\omega,\Omega]=0$, meaning that $\Dr\Omega=0$. Additionally, for any $V$-valued tensorial form $\alpha$, one has $\Dr^2\alpha=\rho_*(\Omega)\cdotact\alpha$.

The gauge group $\calH^\manP\simeq\Autv(\manP)$ acts on twisted fields. Their explicit gauge transformations read, with $C_\gamma:\manP\rightarrow G,\,p\mapsto C(p,\gamma(p))$,
\begin{align*}
\begin{aligned}
\omega^\gamma_p{}&\defeq (\Phi^*\omega)_p
= C_\gamma(p)\inv \omega_p C_\gamma(p) + C_\gamma(p)\inv \dr C_{\gamma\, p}\\
\Omega^\gamma_p{}&\defeq (\Phi^*\Omega)_p
= C_\gamma(p)\inv \Omega_pC_\gamma(p),
\\
\varphi^\gamma(p){}&\defeq (\Phi^*\varphi)(p)
= \rho(C_\gamma(p)\inv )\cdotact\varphi(p),\\
(\Dr\varphi)^\gamma(p){}&\defeq (\Phi^*\Dr\varphi)(p)
= \rho(C_\gamma(p)\inv )\cdotact \Dr\varphi(p).
\end{aligned}
\end{align*}
Performing two gauge transformations is equivalent as performing a single gauge transformation of their product: $(\omega^\gamma)^{\gamma'}=\omega^{\gamma\gamma'}$. Note that the cocycle relation \eqref{cocycle relation} imposes that a gauge transformed $C(H)$-equivariant field satisfies
\begin{align*}
\varphi^\gamma(ph)
&= \rho(C(p\gamma(p),h)\inv)\cdotact\varphi^\gamma(ph) 
\rdefeq \rho(C^\gamma(p,h)\inv)\cdotact\varphi^\gamma(ph),
\end{align*}
where $C^\gamma:(p,h)\mapsto C(p\gamma(p),h)$ is a cocycle (called “gauge transformed cocycle”)\footnote{Mind the difference between a gauge transformed cocycle $C^\gamma$ and a cocycle evaluated at $(p,\gamma(p))$, written $C_\gamma$.}. This means that $\varphi^\gamma$ is a $C^\gamma(H)$-equivariant field. This change of equivariance appears to be of importance in the following, see Sec.~\ref{subsec Isom Changing Cocycle}.

Consider $\sigma:U\to\manP_U$ a (local) smooth section of $\manP$ and consider $\varphi^\manP\in\calF_\Ceq(\manP,V)$, $\omega^\manP\in\calA^C(\manP)$ and $\alpha^\manP\in\Omega_\Ctens^\bullet(\manP,V)$. These fields locally give rise to  $\phi=\sigma^*\varphi^\manP\in\calF(U,V)$, $\phyA \defeq \sigma^*\omega^\manP\in\Omega^1(U,\kg)$ and $\phya \defeq \sigma^*\alpha^\manP\in\Omega^\bullet(U,V)$. Similarly, one can define the local cocycle $C_\sigma=\sigma^*C:U\times H\rightarrow G,\,(x,h)\mapsto C(\sigma(x),h)$. This map is of great use when dealing with several trivializing sections. Let $\sigma$ and $\sigma'$ be two sections of $\manP$, respectively on $U$ and $U'$, two intersecting open subsets, and let $g$ be their gluing function: $\sigma'=\sigma g$. It is relevant to define $C_{\sigma,g}:U\cap U'\rightarrow G,\ x\mapsto C_\sigma(x,g(x))=C(\sigma(x),g(x))$. This map appears in the gluing relations of local twisted fields. Indeed, with the notations $\phi=\sigma^*\varphi^\manP,\ \phi'=\sigma^{\prime*}\varphi^\manP$, one get that for any $x\in U\cap U'$,
\begin{align}
\label{eq loc phi 2}
\phi'(x)
&= \rho(C_{\sigma,g}(x)\inv) \cdotact \phi(x).
\end{align}
With similar notations,
\begin{equation}
\label{eq loc A a 2}
\begin{aligned}
\phyA'_x &= C_{\sigma,g}(x)\inv \phyA_x C_{\sigma,g}(x)+C_{\sigma,g}(x)\inv\dr(C_{\sigma,g})_x,
%\\
&
\phya'_x &= \rho(C_{\sigma,g}(x)\inv)\cdotact \phya_x.
\end{aligned}
\end{equation}
Note that $\dr(C_{\sigma,g})$ is a differential with respect to both occurrences of $x$.

The cocycle relation satisfied by $C$ induces an expression for $C_{\sigma'}$ given $C_\sigma$:
\begin{align*}
C_{\sigma'}(x,h) 
&= C(\sigma'(x),h)
= C(\sigma(x)g(x),h)
= C(\sigma(x),g(x))\inv C(\sigma(x),g(x)h)
\\
&= C_\sigma(x,g(x))\inv C_\sigma(x,g(x)h)
= C_{\sigma,g}(x)\inv C_\sigma(x,g(x)h).
\end{align*}
In particular on $U\cap U'\cap U''\neq 0$,
\begin{align}\label{eq gluing relation C}
C_{\sigma',g'}=C_{\sigma,g}\inv C_{\sigma,gg'},
\end{align}
where $g'$ is a gluing function from $\sigma'$ to another section $\sigma'' = \sigma'g'$ defined on $U''$. This expression ensures the consistency of the gluing properties \eqref{eq loc phi 2} and \eqref{eq loc A a 2} at the intersection of three or more open sets.

These twisted passive gauge transformations correspond to twisted active gauge transformations. Given a local trivializing section $\sigma$ of $\manP$, the action of $\calH^\manP$ on twisted fields induce an action of $\calH_U$ on their local counterparts. This action involves the map $C_{\sigma,\locgamma}:x\mapsto C(\sigma(x),\locgamma(x))$. The local active gauge transformations of the previous local fields are
\begin{align*}
\begin{aligned}
\phi^\locgamma(x)
&= \rho(C_{\sigma,\locgamma}(x)\inv)\cdotact\phi(x),
\\
%&
\phyA^\locgamma_x
&= C_{\sigma,\locgamma}(x)\inv \phyA_x C_{\sigma,\locgamma}(x) + C_{\sigma,\locgamma}(x)\inv\dr(C_{\sigma,\locgamma})_x,
\\
%&
\phya^\locgamma_x
&= \rho(C_{\sigma,\locgamma}(x)\inv)\cdotact \phya_x.
\end{aligned}
\end{align*}
These transformations are formally identical to passive gauge transformations. The maps $C_\sigma$ and $C_{\sigma,\locgamma}$ undergo gauge transformations induced by their global counterpart. Let $\locgamma'=\sigma^*\gamma'$. Then
\begin{align*}
(C_\sigma)^{\locgamma'} \defeq (C^{\gamma'})_\sigma : (x,h) 
&\mapsto C(\sigma(x)\gamma'(\sigma(x)),h) = C(\sigma(x)\locgamma'(x),h),
\\
(C_{\sigma,\locgamma})^{\locgamma'} \defeq (C^{\gamma'})_{\sigma,\locgamma} : x 
&\mapsto C(\sigma(x)\gamma'(\sigma(x)),\locgamma(x)) = C(\sigma(x)\locgamma'(x),\locgamma(x)).
\end{align*}
From this transformation, the cocycle relation \eqref{cocycle relation} on $C$ implies $C_{\sigma,\locgamma\locgamma'}=C_{\sigma,\locgamma}(C_{\sigma,\locgamma'})^\locgamma$, which is equivalent to \eqref{eq gluing relation C}.

%%%%%%%%%%%%%%%%
\section{Isomorphisms with a Standard Structure}
\label{sec Isomorphism}
%%%%%%%%%%%%%%%%

%%%%%%%%%%%%%%%%
\subsection{Correspondence Space and Induced Bundle}
\label{subsec Definition manS manQ}
%%%%%%%%%%%%%%%%

The notations used in this section are defined in Sec.~\ref{subsec Twisted Gauge Fields}. In particular, $(\manP(M,H),G,C)$ is the minimal setting to define twisted gauge fields on $\manP$.

\subsubsection{The Correspondence Space $\manS$}

Let us define $\manS \defeq \manP\times G$. This space is endowed with a right action of $H\times G$ which reads\footnote{Note that a simpler action of $H\times G$ on $\manS$ can be defined as $\widehat{\calR}_{(h,g')}^\manS(p,g) \defeq (ph,gg')$. In the following, this simpler action proves to be useful to consider non-twisted (or trivially twisted) gauge fields such that $C(p,h)=e$.} 
\begin{align*}
\calR_{(h,g')}^\manS s=\calR_{(h,g')}^\manS(p,g) \defeq (ph,C(p,h)\inv gg').
\end{align*}

\begin{lemma}\label{lemma def S ppal bdl}
$\calR^\manS$ is a right action. It induces a projection onto $M$, seen as the space of $\calR^\manS$-equivalence classes: $\pi_\manS(p,g)=\pi_\manP(p)$. $\manS$ is a principal $H\times G$ bundle over $M$ for the right action $\calR^\manS$ with the same trivialization sections as $\manP$.
\end{lemma}
\begin{proof}
The action property of $\calR^\manS$ is obtained by straightforward computations. It relies on the cocycle property \eqref{cocycle relation} satisfied by $C$ and the commutation of left and right multiplication on $G$. Let us prove that $\calR^\manS$ is free and transitive. Let $(p,g)\in\manS$ and $(h,g')\in H\times G$ such that $\calR_{(h,g')}^\manS(p,g)=(p,g)$. This implies
\begin{align*}
%\left\{ 
\begin{aligned}
p &= ph,
%\\
&
g &= C(p,h)\inv gg'.
\end{aligned} 
%\right.
\end{align*}
Since the right action of $H$ on $\manP$ is free, the first equation implies $h=e$. Thus, $C(p,h)=C(p,e)=e$ and the second equation becomes $gg'=g$. Multiplying by $g\inv$ from the left, one gets $g'=e$. Then $\calR^\manS$ is free.

Consider now $(p_1,g_1)$ and $(p_2,g_2)$ two elements of the same fiber in $\manS$. In other words, $\pi_\manP(p_1) \rdefeq \pi_\manS((p_1,g_1))=\pi_\manS((p_2,g_2)) \defeq \pi_\manP(p_2)$. Since $\calR^\manP$ is transitive, there exists $h_{1,2}\in H$ such that $p_2=p_1h_{1,2}$. Let $g'\in G$ arbitrary
\begin{align*}
\calR_{(h_{1,2},g')}^\manS(p_1,g_1)
&= (p_1h_{1,2},C(p,h_{1,2})\inv g_1g')
= (p_2,C(p,h_{1,2})\inv g_1g')
\end{align*}
Let us define $g_{1,2} \defeq \bigl(C(p,h_{1,2})\inv g_1\bigr)\inv g_2$. Then
\begin{align*}
\calR_{(h_{1,2},g_{1,2})}^\manS (p_1,g_1)
&= (p_2,g_2).
\end{align*}
Thus $\calR^\manS$ is transitive, which concludes the proof.

The local trivializing sections of $\manS$ induced by those of $\manP$ are treated in detail in Sec.~\ref{subsec Local Sect S,P,Q}. This further section contains a parametrization of any local trivializing section $\sigma_\manS:U\to\manS_U$ as $\sigma_\manS(x)=(\sigma_\manP(x),f(\sigma_\manP(x))$ with $\sigma_\manP$ a section of $\manP$ and $f:\manP\to G$.
\end{proof}

\subsubsection{The Quotient Bundle $\manQ$}

Let us define $\manQ \defeq \manS/H=\manP\times_{C(H)}G$.
\begin{lemma}
$\manQ$ is a principal $G$-bundle over $M$.
\end{lemma}
\begin{proof}
$G$ acts\footnote{It should be clear from the context whether $\calR$ refers to the action of $H$ on $\manP$, the action of $G$ on $\manQ$ or the action of $H\times G$ on $\manS$. When the ambiguity persists, they will be labeled $\calR^\manP,\ \calR^\manQ,\ \calR^\manS$.} on $\manQ$ as $\calR_{g'}^\manQ q=qg'=[p,g]g' \defeq [p,gg']=[ph,C(p,h)\inv gg']$. This induces a projection onto equivalence classes w.r.t. $\calR^\manQ$ which reads $\pi_\manQ:\manQ\rightarrow M,\ q=[p,g]\mapsto \pi_\manP(p)$. $\pi_\manQ$ is well-defined since $(p,g)\mapsto(ph,C(p,h)\inv g)$ and $(p,g)\mapsto(p,gg')$ commute.
\end{proof}
$\manQ$ inherits its trivializing sections from those of $\manP$. It is equivalently seen as a twisted associated bundle to $\manP$ or as a quotient bundle of $\manS$ with corresponding projection $\pi_H:\manS\rightarrow\manQ$. Note that $\manQ$ depends on the cocycle. $\manS$ doesn't depend on $C$ as a manifold, but it depends on $C$ as a $H$-space. In the following, $C$ is fixed unless otherwise specified. The following diagram highlights the double structure of $\manS$. Each map $\pi_\bullet$ is a principal bundle map.
\begin{equation}
\label{diag:SPQ}
\begin{tikzcd}
 &\manS \arrow[ld, "/G"',"\pi_G"] \arrow[rd, "/C(H)","\pi_H"']\arrow[dd,"\pi_\manS"] & \\
\manP  \arrow[rd, "/C(H)"',"\pi_\manP"] & &\manQ  \arrow[ld, "/G","\pi_\manQ"']\\
 & M & 
\end{tikzcd}
\end{equation}

Since $\manS=\manP\times G$, there exist inclusion maps $i_g:\manP\rightarrow\manS,\ p\mapsto(p,g)$ and $i_p:G\rightarrow\manS,\ g\mapsto(p,g)$. In the following, their tangent linear maps are denoted $\Tg_pi_g(X_p^\manP) \rdefeq X_p^\manP\oplus 0_g$ and $\Tg_gi_p(X_g^G) \rdefeq 0_p\oplus X_g^G$. Any $X_{(p,g)}^\manS\in\Tg_{(p,g)}\manS=\Tg_p\manP\oplus\Tg_gG$ reads $X_{(p,g)}^\manS=X_p^\manP\oplus X_g^G=(X_p^\manP\oplus 0_g)+(0_p\oplus X_g^G)$. Among these vectors, any $X_s^\manS=X_p^\manP\oplus X_g^G\in\Ver_s\manS$ corresponds to an element $(\ka,\kb)\in \text{Lie}(H\times G)=\kh\oplus\kg$. Explicitly,
\begin{align*}
X_s^\manS
= (\ka,\kb)_{(p,g)}^\ver
&= \dert{(p\exp(t\ka),C(p,\exp(t\ka))\inv g\exp(t\kb))}
%\\
%&
= \dert{\calR_{(\exp(t\ka),\exp(t\kb))}^{\calS*}(p,g)}
\\
&= \ka_p^\ver\oplus\left( \Tg_e\Rr_g\circ(\dr_HC\inv )_{(p,e)}(\ka) + \kb_g \right)
\end{align*}
where $\Tg_e\Rr_g$ is the tangent linear map of the right multiplication on $G$ at $e$.

It is worth stating some useful properties of the bundles introduced in this part.

\begin{lemma}\label{lemma Technical Lemma}
The maps $i_e:\manP\to\manS$ and $\pi_H:\manS\to\manQ$ together with the right actions $\calR^\manP$, $\calR^\manS$ and $\calR^\manQ$ satisfy the following properties.
\begin{enumerate}
\item[(i)] The right action on $\manS$ induces the tangent linear map $\Tg_{(p,g)}\calR_{(h,g')}^\manS$ which satisfies
\begin{align*}
\Tg_{(p,g)}\calR_{(h,e)}^\manS(X_p^\manP\oplus X_g^G)
&= \Tg_p\calR_h^\manP(X_p^\manP)\oplus\left[ \Tg_g\Lr_{C(p,h)\inv}(X_g^G)+\Tg_{C(p,h)\inv}\Rr_g\circ\dr_\manP C\inv_{(p,h)}(X_p^\manP) \right],
\\
\Tg_{(p,g)}\calR_{(e,g')}^\manS(X_p^\manP\oplus X_g^G)
&= X_p^\manP\oplus \Tg_g\Rr_{g'}(X_g^G),
\end{align*}
where $\Tg_g\Lr_{g'}$ (resp. $\Tg_g\Rr_{g'}$) is the tangent linear map of $\Lr_{g'}g=g'g$ (resp. $\Rr_{g'}g=gg'$) the left (resp. right) multiplication in $G$ and $(\dr_\manP C^{-1})_{(p,h)}: \Tg_p\manP\rightarrow \Tg_{C(p,h)^{-1}}G$ is the differential of $(p,h)\mapsto C(p,h)^{-1}$ along $p$ at $(p,h)$.

\item[(ii)] The commutation relation of the tangent linear maps of $i_e:\manP\rightarrow\manS,\ p\mapsto(p,e)$ and $\calR^\manP$ is
\begin{equation}
\Tg_{ph}i_e\circ\Tg_p\calR_h^\manP(X_p^\manP)=\Tg_{(p,e)}\calR_{(h,C(p,h))}^\manS\circ\Tg_pi_e(X_p^\manP)+\left[ C(p,h)\inv\dr_\manP C_{(p,h)}(X_p^\manP) \right]_{(ph,e)}^\ver.
\end{equation}

\item[(iii)] The tangent map to $\pi_H$ has kernel
\begin{equation}
\ker \Tg_{(p,g)}\pi_H=\left\{\ka_p^\ver\oplus \Tg_e\Rr_g\circ\dr_HC\inv_{(p,e)}(\ka),\ \ka\in\kh\right\}.
\end{equation}

\item[(iv)] The vertical vectors on $\manQ$ correspond to the vectors tangent to $G$ in $\manS$, \textit{i.e.}
\begin{equation*}
\Tg_{(p,g)}\pi_H\left( 0_p\oplus \Tg_gG \right)=\Ver_{[p,g]}\manQ
\end{equation*}
and the isomorphism is given by $\Tg_{(p,g)}\pi_H(0_p\oplus \kb_g)=\kb_{[p,g]}^\ver,\ \forall \kb\in\kg$.\footnote{Note that in the left hand side, $\kb\in \kg$ is interpreted as a left-invariant vector field valued at $g$, while in the right hand side, $\kb\in \kg\simeq\Tg_eG$ is interpreted as a tangent vector at $e\in G$.}

\item[(v)] The commutation relation of the tangent linear maps of $\pi_H$ and $\calR^\manS$ is
\begin{equation*}
\Tg_{(ph,C(p,h)\inv gg')}\pi_H\circ\Tg_{(p,g)}\calR_{(h,g')}^\manS=\Tg_{[p,g]}\calR_{g'}^\manQ\circ\Tg_{(p,g)}\pi_H,
\end{equation*}
\textit{i.e.} $\Tg\pi_H\circ\Tg\calR_{(h,e)}^\manS=\Tg\pi_H$ and $\Tg\pi_H\circ\Tg\calR_{(e,g')}^\manS=\Tg\calR_{g'}^\manQ\circ\Tg\pi_H$.

\item[(vi)] The commutation relation of the tangent maps of $\pi_H\circ i_e$ and $\calR^\manS$ is
%\begin{multline}
\begin{align*}
\Tg_{(ph,e)}\pi_H\circ\Tg_{ph}i_e\circ\Tg_p\calR_h^\manP(X_p^\manP)
=
\Tg_{[p,e]}\calR_{C(p,h)}^\manQ\circ\Tg_{(p,e)}\pi_H\circ\Tg_pi_e(X_p^\manP)%\vspace{0.1cm}
\\
+ \left[ C(p,h)\inv\dr_\manP C_{(p,h)}(X_p^\manP) \right]_{[ph,e]}^\ver,
%\end{multline}
\end{align*}
or in short, $\Tg(\pi_H\circ i_e)\circ\Tg_p\calR_h^\manP=\Tg\calR_{C(p,h)}^\manQ\circ\Tg_p(\pi_H\circ i_e)+\left[ C(p,h)\inv\dr_\manP C_{(p,h)}(\bullet) \right]_{[ph,e]}^\ver$.
\end{enumerate}

\end{lemma}

The proof of this technical Lemma is given in Appendix~\ref{sec proof technical lemma}.

The representation $\rho_G$ of $G$ on $V$ induces a representation $\rho_{H\times G}$ such that $\rho_{H\times G}((h,g))\cdotact v=\rho_G(g)\cdotact v$. The induced representation $\rho_{H\times G}$ is trivial along $H$. Let us define the following two associated vector bundles $\manQ\times_GV$ and $\manS\times_{H\times G} V$ respectively through the equivalence relations $(q,v)\sim(qg,\rho_G(g\inv)\cdotact v)$ on $\manQ\times V$ and $(s,v)\sim(\calR_{(h,g)}s,\rho_G(g\inv)\cdotact v)$ on $\manS\times V$.

The notion of tensorial forms extends naturally to $\manS$. A form $\alpha\in\Omega^\bullet(\manS,V)$ is called equivariant if
\begin{align*}
\calR_{(h,g)}\alpha
&= \rho_{H\times G}((h,g)\inv)\cdotact\alpha
= \rho_G(g\inv)\cdotact\alpha 
\rdefeq g\inv\cdotact\alpha,\ \forall (h,g)\in H\times G.
\end{align*}
$\alpha$ is horizontal if it vanishes when applied on at least one vector in $\Ver_s\manS$. The resulting space is denoted $\Omega_\tens^\bullet(\manS,V)$, where “$\tens$” encapsulates all together the $G$-equivariance, $H$-invariance and horizontality in this case.

Following the previous diagram \eqref{diag:SPQ}, one can interpret $\manS$ as either a (trivial) $G$-principal bundle over $\manP$ or a $H$-principal bundle over $\manQ$. Consequently, any map or form defined on $\manP$ (resp. $\manQ$) corresponds to a $G$-basic (resp. $H$-basic) map or form on $\manS$ through $\pi_G^*$ (resp. $\pi_H^*$). In particular, the horizontal structures and gauge structures of both $\manP$ and $\manQ$ have their counterparts on $\manS$. Examples of these corresponding maps and forms are detailed in Sec.~\ref{subsec Isom Bdl Geom} and \ref{sec Connections GT on S}.

%%%%%%%%%%%%%%%%
\subsection{Isomorphism of Bundle Geometry}
\label{subsec Isom Bdl Geom}
%%%%%%%%%%%%%%%%

This section carries one of the main results of this paper, which relates the tensorial and $C$-tensorial forms on the previously defined bundles.

Let $(V,\rho_G)$ be a representation of $G$. As before, $\rho_G$ is omitted in the notation.

\begin{proposition}\label{prop isom associated bdls}
The twisted and standard associated bundles $\manP\times_{C(H)}V$ and $\manQ\times_G V$ are both isomorphic to $\manS\times_{H\times G}V$.
\end{proposition}

\begin{proof}
$\calR^\manS$ induces an equivalence relation on $\manP\times G\times V=\manS\times V$:
\begin{align*}
	(0) &:  (p,g_0,v)\sim(ph,C(p,h)\inv g_0g,g^{-1} \cdotact v),\,\forall (h,g)\in H\times G.
\end{align*}
The corresponding space of equivalence classes is the vector bundle $\manS\times_{H\times G}V$ associated to $\manS$. $\calR^\manS$ is made of two commuting right actions of $H$ and $G$ respectively, which give rise to the two following equivalence relations on $\manP\times G\times V$
\begin{align*}
	(1) &: (p,g_0,v)\sim(p,g_0g,g^{-1} \cdotact v)\sim(p,e,g_0 \cdotact v),\,\forall g \in G,\\
	(2) &: (p,g_0,v)\sim(ph,C(p,h)\inv g_0,v),\,\forall h \in H.
\end{align*}
One can quotient along these equivalence relations and define new equivalence relations respectively on the quotient spaces $\manP\times V$ and $\manQ\times V$:
\begin{align*}
	(3) &: (p,g_0\inv \cdotact v)\sim(ph,C(p,h)\inv g_0\inv \cdotact v),\,\forall h \in H,\\
	(4) &: ([p,g_0],v)\sim([p,g_0]g,g^{-1} \cdotact v),\,\forall g \in G.
\end{align*}
The situation is summarized in the following diagram:
\begin{equation}
\label{diag:SPQxV}
\begin{tikzcd}
 &\manS\times V \arrow[ld, "/G"',"(1)"] \arrow[rd, "/C(H)","(2)"']\arrow[dd,"(0)"] & \\
\manP \times V \arrow[rd, "/C(H)"',"(3)"] & &\manQ \times V \arrow[ld, "/G","(4)"']\\
 &\begin{matrix}
 \manS\times_{H\times G}V\\ \manP\times_{C(H)}V\\ \manQ\times_GV
\end{matrix} & 
\end{tikzcd}
\end{equation}
Since the two actions commute, the same equivalence space is obtained by quotienting by $H$ and $G$ in any order. The left branch of diagram \eqref{diag:SPQxV} interprets this quotient space as $\manP\times_{C(H)}V$ while the right branch interprets it as $\manQ\times_G V$, thus proving the isomorphism between the three quotient spaces.
\end{proof}

This correspondence between twisted and standard associated bundles extends to vector-valued tensorial forms.

\begin{proposition}\label{prop isom tensorial forms}
The spaces $\Omega_\Ctens^\bullet(\manP,V)$ and $\Omega_\tens^\bullet(\manQ,V)$ are isomorphic.
\end{proposition}

The proof of this isomorphism relies on the construction of three maps $\chi_{\manS\manP}$, $\chi_{\manQ\manS}$ and $\chi_{\manP\manQ}$ which will now be described. In the following, $\{X_p^\manP\}$ (resp. $\{X_s^\manS\}=\{X_p^\manP\oplus X_g^G\}$, $\{X_q^\manQ\}$) stands for an ordered collection of vector fields in $\Gamma(\Tg\manP)$ (resp. $\Gamma(\Tg\manS)$, $\Gamma(\Tg\manQ)$) evaluated at $p$ (resp. $s=(p,g)$, $q=[p,g]$).

\begin{lemma}\label{lemma chi_SP}
Let us define $\chi_{\manS\manP}$ as
\begin{align*}
\forall \alpha^\manP\in\Omega_\Ctens^\bullet(\manP,V),
\ \chi_{\manS\manP}(\alpha^\manP)_{(p,g)}(\{X_p^\manP\oplus X_g^G\}) 
\defeq g\inv\cdotact\alpha_p^\manP(\{X_p^\manP\}).
\end{align*}
Then $\chi_{\manS\manP}:\Omega_\Ctens^\bullet(\manP,V)\rightarrow\Omega_\tens^\bullet(\manS,V)$.
\end{lemma}
\begin{proof}
One checks that $\chi_{\manS\manP}$ is well defined. From the algebraic definition together with the $H$-equivariance of $\alpha^\manP$, one gets that $\chi_{\manS\manP}(\alpha^\manP)$ is $H$-equivariant. Also, the action of $g\inv$ provides the $G$-equivariance of $\chi_{\manS\manP}(\alpha^\manP)$, which completes the $H\times G$-equivariance. Additionally, when evaluated on at least one vertical tangent vector $(\ka,\kb)_{(p,g)}^\ver$, with $(\ka,\kb)\in\kh\oplus\kg$, $\chi_{\manS\manP}(\alpha^\manP)$ vanishes:
\begin{align*}
\chi_{\manS\manP}(\alpha^\manP)_{(p,g)}((\ka,\kb)_{(p,g)}^\ver,\ldots)=
g\inv\cdotact\alpha_p^\manP(\ka_p^\ver,\ldots)=0
\end{align*}
because $\alpha^\manP$ is horizontal. Then $\chi_{\manS\manP}(\alpha^\manP)\in\Omega_\tens^\bullet(\manS,V)$.
\end{proof}

\begin{lemma}\label{lemma chi_QS}
The map $\chi_{\manQ\manS}$ is implicitly defined by
\begin{align*}
\forall \alpha^\manS\in\Omega_\tens^\bullet(\manS,V),
\ \chi_{\manQ\manS}(\alpha^\manS)_{[p,g]}\circ\Tg_{(p,g)}\pi_H 
\defeq \alpha^\manS_{(p,g)}.
\end{align*}
Then $\chi_{\manQ\manS}:\Omega_\tens^\bullet(\manS,V)\rightarrow\Omega_\tens^\bullet(\manQ,V)$.
\end{lemma}
\begin{proof}
Let us check that $\chi_{\manQ\manS}$ is well defined. $\chi_{\manQ\manS}(\alpha^\manS)$ does not depend on the choice of representative $(p,g)\in [p,g]$ because $\alpha^\manS$ is $H$-invariant. Also, $\ker\Tg_{(p,g)}\pi_H=\{\ka_p^\ver\oplus \Tg_e\Rr_g\circ(\dr_HC\inv )_{(p,e)}(\ka),\ \ka\in\kh\}$ since these are the tangent vectors to curves which read $t\mapsto(ph(t),C(p,h(t))\inv g)$ with $h(0)=e$ and $\dot{h}(0)=\ka$. Since $\ker\Tg_{(p,g)}\pi_H\subset\Ver_{(p,g)}\manS$, the choice of a reciprocal map to $\Tg_{(p,g)}\pi_H$ is defined up to a vertical vector. Because $\alpha^\manS$ is horizontal, this does not affect the image of $\chi_{\manQ\manS}(\alpha^\manS)$. Hence $\chi_{\manQ\manS}$ is defined unambiguously.

The $G$-equivariance of $\chi_{\manQ\manS}(\alpha^\manS)$ is directly inherited from the one of $\alpha^\manS$. Let $\kb\in\kg$ and $\kb_q^\ver$ be its vertical tangent vector in $\Ver_q\manQ$ at $q=[p,g]$.
\begin{align*}
\chi_{\manQ\manS}(\alpha^\manS)_q(\kb_q^\ver,\ldots)
&= \alpha^\manS_{(p,g)}(\ka_p^\ver \oplus \left( \Tg_e \Rr_g\circ(\dr_H C\inv )_{(p,e)}(\ka) + \kb_q^\ver \right),\ldots)=0,
\end{align*}
where $\ka\in\kh$ induces an arbitrary element in the preimage of $\kb_q^\ver$ through $\Tg_{(p,g)}\pi_H$. Since $\alpha^\manS$ is evaluated on a vertical vector, it vanishes. Then $\chi_{\manQ\manS}(\alpha^\manS)$ is vertical, thus $\chi_{\manQ\manS}(\alpha^\manS)\in\Omega_\tens^\bullet(\manQ,V)$.
\end{proof}

\begin{lemma}\label{lemma chi_PQ}
The map $\chi_{\manP\manQ}$ is defined by
\begin{align*}
\forall \alpha^\manQ\in\Omega_\tens^\bullet(\manQ,V),
\ \chi_{\manP\manQ}(\alpha^\manQ)_p(\{X_p^\manP\}) 
&\defeq \alpha_{[p,e]}^\manQ\circ\Tg_{(p,e)}\pi_H(\{X_p^\manP\oplus 0\}).
\end{align*}
Then $\chi_{\manP\manQ}:\Omega_\tens^\bullet(\manQ,V)\rightarrow\Omega_\Ctens^\bullet(\manP,V)$. Additionally, $\chi_{\manP\manQ}(\alpha^\manQ)=i_e^*\pi_H^*\alpha^\manQ=(\pi_H\circ i_e)^*\alpha^\manQ$.
\end{lemma}
\begin{proof}
Let us check that $\chi_{\manP\manQ}$ is well defined. First, compute the equivariance relation satisfied by $\chi_{\manP\manQ}(\alpha^\manQ)$.
\begin{align*}
(\calR_h^{\manP*}\chi_{\manP\manQ}(\alpha^\manQ))_p(\{X_p^\manP\})
&= (\calR_h^{\manP*}i_e^*\pi_H^*\alpha^\manQ)_p(\{X_p^\manP\})
\\
&= \alpha^\manQ_{[ph,e]}\circ\Tg_{(ph,e)}\pi_H\circ\Tg_{ph}i_e\circ\Tg_p\calR_h^\manP(\{X_p^\manP\})
\\
&\overset{\mathclap{(vi)}}{=} 
\begin{multlined}[t]
\alpha^\manQ_{[p,C(p,h)]}\circ\Tg_{[p,e]}\calR_{C(p,h)}^\manQ\circ\Tg_{(p,e)}\pi_H\circ\Tg_pi_e(\{X_p^\manP\})
\\
+\alpha_{[p,C(p,h)]}^\manQ\left(\left\{\left[ C(p,h)\inv\dr_\manP C_{(p,h)}(X_p^\manP) \right]_{[ph,e]}^\ver\right\}\right),
\end{multlined}
\end{align*}
where Lemma~\ref{lemma Technical Lemma}\emph{(vi)} was used. The second term vanishes because $\alpha^\manQ$ is horizontal. By equivariance, the first term leads to
\begin{align*}
(\calR_h^{\manQ*}\chi_{\manP\manQ}(\alpha^\manQ))_p(\{X_p^\manP\})
&= (\calR_{C(p,h)\inv}^{\manQ*}\alpha^\manQ)_{[p,e]}\circ\Tg_{(p,e)}\pi_H\circ\Tg_pi_e(\{X_p^\manP\})
\\
&= C(p,h)\inv\cdotact\alpha^\manQ_{[p,e]}\circ\Tg_{(p,e)}\pi_H\circ\Tg_pi_e(\{X_p^\manP\})
\\
&= C(p,h)\inv\cdotact (i_e^*\pi_H^*\alpha^\manQ)_p(\{X_p^\manP\})
\\
&= C(p,h)\inv\cdotact \chi_{\manP\manQ}(\alpha^\manQ)_p(\{X_p^\manP\}).
\end{align*}
Thus $\chi_{\manP\manQ}(\alpha^\manQ)$ is $C(H)$-equivariant.

Let $\ka\in\kh$ and $\ka_p^\ver$ be its vertical tangent vector in $\Ver_p\manP$ at $p$. Note that $\ka_p^\ver\oplus(\dr_GC\inv)_{(p,e)}(\ka)\in\ker\Tg_{(p,e)}\pi_H$, which implies that $\Tg_{(p,e)}\pi_H(\ka_p^\ver\oplus 0)=\Tg_{(p,e)}\pi_H(0\oplus -(\dr_GC\inv)_{(p,e)}(\ka))$. This enables computing
\begin{align*}
\chi_{\manP\manQ}(\alpha^\manQ)_p(\ka_p^\ver,\ldots)
&= \alpha^\manQ_{[p,e]}\circ \Tg_{(p,e)}\pi_H (\ka_p^\ver\oplus 0,\ldots)
\\
&= \alpha^\manQ_{[p,e]}\circ \Tg_{(p,e)}\pi_H (0\oplus -(\dr_GC\inv )_{(p,e)}(\ka),\ldots)
\\
&= \alpha^\manQ_{[p,e]}\left( \dert{[p, \exp(-t(\dr_GC\inv )_{(p,e)}(\ka))]},\ldots\right)
\\
&= \alpha^\manQ_{[p,e]}\left( \dert{[p,e] \exp(-t(\dr_GC\inv )_{(p,e)}(\ka))},\ldots\right)
\\
&= \alpha^\manQ_{[p,e]}\left( [-(\dr_GC\inv )_{(p,e)}(\ka)]_{[p,e]}^\ver,\ldots\right)=0.
\end{align*}
Then $\chi_{\manP\manQ}(\alpha^\manQ)$ inherits the horizontality of $\alpha^\manQ$. Thus, $\chi_{\manP\manQ}(\alpha^\manQ)\in\Omega_\Ctens^\bullet(\manP,V)$.
\end{proof}

\begin{proof}[Prop.~\ref{prop isom tensorial forms}]
This result relies on the following diagram of maps:
\begin{equation*}
\begin{tikzcd}
 &\Omega_\tens^\bullet(\manS,V) \arrow[rd, "\chi_{\manQ\manS}"] & \\
\Omega_{C\text{-tens}}^\bullet(\manP,V) \arrow[ru, "\chi_{\manS\manP}"] & &\Omega_\tens^\bullet(\manQ,V) \arrow[ll, "\chi_{\manP\manQ}"']
\end{tikzcd}
\end{equation*}
One has to check that these maps compose to the identity map on each space.

Let $\chi_{\manP\manP} \defeq \chi_{\manP\manQ}\circ\chi_{\manQ\manS}\circ\chi_{\manS\manP}:\Omega_\Ctens^\bullet(\manP,V)\rightarrow\Omega_\Ctens^\bullet(\manP,V)$, $\alpha^\manP\in\Omega_\Ctens^k(\manP,V)$ and $\{X^\manP\}\in\Gamma(\Tg\manP)^k$. One has
\begin{align*}
\chi_{\manP\manP}(\alpha^\manP)_p(\{X_p^\manP\})
&= \chi_{\manQ\manS}\circ\chi_{\manS\manP}(\alpha^\manP)_{[p,e]}\circ\Tg_{(p,e)}\pi_H({X_p^\manP\oplus 0})
%\\
%&
= \chi_{\manS\manP}(\alpha^\manP)_{(p,e)}({X_p^\manP\oplus 0})
\\
&
= e\inv\cdotact\alpha^\manP_p({X_p^\manP}).
\end{align*}
So $\chi_{\manP\manP}$ is the identity map.

Let $\chi_{\manQ\manQ} \defeq \chi_{\manQ\manS}\circ\chi_{\manS\manP}\circ\chi_{\manP\manQ}:\Omega_\tens^\bullet(\manQ,V)\rightarrow\Omega_\tens^\bullet(\manQ,V)$, $\alpha^\manQ\in\Omega_\tens^k(\manQ,V)$ and $\{X^\manQ\}\in\Gamma(\Tg\manQ)^k$. For each $X^\manQ$, let $X^\manP\oplus X^G$ be a preimage of $X^\manQ$ through $\Tg\pi_H$: $X_{[p,g]}^\manQ=\Tg_{(p,g)}\pi_H(X_p^\manP\oplus X_g^G)$. Hence,
\begin{align*}
\chi_{\manQ\manQ}(\alpha^\manQ)_{[p,g]}(\{X_{[p,g]}^\manQ\})
&= \chi_{\manQ\manS}\circ\chi_{\manS\manP}\circ\chi_{\manP\manQ}(\alpha^\manQ)_{[p,g]}\circ\Tg_{(p,g)}\pi_H(\{X_p^\manP\oplus X_g^G\})
\\
&
= \chi_{\manS\manP}\circ\chi_{\manP\manQ}(\alpha^\manQ)_{(p,g)}(\{X_p^\manP\oplus X_g^G\})
\\
&= g\inv\cdotact\chi_{\manP\manQ}(\alpha^\manQ)_p(\{X_p^\manP\})
%\\
%&
= g\inv\cdotact\alpha^\manQ_{[p,e]}\circ\Tg_{(p,e)}\pi_H(\{X_p^\manP\oplus 0_e\})
\\
&= \alpha^\manQ_{[p,g]}\circ\Tg_{[p,e]}\calR_g^\manQ\circ\Tg_{(p,e)}\pi_H(\{X_p^\manP\oplus 0_e\})
\\
&
= \alpha^\manQ_{[p,g]}\circ\Tg_{(p,g)}\pi_H\circ\Tg_{(p,e)}\calR_{(e,g)}^\manS(\{X_p^\manP\oplus 0_e\})
\\
&= \alpha^\manQ_{[p,g]}\circ\Tg_{(p,g)}\pi_H(\{X_p^\manP\oplus 0_g\})
\\
&= \alpha^\manQ_{[p,g]}\circ\Tg_{(p,g)}\pi_H(\{X_p^\manP\oplus X_g^G\})- \textstyle\sum\alpha^\manQ_{[p,g]}\circ\Tg_{(p,g)}\pi_H(\ldots,0_p\oplus X_g^G,\ldots)
\end{align*}
The first term is $\alpha_{[p,g]}^\manQ(\{X_{[p,g]}^\manQ\}$. The remaining sum carries the $2^k-1$ terms obtained by developing $\alpha^\manQ_{[p,g]}\circ\Tg_{(p,g)}\pi_H(\{X_p^\manP\oplus X_g^G\})$ with respect to each argument, except the term evaluated at $X_p^\manP\oplus 0_g$ on each argument. According to Lemma~\ref{lemma Technical Lemma}\emph{(iv)}, $\Tg_{(p,g)}\pi_H(0_p\oplus X_g^G)\in\Ver_{[p,g]}\manQ$. Since $\alpha^\manQ$ is horizontal, every term in the sum cancels. This proves $\chi_{\manQ\manQ}$ is the identity map.

Let $\chi_{\manS\manS} \defeq \chi_{\manS\manP}\circ\chi_{\manP\manQ}\circ\chi_{\manQ\manS}:\Omega_\tens^\bullet(\manS,V)\rightarrow\Omega_\tens^\bullet(\manS,V)$, $\alpha^\manS\in\Omega_\tens^k(\manS,V)$ and $\{X^\manS\}=\{X_p^\manP\oplus X_g^G\}\in\Gamma(\Tg\manS)^k$. Similarly as the previous case, since $\chi_{\manS\manS}(\alpha^\manS)$ is horizontal and $0_p\oplus X_g^G\in\Ver_{(p,g)}\manS$, one has $\chi_{\manS\manS}(\alpha^\manS)_{(p,g)}(...,0_p\oplus X_g^G,...)=0$. One is thus led to compute
\begin{align*}
\chi_{\manS\manS}(\alpha^\manS)_{(p,g)}(\{X_{(p,g)}^\manS\})
&= \chi_{\manS\manS}(\alpha^\manS)_{(p,g)}(\{X_p^\manP\oplus 0_g\})
%\\
%&
= g\inv\cdotact\chi_{\manP\manQ}\circ\chi_{\manQ\manS}(\alpha^\manS)_p(\{X_p^\manP\})
\\
&= g\inv\cdotact\chi_{\manQ\manS}(\alpha^\manS)_{[p,e]}\circ\Tg_{(p,e)}\pi_H(\{X_p^\manP\oplus 0_e\})
\\
&= \chi_{\manQ\manS}(\alpha^\manS)_{[p,g]}\circ\Tg_{[p,e]}\calR_g^\manQ\circ\Tg_{(p,e)}\pi_H(\{X_p^\manP\oplus 0_e\})
\\
&
= \chi_{\manQ\manS}(\alpha^\manS)_{[p,g]}\circ\Tg_{(p,g)}\pi_H(\{X_p^\manP\oplus 0_g\})
\\
&
= \alpha^\manS_{(p,g)}(\{X_p^\manP\oplus 0_g\})
= \alpha^\manS_{(p,g)}(\{X_{(p,g)}^\manS\}).
\end{align*}
So $\chi_{\manS\manS}$ is the identity map. This proves that the three spaces are isomorphic.
\end{proof}

A specific case of this result holds for equivariant functions, which are tensorial 0-forms.

\begin{corollary}\label{corollary isom eq functions}
The spaces $\calF_\Ceq(\manP,V)$, $\calF_\eqv(\manQ,V)$ and $\calF_\eqv(\manS,V)$ are isomorphic.
\end{corollary}

\begin{remark}
Another proof of this result consists in using the correspondence between equivariant $V$-valued functions and sections of associated vector bundles.
\begin{align*}
\begin{aligned}
\calF_\eqv(\manQ,V)
&\simeq \Gamma(\manQ\times_G V),
&
\varphi(qg) = g\inv\cdotact\varphi(q)
&\leftrightarrow \phi(x) = [q,\varphi(q)] = [qg,g\inv\cdotact\varphi(q)],\ q\in\pi_\manQ\inv(x).
\end{aligned}
\end{align*}
Since Prop.~\ref{prop isom associated bdls} highlights bundle isomorphisms between $\calP\times_{C(H)}V$, $\calQ\times_G V$ and $\calS\times_{C(H)\times G}V$, their spaces of sections are isomorphic, which concludes the proof.
\end{remark}

\begin{remark}\label{remark forms valued in associated bdls}
The commutative diagram used in the proof of Prop.~\ref{prop isom tensorial forms} can be added a fourth space to mimic the shape of the previous ones, see below. Recall that $\Omega_\tens^\bullet(\manQ,V)\simeq\Omega^\bullet(M,\manQ\times_HV)$ with the isomorphism $\alpha^\manQ\in\Omega_\tens^\bullet(\manQ,V)\leftrightarrow \phyA^\manQ\in\Omega^\bullet(M,\manQ\times_HV)$ given by $\phyA_x^\manQ(\{X_x\})=[q,\alpha_q^\manQ(\{X_q^\hor\})]$ which does not depend on the choice of $q\in\pi_\manQ\inv(x)$ (because $\alpha^\manQ$ is equivariant) or the choice of horizontal lift (because $\alpha^\manQ$ is horizontal). The same holds for $\Omega_\tens^\bullet(\manS,V)\simeq\Omega^\bullet(M,\manS\times_{H\times G}V)$. The isomorphism in Prop.~\ref{prop isom associated bdls} induces isomorphisms $\Omega^\bullet(M,\manQ\times_HV)\simeq\Omega^\bullet(M,\manS\times_{H\times G}V)\simeq\Omega^\bullet(M,\manP\times_{C(H)}V)$. Together with Prop.~\ref{prop isom tensorial forms}, this provides a correspondence between twisted tensorial forms $\alpha^\manP\in\Omega_\Ctens^\bullet(\manP,V)$ and forms on $M$ valued in a twisted associated bundle $\phyA^\manP\in\Omega^\bullet(M,\manP\times_{C(H)}V)$. The correspondence is also provided by $\phyA_x^\manP(\{X_x\})=[p,\alpha_p^\manP(\{X_p^\hor\})]$.
\begin{equation*}
\begin{tikzcd}
 &\Omega_\tens^\bullet(\manS,V) \arrow[rd, "\chi_{\manQ\manS}"] & \\
\Omega_{C\text{-tens}}^\bullet(\manP,V) \arrow[ru, "\chi_{\manS\manP}"]\arrow[rd, "\simeq"'] & &\Omega_\tens^\bullet(\manQ,V) \arrow[ll, "\chi_{\manP\manQ}"']\arrow[ld, "\simeq"']\\
 & \begin{matrix}
 \Omega^\bullet(M,\manS\times_{H\times G}V)\\ \simeq\Omega^\bullet(M,\manP\times_{C(H)}V)\\ \simeq\Omega^\bullet(M,\manQ\times_HV)
\end{matrix}  & 
\end{tikzcd}
\end{equation*}
This completes our study.
\end{remark}

%%%%%%%%%%%%%%%%
\subsection{\texorpdfstring{Twisted Connections on $\manP$ as Induced Structures}{Twisted Connections on P as Induced Structures}}
\label{subsec Isom Connection PQ}
%%%%%%%%%%%%%%%%

In this section, the correspondence between the geometry of the bundle $\manQ$ and its twisted version on top of $\manP$ is extended to a correspondence of their horizontal structures. This correspondence is better established at the level of connection 1-forms, as well as gauge transformations.

Let us extend the definition of the map $\chi_{\manP\manQ}$, introduced in Lemma~\ref{lemma chi_PQ}, to the space of $\kg$-valued 1-forms:
\begin{align*}
	\begin{array}{c}
		\tilde{\chi}_{\manP\manQ}:\\
		 \ 
	\end	{array}
	\left|
	\begin{array}{l}
		\Omega^1(\manQ,\kg)\rightarrow \Omega^1(\manP,\kg)\\
		\omega^\manQ\mapsto\tilde{\chi}_{\manP\manQ}(\omega^\manQ):(p,X_p^\manP)\mapsto \omega^\manQ_{[p,e]}\circ\Tg_{(p,e)}\pi_H(X_p^\manP\oplus 0)
	\end{array}
	\right.
\end{align*}
Note that $\tilde{\chi}_{\manP\manQ}(\omega^\manQ)=i_e ^*\pi_H^*(\omega^\manQ)$, which is similar to the definition of $\chi_{\manP\manQ}$ on tensorial forms.

\begin{proposition}\label{prop isom connections PQ}
Let $\omega^\manQ\in\calA(\manQ)$ be a (standard) connection 1-form on $\manQ$. Then $\tilde{\chi}_{\manP\manQ}(\omega^\manQ)$ is a twisted connection on $\manP$.

Reciprocally, any twisted connection on $\manP$ (with cocycle $C$) is induced by a standard connection on $\manQ$ through $\tilde{\chi}_{\manP\manQ}$.
\end{proposition}

\begin{proof}
\textbf{$\tilde{\chi}_{\manP\manQ}(\omega^\manQ)$ is a twisted connection on $\manP$.}

The equivariance of $\tilde{\chi}_{\manP\manQ}(\omega^\manQ)$ is obtained with the equivariance and vertical normalization of $\omega^\manQ$ together with Lemma~\ref{lemma Technical Lemma}\emph{(vi)}:
\begin{align*}
\begin{aligned}
\calR_h^{\manP*}\tilde{\chi}_{\manP\manQ}(\omega^\manQ)(X_p)
&= \calR_h^{\manP*}i_e^*\pi_H^*\omega^\manQ(X_p)
%\\
%&
= \omega^\manQ_{[ph,e]}\circ\Tg_{ph}(\pi_H\circ i_e)\circ\Tg_p\calR_h^\manP(X_p^\manP)&
\\
&= 
%\begin{multlined}[t]
\omega^\manQ_{[p h,e]}\circ\Tg_{[p,e]}\calR_{C(p,h)}^\manQ\circ\Tg_p(\pi_H\circ i_e)(X_p^\manP)&
\\
&\hspace{2.4cm}+\omega^\manQ_{[p h,e]}\left( \left[ C(p,h)\inv\dr_\manP C_{(p,h)}(X_p^\manP) \right]_{[p h,e]}^\ver \right)&
%\end{multlined}
\\
&= \Ad_{C(p,h)\inv}\omega^\manQ_{[p,e]}\circ\Tg_p(\pi_H\circ i_e)(X_p^\manP)&\\
&\hspace{4.76cm}+ C(p,h)\inv\dr_\manP C_{(p,h)}(X_p^\manP)&
\\
&= \Ad_{C(p,h)\inv}\tilde{\chi}_{\manP\manQ}(\omega^\manQ)_p(X_p^\manP) + C(p,h)\inv\dr_\manP C_{(p,h)}(X_p^\manP)&
\end{aligned}
\end{align*}

The vertical normalization of $\tilde{\chi}_{\manP\manQ}(\omega^\manQ)$ is obtained from the vertical normalization of $\omega^\manQ$ together with Lemmas~\ref{lemma d_H C}, \ref{lemma Technical Lemma}\emph{(iii)} and \emph{(iv)}. Let $\ka\in\kh$ and $\ka^\ver$ the corresponding fundamental vertical vector field.
\begin{align*}
\tilde{\chi}_{\manP\manQ}(\omega^\manQ)_p(\ka_p^\ver)
&= \omega^\manQ_{[p,e]}\circ\Tg_{(p,e)}\pi_H(\ka_p^\ver\oplus 0_e)
%\\
%&
= \omega^\manQ_{[p,e]}\circ\Tg_{(p,e)}\pi_H(0_p\oplus -\dr_HC\inv_{(p,e)}(\ka))
\\
&= \omega^\manQ_{[p,e]}\circ\Tg_{(p,e)}\pi_H(0_p\oplus \dr_HC_{(p,e)}(\ka))
%\\
%&
= \omega^\manQ_{[p,e]}\left(\left[\dr_HC_{(p,e)}(\ka)\right]_{[p,e]}^\ver\right)
\\
&= \dr_HC_{(p,e)}(\ka).
\end{align*}
This proves that $\tilde{\chi}_{\manP\manQ}(\omega^\manQ)$ is a twisted connection on $\manP$.

\medskip
\textbf{$\omega^\manP$ is the image of a unique connection $\omega^\manQ$.}

Recall that the spaces of standard (resp. twisted) connection 1-forms on $\manQ$ (resp. $\manP$) are affine spaces modeled on $\Omega_\tens^1(\manQ,\kg)$ (resp. $\Omega_\Ctens^1(\manP,\kg)$). In the proof of Prop.~\ref{prop isom tensorial forms}, these associated vector spaces are shown to be isomorphic with isomorphism $\tilde{\chi}_{\manP\manQ}$.

Let $\omega_0^\manQ\in\calA(\manQ)$ be a standard connection and let $\omega_0^\manP \defeq \tilde{\chi}_{\manP\manQ}(\omega_0^\manQ)$. Let $\omega^\manP$ be any twisted connection. Then
\begin{align*}
\omega^\manP-\omega_0^\manP{} 
&\rdefeq \alpha^\manP\in\Omega_\Ctens^1(\manP,\kg)
\\
&
= \chi_{\manP\manQ}\left( \chi_{\manP\manQ}\inv(\alpha^\manP) \right)
\end{align*}
where $\tilde{\chi}_{\manP\manQ}\inv(\alpha^\manP)\in\Omega_\tens^1(\manQ,V)$. Then
\begin{align*}
\omega^\manP
&= \omega_0^\manP + \alpha^\manP
= \tilde{\chi}_{\manP\manQ}\left(\omega_0^\manQ + \chi_{\manP\manQ}\inv(\alpha^\manP)\right)
\end{align*}
where $\omega_0^\manQ+\chi_{\manP\manQ}\inv(\alpha^\manP)$ is a standard connection. Since $\chi_{\manP\manQ}$ is invertible on $\Omega_\Ctens^1(\manP,\kg)$, the antecedent is unique. This proves that $\omega^\manP$ is induced by a unique standard connection.
\end{proof}

\begin{remark}
This result provides a geometric interpretation of twisted connections, whose defining properties are purely algebraic. However, this draws a link with the geometry of $\manQ$ instead of $\manP$.

It is worth interpreting the degrees of freedom of these connections. A (twisted) connection is defined fiberwise, provided that it varies smoothly from one fiber to another. Given its value at one point, its values in the corresponding fiber is determined by equivariance. Let $n=\dim M,\ k=\dim H,\ \ell=\dim G$. In a given basis, $\omega_q^\manQ$ is represented by a $\ell\times(n+\ell)$ matrix and $\omega_p^\manP$ is represented by a $\ell\times(n+k)$ matrix. They have to satisfy the vertical normalization $\omega_q^\manQ(\kb_q^\ver)=\kb$ and $\omega_p^\manP(\ka_p^\ver)=\dr_HC_{(p,e)}(\ka)$. These correspond to $\ell$ (resp. $k$) independent equations with $\ell$ components each, which account for $\ell^2$ (resp. $k\ell$) constraints on the $\ell(n+\ell)$ (resp. $\ell(n+k)$) components. This leaves both $\omega^\manQ$ and $\omega^\manP$ with $n\ell$ degrees of freedom corresponding to the freedom of mapping $n=\dim M$ independent vectors to elements of $\kg$ which have $\ell$ components.
\end{remark}

The previous isomorphism between twisted and standard connections is compatible with the definitions of twisted and standard covariant differentials.

\begin{proposition}
Let $\omega^\manQ$ be a standard connection with covariant differential $\Dr^\manQ \defeq \dr_\manQ+\rho_*(\omega^\manQ)$. Let $\Dr^\manP \defeq \dr_\manP+\rho_*(\tilde{\chi}_{\manP\manQ}(\omega^\manQ))$ (which is a twisted covariant differential since $\tilde{\chi}_{\manP\manQ}(\omega^\manQ)$ is a twisted connection). Then for any $\alpha^\manQ\in\Omega_\tens^\bullet(\manQ,V)$,
\begin{align*}
\Dr^\manP(\chi_{\manP\manQ}(\alpha^\manQ))
&= \chi_{\manP\manQ}(\Dr^\manQ(\alpha^\manQ)).
\end{align*}
\end{proposition}

\begin{proof}
Let $\alpha^\manQ\in\Omega_\tens^k(\manQ,V)$. On the one hand,
\begin{align*}
\dr_\manP(\chi_{\manP\manQ}(\alpha^\manQ))
&= \dr_\manP((\pi_H\circ i_e)^*\alpha^\manQ)
= (\pi_H\circ i_e)^*\dr_\manQ(\alpha^\manQ)
= \tilde{\chi}_{\manP\manQ}(\dr_\manQ(\alpha^\manQ))
\end{align*}
where $\dr$ and $(\pi_H\circ i_e)^*$ commute because $\pi_H\circ i_e$ is a differentiable map. On the other hand, recall that given a basis $\{e_i\}_i$ of $V$ and a basis $\{\kb_j\}_j$ of $\kg$, $\alpha^\manQ=\alpha^{\manQ i}\otimes e_i$ and $\omega^\manQ=\omega^{\manQ j}\otimes \kb_j$ where $\alpha^{\manQ i}\in\Omega_\tens^k(\manQ)$ and $\omega^{\manQ j}\in\Omega_\eqv^1(\manQ)$. The implicit summation on repeated indices is used. In these bases, $\rho_*(\omega^\manQ)\cdotact \alpha^\manQ \defeq (\omega^{\manQ j}\wedge\alpha^{\manQ i})\otimes(\rho_*(\kb_j)\cdotact e_i)$ (which does not depend on the basis). Then,
\begin{align*}
\tilde{\chi}_{\manP\manQ}\left(\rho_*(\omega^\manQ)\cdotact \alpha^\manQ\right)_p(\{X_p^\manP\})
&= (\rho_*(\omega^\manQ)\cdotact \alpha^\manQ)_{[p,e]}\circ\Tg_{(p,e)}\pi_H(\{X_p^\manP\oplus 0_e\})
\\
&= \left((\omega^{\manQ j}\wedge\alpha^{\manQ i})_{[p,e]}\circ\Tg_{(p,e)}\pi_H(\{X_p^\manP\oplus 0_e\})\right)\otimes(\rho_*(\kb_j)\cdotact e_i)
\\
&= 
\begin{multlined}[t]
\dfrac{1}{k!}\sum_\sigma\Bigl(\omega_{[p,e]}^{\manQ j}(\Tg_{(p,e)}\pi_H(X_{\sigma(0),p}^\manP\oplus 0_e))
\\
\alpha_{[p,e]}^{\manQ i}(\{\Tg_{(p,e)}\pi_H(X_{\sigma(m),p}^\manP\oplus 0_e)\}_{m\in[1..k]})\Bigr)\otimes(\rho_*(\kb_j)\cdotact e_i)
\end{multlined}
\\
&= 
%\begin{multlined}[t]
\dfrac{1}{k!}\sum_\sigma\Bigl(\tilde{\chi}_{\manP\manQ}(\omega^{\manQ j})_p(X_{\sigma(0),p}^\manP)
%\\
\chi_{\manP\manQ}(\alpha^{\manQ i})(\{X_{\sigma(m),p}^\manP\}_{m\in[1..k]})\Bigr)\\
&\hspace{9cm}\otimes(\rho_*(\kb_j)\cdotact e_i)%\vspace{-2cm}
%\end{multlined}
\\
&= \Bigl( \tilde{\chi}_{\manP\manQ}(\omega^{\manQ j})\wedge\chi_{\manP\manQ}(\alpha^{\manQ i}) \Bigr)_p(\{X_p^\manP\})\otimes\rho_*(\kb_j)\cdotact e_i\\
&= \Bigl( \rho_*(\tilde{\chi}_{\manP\manQ}(\omega^\manQ)\cdotact\chi_{\manP\manQ}(\alpha^\manQ) \Bigr)_p(\{X_p^\manP\}),
\end{align*}
where $\sum_\sigma$ is a sum over all permutations of $\llbracket 0,k \rrbracket$. This concludes the proof.
\end{proof}

The previous propositions highlight the status of twisted gauge fields on $\manP$ as pull-backs of standard gauge fields on $\manQ$ through $\pi_H\circ i_e$. Since fields on $\manQ$ undergo gauge transformations under $\Autv(\manQ) \simeq \calG^\manQ = \{\gamma^\manQ : \manQ\rightarrow G \mid \gamma^\manQ(qg)=g\inv\gamma^\manQ(q)g\}$, their twisted equivalent on $\manP$ undergo an induced gauge transformation.

\begin{proposition}\label{prop isom GT PQ}
The space $\calG^\manP \defeq \{\zeta^\manP:\manP\rightarrow G \mid \zeta^\manP(ph)=C(p,h)\inv\zeta^\manP(p)C(p,h)\}$ is a group under multiplication in $G$. It acts on itself as\footnote{The superscript $\manP$ in $\zeta^\manP$ is dropped for convenience, when unambiguous.} $\zeta'^\zeta=\zeta\inv\zeta'\zeta$ and it acts on twisted fields as
\begin{align*}
\begin{aligned}
\omega^\zeta_p
&\defeq \zeta(p)\inv\omega_p\zeta(p)+\zeta(p)\inv\dr\zeta_p,
%\\
&
\Omega^\zeta_p
&\defeq \zeta(p)\inv\Omega_p\zeta(p),
\\
\varphi^\zeta(p)
&\defeq \rho(\zeta(p)\inv)\cdotact\varphi(p),
%\\
&
(\Dr^\manP\varphi)^\zeta(p)
&\defeq \rho(\zeta(p)\inv)\cdotact \Dr^\manP\varphi(p).
\end{aligned}
\end{align*}
\end{proposition}

\begin{proof}
The group property of $\calG^\manP$ is a consequence of the group property of $G$ and the equivariance property of $\zeta^\manP$ by conjugation. The action of $\calG^\manP$ is a direct consequence of Prop.~\ref{prop isom tensorial forms}. Note that $\calG^\manP=\calF_\Ceq(\manP,G)=\Omega_\Ceq^0(\manP,G)\simeq\Omega_\eqv^0(\manQ,G)=\calF_\eqv(\manQ,G)=\calG^\manQ$, where the twisted equivariance is stated in the definition of $\calG^\manP$. Then, any $\zeta^\manP\in\calG^\manP$ reads $\zeta^\manP=\chi_{\manP\manQ}(\gamma^\manQ)$ for a unique $\gamma^\manQ\in\calG^\manQ$. The action of $\calG^\manP$ is defined such that it is compatible with the isomorphism of structures on $\manP$ and $\manQ$:
\begin{align*}
\tilde{\chi}_{\manP\manQ}(\omega^\manQ)^{\chi_{\manP\manQ}(\gamma^\manQ)} 
&\defeq \tilde{\chi}_{\manP\manQ}\left((\omega^\manQ)^{\gamma^\manQ}\right).
\end{align*}
\end{proof}

It is worth insisting on the double gauge structure underlying the twisted fields on $\manP$.
\begin{center}
\renewcommand{\arraystretch}{1.3}
\begin{tabular}{rcc}
\toprule
 & Geometric structure group $H$ & Action structure group $G$\\
 \midrule
Gauge group & $\calH^\manP\simeq\Autv(\manP)$ & $\calG^\manP\simeq\calG^\manQ\simeq\Autv(\manQ)$
\\
Equivariance & $\gamma^\manP(ph)=h\inv\gamma^\manP(p)h$ & $\zeta^\manP(ph)=C(p,h)\inv\zeta^\manP(p)C(p,h)$
\\
Gauge transformation & $\omega^\gamma=C_\gamma\inv \omega C_\gamma+C_\gamma\inv \dr C_\gamma$ & $\omega^\zeta=\zeta\inv\omega\zeta+\zeta\inv\dr\zeta$\\
\bottomrule
\end{tabular}
\end{center}
On the one hand, equivariance properties encode the link between the fields (here, $\gamma^\manP,\ \zeta^\manP$) and the “geometric” $H$-structure on $\manP$. Since $\gamma^\manP$ belongs to this structure, its equivariance does not involve $C$. On the contrary, $\zeta^\manP$ belongs to the $G$-structure and needs $C$ to be connected to the geometry of $\manP$. On the other hand, twisted fields belong to the $G$-structure. Their $H$-gauge transformations involve $C$ while their $G$-gauge transformations do not need $C$.

%%%%%%%%%%%%%%%%
\section{\texorpdfstring{Connections and Gauge Structure on $\manS$}{Connections and Gauge Structure on S}}
\label{sec Connections GT on S}
%%%%%%%%%%%%%%%%

As a principal $(H\times G)$-bundle over $M$, $\manS=\manP\times G$ can be endowed with a standard connection and it supports an action of its standard gauge group $\Autv(\manS)$. In this section, following the process of Sec.~\ref{subsec Isom Connection PQ}, the spaces of connections and of gauge transformations on $\manS$ are related to those on $\manP$ and $\manQ$.

%%%%%%%%%%%%%%%%
\subsection{\texorpdfstring{Connections on $\manS$}{Connections on S}}
\label{subsec Connections on S}
%%%%%%%%%%%%%%%%

\begin{definition}
A (standard) connection $\omega^\manS\in\calA(\manS)$ is a vector-valued 1-form $\omega^\manS\in\Omega^1(\manS,\kh\oplus\kg)$ satisfying the following equivariance condition and vertical normalization
\begin{align}
\label{eq equiv omega S}
\calR_{(h,g')}^{\manS*}\omega^\manS
&= \Ad_{(h\inv, g^{\prime -1})}\omega^\manS\ \forall (h,g')\in H\times G,
\\
\label{eq VN omega S}
\omega_{(p,g)}^\manS\left((\ka,\kb)_{(p,g)}^\ver\right)
&= \omega_{(p,g)}^\manS(\ka_p^\ver\oplus(\Tg_e\Rr_g\circ\dr_H C_{(p,e)}\inv(\ka)+\kb_g))=(\ka,\kb),\ \forall (\ka,\kb)\in\kh\oplus\kg.
\end{align}
\end{definition}
As a $(\kh\oplus\kg)$-valued form, $\omega^\manS$ decomposes into $(\omega^\kh,\omega^\kg)$ with $\omega^\kh\in\Omega^1(\manS,\kh)$ and $\omega^\kg\in\Omega^1(\manS,\kg)$. Equation \eqref{eq equiv omega S} splits into
\begin{align*}
\begin{aligned}
\calR_{(h,g')}^{\manS*}\omega^\kh
&= \Ad_{h\inv}\omega^\kh,
&
\calR_{(h,g')}^{\manS*}\omega^\kg
&= \Ad_{ g^{\prime -1}}\omega^\kg.
\end{aligned}
\end{align*}
In particular, $\omega^\kh$ is $G$-invariant and $\omega^\kg$ is $H$-invariant. Equation \eqref{eq VN omega S} implies that
\begin{align}
\label{eq VN omega LieH}
\omega_{(p,g)}^\kh(\ka_p^\ver\oplus(\Tg_e\Rr_g\circ\dr_H C_{(p,e)}\inv(\ka)+\kb_g))
&= \ka,
\\
\label{eq VN omega LieG}
\omega_{(p,g)}^\kg(\ka_p^\ver\oplus(\Tg_e\Rr_g\circ\dr_H C_{(p,e)}\inv(\ka)+\kb_g))
&= \kb.
\end{align}
In particular, any $X_g^G\in\Tg_gG$ reads $\Tg_e\Rr_g\circ\dr_H C_{(p,e)}\inv(\ka) + \kb_g$ for an appropriate $\kb\in\kg$. Hence, $\forall \ka\in\kh$, $\forall X_g^G\in\Tg_gG$, $\omega_{(p,g)}^\kh(\ka_p^\ver\oplus X_g^G)=\ka$. This implies that $\omega^\kh$ is $G$-horizontal and its $H$-vertical normalization reads $\omega_{(p,g)}^\kh(\ka_p^\ver\oplus 0_g)=\ka$. Similarly, one gets that $\omega_{(p,g)}^\kg(\ka_p^\ver\oplus\Tg_e\Rr_g\circ\dr_HC_{(p,e)}\inv(\ka))=0$, \textit{i.e.} $\omega^\kg$ is $H$-horizontal. It will be useful to note that this implies
\begin{align*}
\omega_{(p,g)}^\kg(\ka_p^\ver\oplus 0_g)
&= \omega_{(p,g)}^\kg(0_p\oplus\Tg_e\Rr_g\circ\dr_H C_{(p,e)}(\ka))
= \dr_H C_{(p,e)}(\ka)\in\Tg_eG
= \kg.
\end{align*}
Lemma~\ref{lemma d_H C} was used here. Since $\omega^\kg$ is $H$-invariant and $H$-horizontal, it induces a well-defined form on $\manQ$. This is summarized in the following proposition.
\begin{proposition}\label{prop isom connections SQ}
The set $\calA(\manQ)$ of standard connections on $\manQ$ is isomorphic to the set of $H$-invariant $H$-horizontal $\omega^\kg\in\Omega^1(\manS,\kg)$ satisfying \eqref{eq VN omega LieG}. The isomorphism is provided by $\tilde{\chi}_{\manS\manQ} \defeq \pi_H^*:\calA(\manQ)\rightarrow\Omega^1(\manS,\kg)$.
\end{proposition}
\begin{proof}
The proof of this proposition is identical to the proof of Prop.~\ref{prop isom connections PQ}. Similarly, $\omega^\kh$ is $G$-invariant and $G$-horizontal, so it induces a form on $\manP$.
\end{proof}
\begin{proposition}\label{prop isom standard connections SP}
The set $\calA(\manP)$ of standard connections on $\manP$ is isomorphic to the set of $G$-invariant $G$-horizontal $\omega^\kh\in\Omega^1(\manS,\kh)$ satisfying \eqref{eq VN omega LieH}. The isomorphism is provided by $\pi_G^*:\calA(\manP)\rightarrow\Omega^1(\manS,\kh)$.
\end{proposition}
\begin{proof}
The proof is similar to the one of Prop.~\ref{prop isom connections PQ}. The following result provides a view on twisted connections on $\manP$ as induced structures from $\manS$.
\end{proof}
\begin{proposition}\label{prop isom twisted connections SP}
The set $\calA^C(\manP)$ of twisted connections on $\manP$ is isomorphic to the set of $H$-invariant $H$-horizontal $\omega^\kg\in\Omega^1(\manS,\kg)$ satisfying \eqref{eq VN omega LieG}. The isomorphism is provided by $\tilde{\chi}_{\manP\manS} \defeq i_e^*$.
\end{proposition}
\begin{proof}
This is a direct consequence of Prop.~\ref{prop isom connections PQ} and \ref{prop isom connections SQ}. Let $\omega^\manP\in\calA^C(\manP)$. There is a unique $\omega^\manQ\in\calA(\manQ)$ such that $\omega^\manP=\tilde{\chi}_{\manP\manQ}(\omega^\manQ)=i_e^*\pi_H^*\omega^\manQ$. Then $\omega^\manP=i_e^*\omega^\kg$ for a unique $\omega^\kg \defeq \pi_H^*\omega^\manQ$.
\end{proof}

Following the definition in Sec.~\ref{subsec Standard Gauge Fields}, one defines the covariant differential on $\manS$ associated to $\omega^\manS$ as $\Dr^\manS \defeq \dr_\manS+\omega^\manS$. Note that this is a shorthand notation which presupposes the use of a representation. $\Dr^\manS \defeq \dr_\manS+\rho_*(\omega^\manS)$ with $\rho_*$ induced on $\kh\oplus\kg$ by $\rho_{H\times G}$. Given the context, the appropriate representation is $(W\times V,\rho_{H\times G})$ such that $\rho_{H\times G}(h,g)\cdotact(w,v) \defeq (\rho_H(h)\cdotact w,\rho_G(g)\cdotact v)$. As a consequence of the previous correspondences, it is relevant to define “semi-covariant” differentials $\Dr^\kh=\dr_\manS+(\omega^\kh,0)$ and $\Dr^\kg=\dr_\manS+(0,\omega^\kg)$. The previous propositions together with the commutation of $\dr$ and the pullback of any bundle map implies that $\Dr^\kh$ corresponds to a unique covariant differential on $\manP$, and, $\Dr^\kg$ corresponds to a unique covariant differential on $\manQ$ and a unique twisted covariant differential on $\manP$.

%%%%%%%%%%%%%%%%
\subsection{\texorpdfstring{Gauge Structure on $\manS$}{Gauge Structure on S}}
\label{subsec Gauge Structure on S}
%%%%%%%%%%%%%%%%

The group of vertical automorphisms of $\manS$, $\Autv(\manS) \defeq \{\Phi^\manS\in\Diff(\manS) \mid  \Phi^\manS\circ\calR_{(h,g')}^\manS = \calR_{(h,g')}^\manS\circ\Phi^\manS,\ \pi_\manS\circ\Phi^\manS=\pi_\manS\}$ is isomorphic to the gauge group of $\manS$, $\Gau(\manS) \defeq \{\gamma^\manS : \manS\rightarrow  H\times G \mid \calR_{(h,g')}^{\manS*}\gamma^S=(h\inv, g^{\prime -1})\gamma^\manS(h,g')\}$. Each $\gamma^\manS\in\Gau(\manS)$ decomposes as $(\gamma^H,\gamma^G)$ and the equivariance relation of $\gamma^\manS$ implies
\begin{align*}
\begin{aligned}
\calR_{(h,g')}^{\manS*}\gamma^H
&= h\inv\gamma^H h,
&
\calR_{(h,g')}^{\manS*}\gamma^G
&= g^{\prime -1}\gamma^Gg'.
\end{aligned}
\end{align*}
Hence, $\Gau(\manS)=\calH^\manS\times\calG^\manS$ where $\calH^\manS=\{\gamma^H : \manS\rightarrow H \mid \calR_{(h,g')}^{\manS*}\gamma^H = h\inv\gamma^H h\}$ and $\calG^\manS=\{\gamma^G : \manS \rightarrow G \mid \calR_{(h,g')}^{\manS*}\gamma^G= g^{\prime -1}\gamma^Gg'\}$.

Note that $\gamma^G$ is a $H$-invariant map. Thus it induces a map on $\manQ$.
\begin{proposition}\label{prop isom GT SQ}
There is an isomorphism between $\calG^\manQ$ and $\calG^\manS$ given by $\pi_H^*$.
\end{proposition}
\begin{proof}
The proof is analogous to the one of Prop.~\ref{prop isom GT PQ} and is a direct consequence of Corollary~\ref{corollary isom eq functions}:
\begin{align*}
\calG^\manQ
=\calF_\eqv(\manQ,G)
\simeq\calF_\eqv(\manS,G)
=\calG^\manS,
\end{align*}
where the equivariance on $\calF_\eqv(\manS,G)$ is a $G$-equivariance $H$-invariance.
\end{proof}

An analogous result holds for $H$-gauge transformations.
\begin{proposition}\label{prop isom GT SP}
There is an isomorphism between $\calH^\manP$ and $\calH^\manS$ given by $\pi_G^*$ and $i_e^*$.
\end{proposition}
\begin{proof}
On the one hand,
\begin{align*}
(i_e^*\gamma^H)(ph)
&= \gamma^H(ph,e)
= h\inv\gamma^H(p,C(p,h))h
= h\inv\gamma^H(p,e)h
= h\inv( i_e^*\gamma^H)(p)h.
\end{align*}
On the other hand,
\begin{align*}
(\pi_G^*\gamma^\manP)(ph,C(p,h)\inv gg')
= \gamma^\manP(ph)
= h\inv\gamma^\manP(p)h
= h\inv(\pi_G^*\gamma^\manP)(p,g)h,
\end{align*}
since, as stated before, $\pi_G^*\gamma^\manP$, being $G$-invariant, does not depend on the second variable of its argument. Finally, note that $\pi_G\circ i_e=\id_\manP$ and $i_e\circ\pi_G(p,g)=(p,e)$, so $(\pi_G\circ i_e)^*$ is the identity on maps on $\manP$ and $(i_e\circ\pi_G)^*$ is the identity on $G$-invariant maps on $\manS$.
\end{proof}

Finally, the following result provides an interpretation for twisted gauge transformations on $\manP$.
\begin{proposition}\label{prop isom twisted GT SP}
There is an isomorphism between $\calG^\manP$ and $\calG^\manS$ given by $i_e^*$.
\end{proposition}
\begin{proof}
This comes from Prop.~\ref{prop isom GT PQ} and \ref{prop isom GT SQ}. Given $\zeta^\manP\in\calG^\manP$, there is a unique $\gamma^\manQ\in\calG^\manQ$ such that $\zeta^\manP=i_e^*\pi_H^*\gamma^\manQ$. Since $\gamma^\manQ\in\calG^\manQ\mapsto\pi_H^*\gamma^\manQ \rdefeq \gamma^G\in\calG^\manS$ is an isomorphism, the correspondence $\zeta^\manP=i_e^*\gamma^\manS$ is 1:1.
\end{proof}

%%%%%%%%%%%%%%%%
\subsection{\texorpdfstring{Some Comments about $\manS$ and $\manQ$}{Some Comments about S and Q}}
\label{subsec Interpretation Comments}
%%%%%%%%%%%%%%%%

The current and previous sections highlight the origin of the standard and twisted gauge structure on $\manP$. $\manP(M,H)$ is a subbundle of a specific bundle $\manS(M,H\times G)$ with the same trivializing sections. Since $\manS=\manP\times G$ is trivial (as a manifold) along $G$, $\manP$ foliates $\manS$. In particular, $\manP$ is embedded in $\manS$ through $i_e:\manP\rightarrow\manS,p\mapsto(p,e)$. The standard and twisted gauge structures on $\manP$ appear as being the pullback of the full (standard) gauge structure on $\manS$ through $i_e$.

Note that $\manS$ is trivial along $G$ as a manifold but it is not trivial as a principal bundle (\textit{i.e.} as a $(H\times G)$-space). This is due to the right action of $H\times G$ on $\manS$ considered here (specifically the right action of $H$) which connects $H$ and $G$: $\calR_{(h,g')}^\manS(p,g)=(ph,C(p,h)\inv gg')$. In terms of $\calR_{(e,g)}^\manS$ and of the simpler (“non-mixing”) action $\widehat{\calR}$ of $H$ on $\manS$ defined as $\widehat{\calR}_h(p,g)=(ph,g)$, $\calR_{(h,e)}^\manS$ reads
\begin{align*}
\calR_{(h,e)}^\manS(p,g)
= (ph,C(p,h)\inv g)
= (ph,gg\inv C(p,h)\inv g)
= \widehat{\calR}_h\circ\calR_{(e,g\inv C(p,h)\inv g)}^\manS(p,g).
\end{align*}
(Conversely, $\widehat{\calR}_h(p,g)=\calR_{(h,g\inv C(p,h) g)}^\manS(p,g)$.) This action carries the action\footnote{$g\mapsto C(p,h)\inv g$ is an action of $H$ on $G$ which does not preserve the group structure of $G$, \textit{i.e.} this action is not a morphism of $G$.} of $H$ on $G$ through $C$ which induces twisted gauge fields. Additionally, $i_e$ has the right interplay with $\calR^\manS$ in order to generate the twisted equivariance on $\manP$. Let $\varphi^\manS:\manS\rightarrow G$ be both $G$-equivariant and $H$-invariant, and, let $i_e^*\varphi^\manS$ be the map it induces on $\manP$.
\begin{align*}
(i_e^*\varphi^\manS)(ph)
&= \varphi^\manS(ph,e)
= \varphi^\manS(p,C(p,h))
= C(p,h)\inv\cdotact\varphi^\manS(p,e)
= C(p,h)\inv\cdotact(i_e^*\varphi^\manS)(p).
\end{align*}
A deeper interpretation of this equation is given in terms of cohomology in Sec.~\ref{subsec Interpretation Comments Again}.

Following the same interpretation, $\manQ$ is seen as the space carrying only the $G$-structure inherited from $\manS$ (\textit{i.e.} the twisted structure on $\manP$) and forgetting about the $H$-structure on $\manS$ (\textit{i.e.} the standard structure on $\manP$). Indeed, contrary to $\manS$, $\manQ$ is not a $H$-space. However, since $\manQ$ is the space of orbits of $\calR_{(h,e)}^\manS$ (which explicitly depends on $C$), the topology of $\manQ$ is expected to depend on $C$. See Sec.~\ref{subsec Isom Changing Cocycle} for more details. This provides three different approaches of $C$. The cocycle deals with
\begin{itemize}
\item the algebraic structure on $\manP$,
\item the geometry of $\manS$,
\item the topology of $\manQ$.
\end{itemize}
Note that, since $\manP$ is embedded into $\manS$ through $i_e$, one can equally see standard and twisted fields on $\manP$ as standard and twisted fields on $\manS$. Indeed, one can pull back these fields from $\manP$ to $\manS$ through $\pi_G$, for instance, the connections, according to
\begin{equation*}
\begin{tikzcd}
\manS\arrow[rr,"i_e\circ\pi_G"]\arrow[rd,"\pi_G"']& &\manS & &\tilde{\omega}^\manS=(\omega^\kh,\tilde{\omega}^\kg)& &\omega^\manS=(\omega^\kh,\omega^\kg)\arrow[ll,"\pi_G^*i_e^*"',mapsto]\arrow[ld,"i_e^*",mapsto]\\
 & \manP\arrow[ru,"i_e"'] & & & & (\omega^\manP_H,\omega^\manP_G)\arrow[lu,"\pi_G^*",mapsto]& 
\end{tikzcd}
\end{equation*}
where $\omega^\manS$ is a standard connection on $\manS$, $\omega^\manP_H$ stands for a standard connection on $\manP$, $\omega^\manP_G$ stands for a twisted connection on $\manP$ and  $\tilde{\omega}^\manS=(\omega^\kh,\tilde{\omega}^\kg)$ stands for a standard and a twisted connection on $\manS$. The same correspondence holds for equivariant maps $\varphi^\manS$ and tensorial forms $\alpha^\manS$. The intuition behind this move is that the pulled-back fields correspond to the standard fields evaluated at the specific point defined by the embedding $i_e$: $\tilde{\varphi}^\manS(p,g) \defeq \varphi^\manS(p,e)$ for any $g\in G$.
\begin{definition}
A twisted form on $\manS$ (or twisted $H$-field) is a $C(H)$-equivariant $G$-invariant form which is $G$-horizontal (\textit{i.e.} it vanishes on vectors of the form $0\oplus X_g^G$) and satisfies the relevant property with respect to vertical tangent vectors to $\manP$. For instance, a twisted connection $\omega_C^\manS\in\Omega^1(\manS,\kh\oplus\kg)$ satisfies

\begin{itemize}
\item $\calR_{(e,g)}^{\manS*}\omega_C^\manS=\omega_C^\manS$ ($G$-equivariance),
\item $\calR_{(h,e)}^{\manS*}\omega_{C|(p,g)}^\manS=\Ad_{C(p,h)\inv}\omega_{C|(p,g)}^\manS+C(p,h)\inv\dr_\manP C_{|(p,h)}\circ\Tg_{(p,g)}\pi_G$ ($C(H)$-equivariance),
\item $\omega_{C|(p,g)}^\manS((\ka,\kb)_{(p,g)}^\ver)=\dr_H C_{|(p,e)}(\ka)$ (vertical normalization w.r.t. $\Ver\manP\oplus 0$ and horizontality w.r.t. $0\oplus\Tg G$).
\end{itemize}
\end{definition}
\begin{proposition}\label{prop twisted fields on S}
$\pi_G^*i_e^*$ maps standard $H$-fields to themselves and standard $G$-fields to twisted $H$-fields.
\end{proposition}
\begin{proof}
The proof is given for connections and a similar proof holds for fields of other nature. Let $\omega^\manS$ be a standard connection on $\manS$. Following Prop.~\ref{prop isom connections SQ}, \ref{prop isom standard connections SP} and \ref{prop isom twisted connections SP}, this connection decomposes as $\omega^\manS=\pi_G^*\omega^\manP+\pi_H^*\omega^\manQ$ where $\omega^\manP$ (resp. $\omega^\manQ$) is a standard\footnote{Beware, the notations are different from Sec.~\ref{subsec Isom Connection PQ}} $H$-connection (resp. $G$-connection) on $\manP$ (resp. $\manQ$). Additionally, $\omega_C^\manP \defeq i_e^*\pi_H^*\omega^\manQ$ is a twisted connection on $\manP$. On the one hand,
\begin{align*}
\pi_G^*i_e^*(\pi_G^*\omega^\manP)
&= \pi_G^*(\pi_G\circ i_e)^*\omega^\manP
= \pi_G^*(\id_\manP)^*\omega^\manP
= \pi_G^*\omega^\manP
\end{align*}
So, the component of $\omega^\manS$ related to $H$ is unchanged. On the other hand,
\begin{align*}
\pi_G^*i_e^*(\pi_H^*\omega^\manQ)
&= \pi_G^*\omega_C^\manP 
\rdefeq \omega_C^\manS.
\end{align*}
This new field satisfies
\begin{align*}
\calR_{(e,g)}^{\manS*} \omega_C^\manS
&= \calR_{(e,g)}^{\manS*}\pi_G^*\omega_C^\manP
= (\pi_G\circ\calR_{(e,g)}^\manS)^*\omega_C^\manP
= \pi_G^*\omega_C^\manP=\omega_C^\manS,
\end{align*}
\textit{i.e.} $\omega_C^\manS$ is $G$-invariant, and
\begin{align*}
\calR_{(h,e)}^{\manS*}\omega_{C|(p,g)}^\manS 
&= \calR_{(h,e)}^{\manS*}\pi_G^*\omega_{C|(p,g)}^\manP
= (\pi_G\circ\calR_{(h,e)}^\manS)^*\omega_{C|(p,g)}^\manP
= (\calR_h^\manP\circ\pi_G)^*\omega_{C|(p,g)}^\manP
= \pi_G^*\calR_h^{\manP*}\omega_{C|(p,g)}^\manP
\\
&= \calR_h^{\manP*}\omega_{C|p}^\manP\circ\Tg_{(p,g)}\pi_G
= (\Ad_{C(p,h)\inv}\omega_{C|p}^\manP + C(p,h)\inv\dr_\manP C_{|(p,h)})\circ\Tg_{(p,g)}\pi_G
\\
&= \Ad_{C(p,h)\inv}\omega_{C|(p,g)}^\manS + C(p,h)\inv\dr_\manP C_{|(p,h)}\circ\Tg_{(p,g)}\pi_G,
\end{align*}
\textit{i.e.} $\omega_C^\manS$ is $C(H)$-equivariant. Let $(\ka,\kb)\in\kh\oplus\kg$.
\begin{align*}
\omega_{C|(p,g)}^\manS((\ka,\kb)_{(p,g)}^\ver)
&= \pi_G^*\omega_{C|(p,g)}^\manP \left(\ka_p^\ver \oplus \left( \Tg_e\Rr_g\circ(\dr_HC\inv )_{(p,e)}(\ka) + \kb_g \right)\right)
= \omega_{C|p}^\manP(\ka_p^\ver)
\\
&= \dr_H C_{|(p,e)}(\ka).
\end{align*}
\end{proof}

%%%%%%%%%%%%%%%%
\section{Groupoid Approach}
\label{sec Groupoids}
%%%%%%%%%%%%%%%%

This section provides an alternative point of view on $\manS$ and $\pi_H\circ i_e$ in terms of Lie groupoids. The content of this section follows the conventions of \cite{mackenzie_general_2005, landsman_muhlyrenaultwilliams_2000, del_hoyo_lie_2013, bursztyn_lie_2023, del_hoyo_cohomology_2024} and some results can be found there.

%%%%%%%%%%%%%%%%
\subsection{Definitions and Morita Equivalence}
\label{subsec Def and Morita}
%%%%%%%%%%%%%%%%

\begin{definition}
A \textbf{Lie groupoid} $G=(G_1\rra G_0)$ consists of two smooth manifolds $G_1$, $G_0$, called respectively ``space of arrows'' and ``space of objects'' together with
\begin{itemize}
\item two smooth surjections $\alpha,\beta:G_1\to G_0$ respectively called ``source (projection)'' and ``target (projection)'';
\item a smooth partial multiplication $\mu:G*G\to G,\ (g',g)\mapsto g'g$, where $G*G\defeq\{(g',g)\in G_1\times G_1,\ \alpha(g')=\beta(g)\}$;
\item a smooth inclusion map $i:G_0 \to G_1,\ x\mapsto \bbbone_x$ and a smooth inverse map $\mathrm{inv}:G_1 \to G_1,\ g\mapsto g\inv$;
\end{itemize}
satisfying
\begin{enumerate}[label=(\roman*)]
\item $\alpha(g'g)=\alpha(g)$ and $\beta(g'g)=\beta(g')$ for all $(g',g)\in G*G$;
\item $\mu$ is associative;
\item $\alpha(\bbbone_x)=\beta(\bbbone_x)=x$ for all $x\in G_0$;
\item $g\bbbone_{\alpha(g)}=\bbbone_{\beta(g)}g=g$ for all $g\in G_1$;
\item $\alpha(g\inv)=\beta(g)$, $\beta(g\inv)=\alpha(g)$, $g\inv g=\bbbone_{\alpha(g)}$ and $gg\inv=\bbbone_{\beta(g)}$ for all $g\in G_1$.
\end{enumerate}
\end{definition}

\begin{definition}
Let $H=(H_1\rra H_0)$ and $G=(G_1\rra G_0)$ be two Lie groupoids. A \textbf{(Lie) groupoid morphism} $\phi=(\phi_1,\phi_0):H\to G$ is a pair of smooth maps $\phi_1:H_1\to G_1$ and $\phi_0:H_0\to G_0$ such that $\alpha'\circ\phi_1=\phi_0\circ\alpha$, $\beta'\circ\phi_1=\phi_0\circ\beta$ and $\phi_1(h'h)=\phi_1(h')\phi_1(h)$ for any $(h',h)\in H*H$.
\end{definition}

\begin{definition}
Let $G=(G_1\rra G_0)$ be a Lie groupoid, $S$ a smooth manifold and $\tau:S\to G_0$ a smooth map. Denote $G*S=\{(g,s)\in G_1\times S,\ \alpha(g)=\tau(s)\}$. A \textbf{left Lie groupoid action} of $G$ on $\tau$, denoted $(G_1\rra G_0)\lact(S\xrightarrow{\tau}G_0)$ or $G\lact S$, is a smooth map $A:G*S\to S,\ (g,s)\mapsto A_gs=:gs$ satisfying
\begin{enumerate}[label=(\roman*)]
\item $\tau(gs)=\beta(g)$ for all $(g,s)\in G*S$;
\item $g'(gs)=(g'g)s$ for all compatible $(g',g,s)\in G_1\times G_1\times S$;
\item $\bbbone_{\tau(s)}s=s$ for all $s\in S$.
\end{enumerate}
\end{definition}
Equivalently, one can define right groupoid actions $(S\xrightarrow{\sigma}G_0)\ract(G_1\rra G_0)$ where the source and target are exchanged. Both $\tau$ and $\sigma$ are called "momentum maps", or sometimes "anchors".
\begin{definition}
A (left Lie) \textbf{groupoid representation} is a (left Lie) groupoid action $(G_1\rra G_0)\lact(E\xrightarrow{\tau}G_0)$ where $E\xrightarrow{\tau}G_0$ is a vector bundle and $A_g:E_{\alpha(g)}\to E_{\beta(g)}$ is a linear map for any $g\in G_1$.
\end{definition}

\begin{definition}
Let $(H_1\rra H_0)\lact(S_H\xrightarrow{\tau_H}H_0)$ and $(G_1\rra G_0)\lact(S_G\xrightarrow{\tau_G}G_0)$ be two Lie groupoid actions. Let $(\phi_1,\phi_0):G\to G$ be a groupoid morphism. An \textbf{equivariant map} $\psi:S_H\to S_G$ over $\phi_0$ is a map satisfying $\psi(hs)=\phi_1(h)\psi(s)$ for all $(h,s)\in H*S$.
\end{definition}
Any action of a Lie group on a smooth manifold can be encoded as a groupoid.
\begin{definition}
Let $G$ be a Lie group and $M$ a smooth manifold endowed with a left Lie group action $a:G\to \Diff(M)$ reading $a_g(x)=gx$. The \textbf{action groupoid} $G\ltimes M=(G\times M\rra M)$ is the groupoid with arrows $(g,x)\in G\times M$, with source and target $\alpha(g,x)=x$ and $\beta(g,x)=gx$, and with partial multiplication $(g',gx)(g,x)=(g'g,x)$ for all $(g',g,x)\in G\times G\times M$
\end{definition}

\begin{definition}
A Lie groupoid action $G\lact S$ is a \textbf{principal action} if
\begin{enumerate}[label=(\roman*)]
\item $\tau:S\to G_0$ is a surjective submersion (\textit{i.e.} $\forall s\in S,\ \rank(\Tg_s\tau)=\dim G_0$) and
\item the action is free (\textit{i.e.} $\forall s\in S,\ gs=s \Leftrightarrow g=\bbbone_{\tau(s)}$) and proper (\textit{i.e.} the map $G*S\to S\times S,\ (g,s)\mapsto (gs,s)$ is proper).
\end{enumerate}
\end{definition}
The definition adapts to right actions $S\ract G$.

\begin{definition}
Let $H=(H_1\rra H_0)$ and $G=(G_1\rra G_0)$ be two left Lie groupoids. A smooth manifold $S$ is called a \textbf{$(H,G)$-bibundle} if it is endowed with a left and a right Lie groupoid actions $H\lact(S\xrightarrow{\tau} H_0)$ and $(S\xrightarrow{\sigma} G_0)\ract G$ which commute, \textit{i.e.}
\begin{itemize}
\item $\tau(sg)=\tau(s)$ for all $(s,g)\in S*G$;
\item $\sigma(hs)=\sigma(s)$ for all $(h,s)\in H*S$;
\item $(hs)g=h(sg)$ for all $(h,s,g)\in H_1\times S\times G_1$ such that $\alpha(h)=\tau(s)$ and $\sigma(s)=\beta(g)$.
\end{itemize}
A $(H,G)$-bibundle is denoted $H\lact S\ract G$.
\end{definition}
It is understood that $\alpha$, $\beta$ may refer to $\alpha^H,\beta^H:H_1\to H_0$ or $\alpha^G,\beta^G:G_1\to G_0$ depending on the context (and of the argument).

\begin{definition}
A $(H,G)$-bibundle $H\lact S\ract G$ is called an \textbf{equivalence bibundle} (or \textbf{Hilsum-Skandalis bibundle}) between $H$ and $G$ if
\begin{enumerate}
\item $H\lact S$ and $S\ract G$ are both principal actions and
\item one has $S/G\simeq H_0$ and $H\backslash S\simeq G_0$ with isomorphisms induced by $\tau$ and $\sigma$ respectively.
\end{enumerate}
\end{definition}

\begin{definition}
Two groupoids $H$ and $G$ are \textbf{Morita equivalent} (or \textbf{weakly equivalent}) if there exists an equivalence bibundle between them.
\end{definition}

\begin{proposition}\label{prop Morita isom Rep}
If two groupoids $H$ and $G$ are equivalent, then the categories of representations $\Rep(H)$ and $\Rep(G)$ are isomorphic.
\end{proposition}
This proposition is stated and generalized in \cite[Theorem 8.6]{bursztyn_lie_2023}.

%%%%%%%%%%%%%%%%
\subsection{Morita Equivalence of Untwisting Bundles}
\label{subsec Morita Eq Untwist Bdls}
%%%%%%%%%%%%%%%%

This section applies the definitions in the previous section to the context of the bundles $\manP$, $\manS$ and $\manQ$.

Let's define the action groupoid $H\ltimes\manP=(H\times \manP\rra\manP)$ with
\begin{itemize}
\item source and target projections $\alpha(h,p)\defeq p$ and $\beta(h,p)\defeq ph\inv\rdefeq h\cdotact p$ for any $(h,p)\in H\times\manP$;
\item partial multiplication $(h'h,p)\rdefeq(h',h\cdotact p)(h,p)$, or equivalently, $(h'h,p)\rdefeq(h',ph\inv)(h,p)$ for any $(h',h,p)\in H\times H\times\manP$;
\item inverse map $(h,p)\inv\defeq(h\inv,h\cdotact p)=(h\inv,ph\inv)$ for any $(h,p)\in H\times\manP$;
\item identity map $\bbbone_p\defeq(e,p)$ for any $p\in \manP$.
\end{itemize}
One defines the action groupoid $G\ltimes\manQ=(G\times \manQ\rra\manQ)$ similarly. The same notation is used for sources and targets of both groupoid when the context makes the distinction clear.

There is a left Lie groupoid action of $H\ltimes \manP$ on $\pi_G:\manS\to\manP,\ (p,g_0)\mapsto p$ given by $A_{(h,p)}s=\calR_{(h\inv,e)}^\manS s\rdefeq (h,p)\cdotact s$ for any $((h,p),s)\in (H\times\manP)*\manS=\{((h,p),s)\in (H\times\manP)\times\manS,\ \alpha(h,p)=\pi_G(s)\}=\{((h,p),s)\in (H\times\manP)\times\manS,\ \pi_G(s)=p\}$. Taking into account the compatibility condition, the groupoid action reads
\begin{align*}
A_{(h,p)}(p,g_0)=(ph\inv,C(p,h\inv)\inv g_0)=(h\cdotact p,C(h\cdotact p,h)g_0).
\end{align*}

Similarly, there is a right Lie groupoid action of $G\ltimes \manQ$ on $\pi_H:\manS\to\manQ,\ (p,g_0)\mapsto [p,g_0]$ given by $A_{(g,q)}s=\calR_{(e,g)}^\manS s\rdefeq s\cdotact (g,q)$ for any $(s,(g,q))\in \manS*(G\times\manQ)=\{(s,(g,q))\in \manS\times(G\times\manQ),\ \pi_H(s)=\beta(g,q)\}=\{(s,(g,q))\in \manS\times(G\times\manQ),\ \pi_H(s)=qg\inv\}$. Thus the groupoid action reads
\begin{align*}
A_{(g,g\inv\cdotact[pg_0])}(p,g_0)=A_{(g,[pg_0g])}(p,g_0)=(p,g_0g).
\end{align*}
Thus, $(H\ltimes \manP)\lact \manS\ract(G\ltimes\manQ)$ is a bibundle. In addition, both actions are principal ($\pi_G$ and $\pi_H$ are surjective submersions and both actions are free and proper) since $\manS(\manP,G,\pi_G)$ and $\manS(\manQ,H,\pi_H)$ are principal bundles. In addition, since both groupoid actions are induced by group actions, the orbit spaces are $(H\ltimes\manP)\backslash\manS=\manS/H=\manQ$ and $\manS/(G\ltimes\manQ)=\manS/G\simeq\manP$. As a consequence, $(H\ltimes \manP)\lact \manS\ract(G\ltimes\manQ)$ is an equivalence bundle and $H\ltimes\manP$ and $G\ltimes\manQ$ are Morita equivalent.

Let's remark that the previous setup specializes in the case $G=\{e\}$ and leads to the equivalence bibundle $(H\ltimes \manP)\lact\manP\ract(M\rra M)$. Indeed, $\manS=\manP\times\{e\}\simeq\manP$, $\manQ=\manS/H\simeq\manP/H\simeq M$ and $G\ltimes\manQ=(\{e\}\ltimes M\rra M)=(M\rra M)$. Thus $H\ltimes\manP$ is Morita equivalent to $M\rra M$. One can adapt this result by replacing $\manP$ (resp. $H$) by $\manQ$ (resp. $G$) to get the equivalence bibundle $(G\ltimes \manQ)\lact\manQ\ract(M\rra M)$ and the Morita equivalence between $G\ltimes\manQ$ and $M\rra M$. Both equivalences are summarized in
\begin{align*}
(H\ltimes\manP)\lact\manP\ract(M\rra M)\lact\manQ\ract(G\ltimes\manQ),
\end{align*}
where the left and right actions on $\manQ$ were inverted. The two actions $\manP\ract(M\rra M)\lact\manQ$ enable quotienting the space $\manP\times\manQ$ by $M\rra M$, leading to $\manP\times_M\manQ=\{ (p,q)\in\manP\times\manQ,\ \pi_\manP(p)=\pi_\manQ(q) \}=\{ (p,[p,g]),\ (p,g)\in\manP\times G \}\simeq\manS$.

%%%%%%%%%%%%%%%%
\subsection{Comparison with Sec.~\ref{sec Isomorphism} and \ref{sec Connections GT on S} }
\label{subsec Comparison}
%%%%%%%%%%%%%%%%

The groupoid approach of the untwisting procedure clarifies several constructions at hand. As seen above, $\manS$ is the Hilsum-Skandalis bibundle encoding the Morita equivariance of the groupoids $H\ltimes \manP$ and $G\ltimes \manQ$. The map used to pullback standards gauge fields on $\manQ$ to twisted gauge fields on $\manP$ (see Sec.~\ref{subsec Isom Bdl Geom} and \ref{subsec Isom Connection PQ}) appears to be the base map of the groupoid morphism $\phi=(\phi_1,\phi_0):H\ltimes\manP\to G\ltimes \manQ$ be defined as $\phi_0=\pi_H\circ i_e:p\mapsto [p,e]$ and $\phi_1:(h,p)\mapsto(C(h\cdotact p,h),\phi_0(p))=(C(h\cdotact p,h),[p,e])$.
\begin{remark}
It is more straightforward that $(p,h)\mapsto ([p,e],C(p,h))$ is a groupoid morphism $\manP\rtimes H\to\manQ\rtimes G$. Indeed, for any $(p,h,h')\in\manP\times H\times H$,
\begin{align*}
([p,e],C(p,hh'))&=([p,e],C(p,h)C(ph,h'))\\
&=([p,e],C(p,h))\,([p,e]C(p,h),C(ph,h'))\\
&=([p,e],C(p,h))\,([ph,e],C(ph,h'))
\end{align*}
This comes from the use of right group actions (of $H$ on $\manP$ and of $G$ on $\manQ$) to define right action groupoids $\manP\rtimes H$ and $\manQ\rtimes G$. The former involves inverses in order to define the left action $h\cdotact p\defeq ph\inv$ of $H$ on $\manP$. However, most of the literature on Morita equivalence for groupoids and Hilsum-Skandalis bibundles uses left Lie groupoids. This justifies the choice of the current cumbersome convention.
\end{remark}
\begin{remark}\label{remark C Gpd morphism}
Note also that in this context, a cocycle $C$ is seen as a groupoid morphism $(\manP\rtimes H)\to(G\rra\{*\})$. Furthermore, any groupoid morphism $(\manP\rtimes H)\to(G\rra\{*\})$ defines a cocycle since the compatibility condition for the partial multiplication induces the cocycle relation.
\end{remark}
One can recover the isomorphism between $C$-equivariant maps on $\manP$ and equivariant maps on $\manQ$ as a commutation relation between equivariant sections of representation bundles of $(H\ltimes \manP)$, $(G\ltimes \manQ)$ and an equivariant map induced by $\phi$ between the representations.

Consider\footnote{Recall that $(V,\rho)$ is a representation of the Lie group $G$.} $E_\manP\defeq\manP\times V$ together with the map $\prp:E_\manP\to\manP,\ (p,v)\mapsto p$. This is a trivial vector bundle over $\manP$. One defines a groupoid action of $H\ltimes\manP$ on $E_\manP$ by $(h,p)\cdotact (p_0,v)=(ph\inv,\rho(C(h\cdotact p,h))v)$ for any $((h,p),(p_0,v))\in (H\ltimes \manP)*E_\manP=\{((h,p),(p_0,v))\in (H\ltimes \manP)\times E_\manP,\ \alpha(h,p)=\prp(p_0,v)\}$. The compatibility condition rewrites $p=p_0$ and one checks that $\prp((h,p)\cdotact(p,v))=h\cdotact p=\beta(h,p)$, $(e,p)\cdotact(p,v)=(p,v)$ and $(h',h\cdotact p)\cdotact\bigl((h,p)\cdotact(p,v)\bigr)=(h'h,p)\cdotact(p,v)$ for any $(h,h',p,v)\in H\times H\times\manP\times V$.

Equivalently, one defines $E_\manQ\defeq \manQ\times V$ together with the projection $\prq:E_\manQ\to\manQ,\ (q,v)\mapsto q$ and the groupoid action of $G\ltimes \manQ$ on $E_\manQ$ given by $(g,q)\cdotact(q,v)\defeq (g\cdotact q,\rho(g)v)$.

Let's consider the map $\psi:E_\manP\to E_\manQ,\ (p,v)\mapsto ([p,e],v)$. Note that $\psi(p,v)=(\phi_0(p),v)$ for all $(p,v)\in E_\manP$. One checks that $\prq\circ\psi=\phi_0\circ\prp$ and $\psi((h,p)\cdotact(p,v))=\phi_1(h,p)\cdotact\psi((p,v))$ for any $((h,p),(p,v))\in(H\ltimes\manP)*E_\manP$. This $\psi$ is $\phi_1$-equivariant.
\begin{proposition}
The maps $(\phi_1,\phi_0)$ and $\psi$ provide a bijection between $(H\ltimes\manP)$-equivariant sections of $E_\manP$ and $(G\ltimes\manQ)$-equivariant sections of $E_\manQ$.
\end{proposition}

\begin{proof}
The proof is based on the following diagram:
\begin{equation}\label{diag:Corresp equiv sec}
\begin{tikzcd}
E_\manP\arrow[rr,"\psi"] & & E_\manQ \\
 & & \\
\manP\arrow[rr,"\phi_0"]\arrow[uu, "f^\manP"] & & \manQ\arrow[uu, "f^\manQ"]
\end{tikzcd}
\end{equation}
The commutation of this diagram induces a 1:1 correspondence between $(G\ltimes\manQ)$-equivariant sections $f^\manQ:\manQ\to E_\manQ$ and $(H\ltimes\manP)$-equivariant sections $f^\manP:\manP\to E_\manP$. Indeed, two corresponding maps satisfy
\begin{align}\label{eq link equiv sections}
\psi\circ f^\manP=f^\manQ\circ\phi_0.
\end{align}
In addition, these maps decompose as $f^\manP=(\id_\manP,f^\manP_2),\ f^\manP_2:\manP\to V$ and $f^\manQ=(\id_\manQ,f^\manQ_2),\ f^\manQ_2:\manQ\to V$. The equivariance of $f^\manP$ and $f^\manQ$ translates into
\begin{align*}
f^\manP_2(h\cdotact p)&=\rho(C(h\cdotact p,h))f^\manP_2(p)\\
f^\manQ_2(g\cdotact q)&=\rho(g)f^\manQ_2(q)
\end{align*}
for any $(h,p)\in H\times \manP$ and $(g,q)\in G\times \manQ$. For any $p\in\manP$, equation \eqref{eq link equiv sections} reads
\begin{align*}
\psi(p,f^\manP_2(p)) =\left(\phi_0(p),f^\manQ_2(\phi_0(p))\right)\hspace{1cm}\Leftrightarrow\hspace{1cm} \left(\phi_0(p),f^\manP_2(p)\right) =\left(\phi_0(p),f^\manQ_2(\phi_0(p))\right).
\end{align*}
Then the partial maps $f^\manQ_2$ and $f^\manP_2$ must satisfy $f^\manP_2=f^\manQ_2\circ\phi_0=\phi_0^*f^\manQ_2$. One checks that $\phi^*:\calF_\eqv(\manQ,V)\to\calF_\Ceq(\manP,V)$ is an isomorphism whose inverse is $f^\manP_2\mapsto \bigl([p,g]\mapsto \rho(g\inv) f^\manP_2(p)\bigr)$. The $C$-equivariance of $f^\manP_2$ ensures that the latter definition is independent of the representative $(p,g)\in\pi_H\inv([p,g])$.

Note that the correspondence between $f^\manP_2$ and $f^\manQ_2$ is also a consequence of Corollary~\ref{corollary isom eq functions}.
\end{proof}

\begin{remark}
Note that the bijection between equivariant maps comes from the existence of a particular groupoid morphism $\phi:H\ltimes\manP\to G\ltimes \manQ$ and an equivariant map $\psi:E_\manP\to E_\manQ$. This is not a direct consequence of the Morita equivalence. In addition, the generalization of this bijection to tensorial forms is not straightforward and the tools needed exceed $\phi$ and $\psi$.
\end{remark}

%%%%%%%%%%%%%%%%
\section{Concerning Cocycles}
\label{sec Cocycles}
%%%%%%%%%%%%%%%%

%%%%%%%%%%%%%%%%
\subsection{Group Cohomology}
\label{subsec Group Cohomology}
%%%%%%%%%%%%%%%%

The cocycle $C$ is closely related to an instance of non-Abelian group cohomology. Let us recall that given a group $H$ and a $H$-module $\M$, one can define the group cohomology of $H$ with coefficients in $\M$ in the following way. The space of $n$-cochains is $\calF(H^n,\M)$, the space of maps from $H^n$ to $\M$, the differential $\delta^n:\calF(H^n,\M)\rightarrow\calF(H^{n+1},\M)$ is defined as
\begin{align*}
(\delta^n f)(h_1,\ldots,h_{n+1})
= 
%\begin{multlined}[t]
a_{h_1} \cdotact f(h_2,\ldots,h_{n+1}) + \sum_{i=1}^n(-1)^i f(h_1,\ldots,h_{i-1},h_ih_{i+1},h_{i+2}, \ldots, h_{n+1})&
\\
+ (-1)^{n+1}f(h_1,\ldots,h_n)&,
%\end{multlined}
\end{align*}
where $a_h$ is the action of $H$ on $\M$. One defines $n$-cocycles as elements of $\ker \delta^n$ and $n$-coboundaries as elements of $\im \delta^{n-1}$. Since $\delta^{n+1}\circ\delta^n=0$, $\im\delta^{n-1}\subset\ker\delta^n$, which enables defining the equivalence relation “up to a coboundary”: $f\sim f+\delta^{n-1} g,\ \forall f\in\ker\delta^n,\ \forall g \in\calF(H^{n-1},\M)$. The spaces of equivalence classes under this relation are the cohomology groups\footnote{Note that this notion of cohomology corresponds to the singular cohomology of the Eilenberg-Maclane space $K(H,1)$ with coefficients in $\M$, see \cite{masson_introduction_2008}} of $H$ with coefficients in $\M$. It is worth giving an explicit expression for $\delta^0$ and $\delta^1$:
\begin{align*}
\begin{aligned}
(\delta^0 f)(h_1)
&= a_{h_1}\cdotact f - f,
&
(\delta^1 f)(h_1,h_2)
&= a_{h_1}\cdotact f(h_2)- f(h_1h_2)+f(h_1).
\end{aligned}
\end{align*}

The group cohomology admits a generalization when $\M$ is a (not necessarily Abelian) group supporting an action of $H$. Let us define a 1-cocycle as a map $\widehat{C}:H\rightarrow \M$ satisfying\footnote{The group action in $\M$ is written with the multiplicative convention here, contrary to the additive convention in the Abelian case.}
\begin{align}
\label{eq cocycle relation NAGC}
\widehat{C}(h_1h_2)
= \widehat{C}(h_1)(a_{h_1}\cdotact\widehat{C}(h_2)).
\end{align}
Note that this is equivalent to $(a_{h_1}\cdotact\widehat{C}(h_2))\widehat{C}(h_1h_2)\inv\widehat{C}(h_1)=e$, which corresponds to $(\delta^1\widehat{C})(h_1,h_2)=0$ in the Abelian case, hence the name “1-cocycle”. Let us define an equivalence relation on cocycles by
\begin{align}
\label{eq cocycle equivalence NAGC}
\widehat{C}\sim\widehat{C}' 
\Longleftrightarrow 
\exists f\in\M,\ \forall h\in H,\ \widehat{C}'(h) = f\inv\widehat{C}(h)(a_h\cdotact f).
\end{align}
Once again, this reduces to the addition of a coboundary in the Abelian case. This defines the non-Abelian group cohomology of $H$ with coefficients in $\M$. Notice that this group cohomology exists only for $n=0,1$.

Let us consider the specific case where $\M=\calF(\manP,G)$. This is an (\textit{a priori} non-Abelian) group whose product is inherited from the product on $G$ by pointwise multiplication. Also, $\calF(\manP,G)$ supports a right action of $H$ induced by the right action of $H$ on $\calP$: $a_h\cdotact f=\calR_h^{\manP*}f$, \textit{i.e.} $(a_h\cdotact f)(p)=f(ph)$. Then, any 1-cocycle $\widehat{C}\in\calF(H,\calF(\manP,G))$ is equivalent to a map $C\in\calF(\manP\times H,G)$ satisfying $C(p,h_1h_2)=C(p,h_1)C(ph_1,h_2)$ (as a consequence of \eqref{eq cocycle relation NAGC}), with the correspondence given by $C(p,h)=(\widehat{C}(h))(p)$. In this sense, the cocycle defined in Sec.~\ref{subsec Twisted Gauge Fields}, or equivalently the groupoid morphism defined in Remark~\ref{remark C Gpd morphism}, is a 1-cocycle in the non-Abelian group cohomology of $H$ with coefficients in $\calF(\manP,G)$. This approach provides a natural notion of equivalence of cocycles inherited from \eqref{eq cocycle equivalence NAGC}:
\begin{align*}
C\sim C'
\Longleftrightarrow 
\exists f\in\calF(\manP,G),\ \forall (p,h)\in \manP\times H,\ C'(p,h) = f(p)\inv C(p,h)f(ph).
\end{align*}
One checks trivially that $C'$ satisfies the cocycle relation if $C$ does. In the following, given a cocycle $C$ and a map $f\in\calF(\manP,G)$, the equivalent cocycle mentioned above is denoted $C^f:(p,h)\mapsto f(p)\inv C(p,h)f(ph)$.

%%%%%%%%%%%%%%%%
\subsection{Twisted Associated Bundles with Different Cocycles}
\label{subsec Isom Changing Cocycle}
%%%%%%%%%%%%%%%%

This section points out several results on associated bundles through related cocycles. First is shown the role of equivalent cocycles (up to a coboundary).

\begin{proposition}\label{prop isom associated up to coboundary}
Let $f\in\calF(\manP,G)$ and $C,C^f$ be two equivalent cocycles. Then there exists an isomorphism of associated bundles:
\begin{align*}
\manP \times_{C^f}V \simeq \manP \times_C V.
\end{align*}
\end{proposition}

\begin{proof}
Let $\eta^f:\manP \times V\rightarrow\manP \times V$ defined by $\eta^f:(p,v)\mapsto(p,f(p)\cdotact v)$ (where the representation $\rho$ of $G$ on $V$ is implied). Let $\calR_h^C:\manP \times V\rightarrow\manP \times V$ be the right action of $H$ on $\manP\times V$ associated with $C$, \textit{i.e.} $\calR_h^C:(p,v)\rightarrow(ph,C(p,h)\inv\cdotact v)$. Identically, $\calR_h^{C^f}(p,v) \defeq (ph,C^f(p,h)\inv\cdotact v)$. The following diagram is commutative:
\begin{equation*}
\begin{tikzcd}
\manP\times V\arrow[rr,"\calR_h^{C^f}"]\arrow[dd,"\eta^f"]& &\manP\times V\arrow[dd,"\eta^f"]\\
& & \\
\manP\times V\arrow[rr,"\calR_h^C"]& &\manP\times V
\end{tikzcd}
\end{equation*}
On the one hand,
\begin{align*}
(p,v) \xmapsto{\eta^f} 
(p,f(p)\cdotact v)
\xmapsto{\calR_h^C} 
(ph,C(p,h)\inv f(p)\cdotact v),
\end{align*}
and, on the other hand,
\begin{align*}
(p,v)
\xmapsto{\calR_h^{C^f}} 
(ph,C^f(p,h)\inv\cdotact v)
= (ph,f(ph)\inv C(p,h)\inv f(p)\cdotact v)
\xmapsto{\eta^f }
(ph,C(p,h)\inv f(p)\cdotact v).
\end{align*}
As a consequence, $\eta^f$ induces an isomorphism $[\eta^f]:\manP\times_{C^f}V\rightarrow\manP\times_C V$ with inverse $[\eta^{f\inv}]$.
\end{proof}

\begin{remark}
This induces a map from elements of the cohomology $H^1(H,\calF(\manP,G))$ to twisted associated bundles (up to an isomorphism). Note that this map may not be an injection. In particular, depending on the space $V$, several non-equivalent cocycles may induce isomorphic associated bundles.
\end{remark}

\begin{corollary}
Any coboundary $B(p,h)=f(p)\inv f(ph)$ for a given $f:\manP\rightarrow G$ induces trivial twisted associated bundles $\manP\times_B V\simeq M\times V$.
\end{corollary}
\begin{proof}
The proof relies on the fact that $\manP\times_C V\simeq M\times V$ for $C=e$.
\end{proof}
Another relation comes from precomposing a cocycle with an automorphism of $\manP$.

\begin{definition}
Let $C$ be a cocycle and $\Psi\in\Aut(\manP)$. Let us use the notation $C^\Psi:(p,h)\mapsto C(\Psi(p),h)$. It is straightforward to check that $C^\Psi$ is a cocycle. In the particular case of a gauge transformation $\Phi\in\Autv(\manP)$, $\Phi$ reads $\Phi(p)=p\gamma(p)$ for a unique $\gamma\in\calH^\manP$ and the notation $C^\gamma:(p,h)\mapsto C(p\gamma(p),h)$ is adopted, following the convention adopted in Sec.~\ref{subsec Twisted Gauge Fields}.
\end{definition}

\begin{proposition}\label{prop isom associated bdls after Aut}
Let $\Psi\in\Aut(\manP)$ and $C$ a cocycle. There exists an isomorphism of associated bundles:
\begin{align*}
\manP \times_{C^\Psi} V \simeq \manP \times_C V.
\end{align*}
\end{proposition}

\begin{proof}
Let $\eta^\Psi:\manP \times V\rightarrow\manP \times V$ defined by $\eta^\Psi:(p,v)\mapsto(\Psi(p),v)$. As defined before, $\calR_h^C(p,v) \defeq (ph,C(p,h)\inv\cdotact v)$ and $\calR_h^{C^\Psi}(p,v) \defeq (ph,C^\Psi(p,h)\inv\cdotact v)$. The following diagram commutes:
\begin{equation*}
\begin{tikzcd}
\manP\times V\arrow[rr,"\calR_h^{C^\Psi}"]\arrow[dd,"\eta^\Psi"]& &\manP\times V\arrow[dd,"\eta^\Psi"]\\
& & \\
\manP\times V\arrow[rr,"\calR_h^C"]& &\manP\times V
\end{tikzcd}
\end{equation*}
On the one hand
\begin{align*}
(p,v)
\xmapsto{\eta^\Psi} 
(\Psi(p),v)
\xmapsto{\calR_h^C} 
(\Psi(p)h,C(\Psi(p),h)\inv\cdotact v)
= (\Psi(ph),C^\Psi(p,h)\inv\cdotact v).
\end{align*}
On the other hand
\begin{align*}
(p,v)
\xmapsto{\calR_h^{C^\Psi}} 
(ph,C^\Psi(p,h)\inv\cdotact v)
\xmapsto{\eta^\Psi }
(\Psi(ph), C^\Psi(p,h)\inv\cdotact v).
\end{align*}
As a consequence, $\eta^\Psi$ induces an isomorphism $[\eta^\Psi]:\manP\times_{C^\Psi}V\rightarrow\manP\times_C V$ with inverse $[\eta^{\Psi\inv}]$.
\end{proof}

This result is of particular importance in the case of gauge transformations. As stated in Sec.~\ref{subsec Twisted Gauge Fields}, a $C(H)$-equivariant map $\varphi\in\calF_\eqv(\manP,V)$ transforms under a gauge transformation $\gamma\in\calH^\manP$ into a $C^\gamma(H)$-equivariant map $\varphi^\gamma$. These maps are in 1:1 correspondence with sections of different twisted associated bundles, namely $\phi\in\Gamma(\manP\times_C V)$ and $\phi^\gamma\in\Gamma(\manP\times_{C^\gamma} V)$. In other words, performing a gauge transformation at the level of $\manP$ leads to gauge-transformed equivariant fields which correspond to section of a different associated bundle, namely, $\manP\times_{C^\gamma} V$. This should be interpreted as a change of nature of the associated fields under gauge transformation. However, as a consequence of Prop.~\ref{prop isom associated bdls after Aut}, $\manP\times_C V\simeq\manP\times_{C^\gamma} V$. This may explain why this feature has gone unnoticed so far.

The previous propositions apply in the study of $\manQ=\manP\times_CG$ as well.
\begin{corollary}\label{corollary eq C implies isom Q}
Let $f\in\calF(\manP,G)$, $\Psi\in\Aut(\manP)$ and $\gamma\in\calH^\manP$.
\begin{enumerate}
\item $C$ and $C^f$ induce isomorphic bundles $\calQ^C=\manP\times_CG$ and $\calQ^{C^f}=\manP\times_{C^f}G$. In particular, $B(p,h)=f(p)\inv f(ph)$ induces a trivial bundle $\calQ^B\simeq M\times G$.
\item $C$ and $C^\Psi$ induce isomorphic bundles $\calQ^C=\manP\times_CG$ and $\calQ^{C^\Psi}=\manP\times_{C^\Psi}G$. In particular, if $\Phi:p\mapsto p\gamma(p)\in\Autv(\manP)$, then $\calQ^C\simeq\calQ^{C^\gamma}$.
\end{enumerate}
\end{corollary}
In the next section, a similar result about $\manS$ is proved, which is even stronger.

%%%%%%%%%%%%%%%%
\subsection{\texorpdfstring{Link Between the Cohomology of $C$ and the Geometry of $\manS$}{Link Between the Cohomology of C and the Geometry of S}}
\label{subsec Interpretation Comments Again}
%%%%%%%%%%%%%%%%

Let us return to the geometric interpretation as stated in Sec.~\ref{subsec Interpretation Comments}. In order to relate $\manP$ to $\manS$, a convenient (though arbitrary) choice of inclusion was made by using $i_e$. It was shown that this inclusion gave an effective correspondence between $G$-equivariant $H$-invariant fields\footnote{If $k$-forms are considered, one should add the condition of being horizontal or of satisfying the relevant vertical normalization.} $\varphi^\manS$ on $\manS$ and twisted fields $\varphi^\manP$ on $\manP$. The equivariance of $i_e^*\varphi^\manS$ is recalled here:
\begin{align*}
(i_e^*\varphi^\manS)(ph)
&= \varphi^\manS(ph,e)
= \varphi^\manS(p,C(p,h))
= C(p,h)\inv\cdotact\varphi^\manS(p,e)
= C(p,h)\inv\cdotact(i_e^*\varphi^\manS)(p).
\end{align*}
It is worth paying attention to the pullbacks of $\varphi^\manS$ along other inclusion maps. For example, for $g_0\in G$, $i_{g_0}$ provides
\begin{align*}
(i_{g_0}^*\varphi^\manS)(ph)
&= \varphi^\manS(ph,g_0)
= \varphi^\manS(p,C(p,h)g_0)
= \varphi^\manS(p,g_0g_0\inv C(p,h)g_0)
= (g_0\inv C(p,h)g_0)\inv\cdotact\varphi^\manS(p,g_0)
\\
&= C^{g_0}(p,h)\inv\cdotact(i_{g_0}^*\varphi^\manS)(p),
\end{align*}
where $C^{g_0}:(p,h)\mapsto g_0\inv C(p,h)g_0$ is another cocycle. In general, let us define $i_f:\manP\rightarrow\manS$, $p\mapsto(p,f(p))$ for a given $f\in\calF(\manP,G)$, then
\begin{align*}
(i_f^*\varphi^\manS)(ph)
&= \varphi^\manS(ph,f(ph))
= \varphi^\manS(p,C(p,h)f(ph))
= \varphi^\manS(p,f(p)f(p)\inv C(p,h)f(ph))
\\
&= (f(p)\inv C(p,h)f(ph))\inv\cdotact\varphi^\manS(p,f(p))
= C^f(p,h)\inv\cdotact(i_f^*\varphi^\manS)(p).
\end{align*}
This shows that changing the embedding of $\manP$ into $\manS$ affects the equivariance of induced twisted fields by changing the cocycle within the same equivalence class.

Let us recall that $\calF(\manP,G)\simeq\Gamma(\manS(\manP,G))$ where the isomorphism is given by the correspondence $f\in\calF(\manP,G)\leftrightarrow i_f\in\Gamma(\manS(\manP,G))$. One trivially checks that $i_f$ defined above satisfies $\pi_G\circ i_f=\id_\manP$. This motivates the reinterpretation of the group cohomology introduced in Sec.~\ref{subsec Group Cohomology}. The 1-cocycle $\widehat{C}$ is an element of\footnote{Here, because $\manS=\manP\times G$, its space of sections inherits a group operation induced by the multiplication in $G$: $(i_fi_{f'})(p) \defeq (p,f(p)f'(p))$. In this way $\calF(\manP,G)\simeq\Gamma(\manS(\manP,G))$ is a group isomorphism.} $H^1(H,\calF(\manP,G))\simeq H^1(H,\Gamma(\manS(\manP,G)))$. Changing the cocycle in the same equivalence class corresponds to a change of section of the same bundle $\manS(\manP,G)$. On the contrary, moving from one cocycle to another in a different equivalence class implies changing the right action $\calR_{(h,e)}^\manS$ of $H$ on $\manS$. This may modify the structure and underlying topology of $\manS$. In other words, the elements of $H^1(H,\Gamma(\manS(\manP,G)))$ label the set of actions of $H$ on $\manS(\manP,G)$ that promote it to non-isomorphic bundles $\manS(M,H\times G)$. This is formalized in the following theorem.

\begin{theorem}\label{theorem eq C equiv isom S}
Let $C_1,C_2\in\calF(\manP\times H,G)$ be two cocycles. Let $\manS_1 \defeq \manP\times G \rdefeq \manS_2$ be two principal $(H\times G)$-bundles supporting the following right actions of $H\times G$ respectively:
\begin{align*}
\calR_{(h,g')}^i(p,g) 
&\defeq (ph,C_i(p,h)\inv gg'),
\hspace{1cm} i=1,2.
\end{align*}
Then $\manS_1$ and $\manS_2$ are isomorphic bundles if and only if $C_1$ and $C_2$ are equivalent cocycles up to an automorphism of $\manP$.
\begin{align*}
\exists & \Phi:\manS_1\xrightarrow{\simeq} \manS_2,\ \forall (h,g')\in H\times G,\ \Phi\circ\calR_{(h,g')}^1=\calR_{(h,g')}^2\circ\Phi
\\
\Longleftrightarrow \exists & (f,\Psi)\in\calF(\manP,G)\times\Aut(\manP),\ C_1(p,h)
= (C_2^f)^\Psi(p,h)=f(\Psi(p))\inv C_2(\Psi(p),h)f(\Psi(p)h).
\end{align*}
\end{theorem}

\begin{remark}
In practice, one can consider an automorphism $\Psi_\manS$ of, say, $\manS_2$ and work with $\tilde{\manS}_2=\Psi_\manS(\manS_2)$. By definition, $\manS_2\simeq\tilde{\manS}_2$. One can choose $\Psi_\manS$ such that it compensates $\Psi\in\Aut(\manP)$ introduced in Theorem~\ref{theorem eq C equiv isom S}. This enables working with ``$C_1=\tilde{C}_2^f$'' effectively. In other words, $\Phi$ consist in changing the action of $H$ on $\manS$ without moving points in $\manP$.
\end{remark}

\begin{proof}
Assume there exists a bundle isomorphism $\Phi:\manS_1\xrightarrow{\simeq}\manS_2$. For any $(p,g)\in \manS_1$ and for any $(h,g')\in H\times G$, one has
\begin{align*}
\Phi(\calR_{(h,g')}^1(p,g))
= \calR_{(h,g')}^2(\Phi(p,g))
\end{align*}
Let us decompose $\Phi=(\Phi_\manP,\Phi_G)$ with $\Phi_\manP:\manS_1\rightarrow \manP$ and $\Phi_G:\manS_1\rightarrow G$. Then the previous equation splits into
\begin{align}
\label{eq system}
\left\{ 
\begin{aligned}
\Phi_\manP(ph,C_1(p,h)\inv g g') &= \Phi_\manP(p,g)h \\
\Phi_G(ph,C_1(p,h)\inv g g') &= C_2(\Phi_\manP(p,g),h)\inv\Phi_G(p,g)g'.
\end{aligned}
\right.
\end{align}
At $g'=g\inv C_1(p,h) g$ and $g'=g\inv C_1(p,h)$, the first equation of \eqref{eq system} reads
\begin{align*}
\begin{aligned}
\Phi_\manP(ph,g)
&= \Phi_\manP(p,g)h,
&
\Phi_\manP(ph,e)
&= \Phi_\manP(p,g)h,
\end{aligned}
\end{align*}
which proves that $\Phi_\manP$ is independent of its second argument and $\Phi_\manP(\cdot,e)\in\Aut(\manP)$. At $h=e$, the second equation in \eqref{eq system} reads
\begin{align*}
\Phi_G(p,gg')
= \Phi_G(p,g)g' 
\Longrightarrow \Phi_G(p,g)
= \Phi_G(p,e)g.
\end{align*}
And, at $g'=e$, the same equation reads
\begin{align*}
\Phi_G(ph,C_1(p,h)\inv g)
&= C_2(\Phi_\manP(p,g),h)\inv\Phi_G(p,g)
\\
\Phi_G(ph,e)C_1(p,h)\inv g
&= C_2(\Phi_\manP(p,e),h)\inv\Phi_G(p,e)g
\end{align*}
which leads to
\begin{align*}
C_1(p,h)
&= \Phi_G(p,e)\inv C_2(\Phi_\manP(p,e),h)\Phi_G(ph,e)
%\\
%&
= f(\Psi(p))\inv C_2(\Psi(p),h)f(\Psi(ph))
\end{align*}
where $\Psi:p\mapsto\Phi_\manP(p,e)$ is an automorphism of $\manP$ and $f:p\mapsto \Phi_G(\Psi\inv(p),e)\in\calF(\manP,G)$.

Conversely, if $C_1=(C_2^f)^\Psi$, then the isomorphisms $\eta^f:(p,g)\mapsto(p,f(p)g)$ and $\eta^\Psi:(p,g)\mapsto(\Psi(p),g)$ induce an isomorphism $\eta^f\circ\eta^\Psi:\manS_1\rightarrow\manS_2$. Consider the following diagram
\begin{equation*}
\begin{tikzcd}
\manS_1=\manP\times G\arrow[rr,"\calR_{(h,g')}^1"]\arrow[dd,"\eta^f\circ\eta^\Psi"]& &\manS_1=\manP\times G\arrow[dd,"\eta^f\circ\eta^\Psi"]\\
& & \\
\manS_2=\manP\times G\arrow[rr,"\calR_{(h,g')}^2"]& &\manS_2=\manP\times G
\end{tikzcd}
\end{equation*}
This diagram is a commutative diagram. Indeed, on the one hand,
\begin{align*}
(p,g)\xmapsto{\calR_{(h,g')}^1}(ph,C_1(p,h)\inv gg')
\xmapsto{\eta^f\circ\eta^\Psi} (\Psi(ph),f(\Psi(ph))C_1(p,h)\inv gg')
\end{align*}
and, on the other hand,
\begin{align*}
(p,g)
\xmapsto{\eta^f\circ\eta^\Psi} 
(\Psi(p),f(\Psi(p))g)
\xmapsto{\calR_{(h,g')}^2}
&\ (\Psi(p)h,C_2(\Psi(p),h)\inv f(\Psi(p))gg')
\\
&= (\Psi(ph),f(\Psi(ph))f(\Psi(p))\inv C_2(\Psi(p),h)\inv f(\Psi(p))gg')
\\
&= (\Psi(ph),f(\Psi(ph)) (C_2^f)^\Psi(p,h)\inv gg')
\\
&= (\Psi(ph),f(\Psi(ph)) C_1(p,h)\inv gg')
\end{align*}
The commutation $\calR_{(h,g')}^1\circ(\eta^f\circ\eta^\Psi)=(\eta^f\circ\eta^\Psi)\circ\calR_{(h,g')}^2$ implies that $\eta^f\circ\eta^\Psi$ is a principal bundle isomorphism, which concludes the proof.
\end{proof}
\begin{remark}
The results obtained in Prop.~\ref{prop isom tensorial forms}, \ref{prop isom connections PQ}, \ref{prop isom GT PQ}, \ref{prop isom twisted connections SP} and \ref{prop isom twisted GT SP} use the map $i_e$. Similar results can be obtained using a generic embedding $i_f$ of $\manP$ in $\manS$ instead of $i_e$. This provides
\begin{itemize}
\item an isomorphism $i_f^*\pi_H^*$ between standard fields on $\manQ$ and twisted fields on $\manP$,
\item an isomorphism $i_f^*$ between $G$-equivariant $H$-invariant fields on $\manS$ and twisted fields on $\manP$.
\end{itemize}
Provided that the right action $\calR^\manS$ – defining $\manS$ as a principal fiber bundle – involves $C$, then $i_f^*\pi_H^*$ (resp. $i_f^*$) maps standard fields on $\manQ$ (resp. $\manS$) to $C^f(H)$-equivariant fields on $\manP$. In other words, the twist on $\manP$ comes from both the right action $\calR^\manS$ on $\manS$ and the canonical embedding $i_f$ of $\manP$ in $\manS$.
\end{remark}
\begin{remark}
The correspondence between equivalence classes of cocycles $C$ and equivalence classes of bundles $\manS$ up to an isomorphism covering the identity map on $\manP$ enables relating the cohomology $H^1(H,\calF(\manP,G))$ to geometric structures such as characteristic classes and classifying bundles.
\end{remark}

%%%%%%%%%%%%%%%%
\section{Examples of cocycles}
\label{sec Examples of cocycles}
%%%%%%%%%%%%%%%%

%%%%%%%%%%%%%%%%
\subsection{Trivial Cocycle}
\label{subsec Trivial Cocycle}
%%%%%%%%%%%%%%%%

Consider the case where $C(p,h)=e$ for any $(p,h)\in\manP\times H$. Then, twisted fields on $\manP$ are $H$-invariant: $\varphi(ph)=C(p,h)\inv\cdotact\varphi(p)=\varphi(p)$. Furthermore, twisted connections satisfy the vertical normalization $\omega_p(\ka_p^\ver)=\dr_H C_{|(p,e)}(\ka)=0,\,\forall \ka\in\kh$.

The correspondence space of this example is $\manS(M,H\times G)=\manP\times G$ together with the right action $\calR_{(h,g')}^\manS(p,g)=(ph,gg')=\widehat{\calR}_h\circ\calR_{(e,g')}^\manS(p,g)$. This makes $\manS$ trivial along $G$ both as a manifold and as a principal bundle. The quotient $\manS/H$ induced by $\calR_{(h,e)}^\manS$ is the manifold $\manQ\simeq M\times G$ with points $q=[p,g]=[ph,g]\xleftrightarrow{1:1}(x,g)=(\pi_\manP(p),g)$. According to Sec.~\ref{sec Isomorphism}, twisted fields on $\manP$ correspond (via pullback) to standard fields on $\manQ$, \textit{i.e.} fields on $M\times G$. This is consistent with the fact that $H$-invariant fields on $\manP$ (or $H$-basic forms) induce well-defined fields on $M$.

Note that $C(p,h)=e$ is equivalent to $C^f(p,h)=f(p)\inv f(ph)$ for any $f\in\calF(\manP,G)$. Thus, any twisted gauge field on $\manP$ through $C^f$ corresponds to a $H$-invariant field, and then, to a field on $M$. The case of gauge-invariant fields is of interest since it may correspond to dressed fields after a global application of the DFM.

%%%%%%%%%%%%%%%%
\subsection{Identity Map}
\label{subsec Identity Map}
%%%%%%%%%%%%%%%%

Consider the case where $G=H$ and $C(p,h)=h,\,\forall (p,h)\in\manP\times H$. In this case, twisted fields are genuine fields satisfying a standard $H$-equivariance condition: $\varphi(ph)=C(p,h)\inv\cdotact\varphi(p)=h\inv\cdotact\varphi(p)$. Similarly, any $C(H)$-tensorial form is a standard tensorial form and any twisted connection is standard.

The formalism developed previously becomes $\calS(M,H\times H)=\calP\times H$ as correspondence space, together with the right action $\calR_{(h_1,h_2)}(p,h)=(ph_1,h_1\inv hh_2)$. There are two ways of interpreting this bundle. Either the two copies of $H$ (denoted $H_1,\,H_2$ here) can be thought as encoding physically different degrees of freedom. In particular, one can generate associated vector bundles to $\manS$ through representations of either $H_1$ or $H_2$ and interpret their sections as matter fields with different gauge degrees of freedom. Or these two copies of $H$ can be considered as a redundancy. Hence, there are two ways of canceling this redundancy: either by considering equivalence classes on $\manS$ generated by $\calR_{(e,h_2)}^\manS$, thus recovering $\manP$, or equivalence classes generated by $\calR_{(h_1,e)}^\manS$, thus generating $\manQ$. These two quotients are isomorphic since any principal bundle is associated to itself: $\manP\simeq\manP\times_HH=\{[p,h_0]=[ph,h\inv h_0]\}=\manQ$.

$C(p,h)=h$ is equivalent to $C^f(p,h)=f(p)\inv hf(ph)$, for any $f:\manP\rightarrow H$. Note that this cocycle measures the difference between $f(ph)$ and $h\inv f(p)$. Let us assume that there exists a map $f$ such that $f(ph)=h\inv f(p)$, \textit{i.e.} $f$ has equivariance property $(\calR_h^{\manP*}f)(p)=\Lr_{h\inv}f(p)$. Then $f$ is a global $H$-dressing field. Consequently, $f$ induces a global section of the associated bundle $\manP\times_HH\simeq\manP$ which reads $x\mapsto [p,f(p)]=[ph,h\inv f(p)]=[ph,f(ph)]$. Hence $\manP$ is trivial. This is consistent with the previous example.

%%%%%%%%%%%%%%%%
\subsection{Group Homomorphism}
\label{subsec Group Homomorphism}
%%%%%%%%%%%%%%%%

Consider a cocycle $C$ which is independent of $\manP$. In other words, $C$ is independent of its first variable (which is omitted in the notations in this section). Consequently, $C:H\rightarrow G$ satisfies the cocycle relation $C(hh')=C(h)C(h')$, \textit{i.e.} it is a group homomorphism. (Note that conversely, any homomorphism $H\rightarrow G$ can be interpreted as a cocycle with trivial dependence in $p\in\manP$.) In this context, twisted fields satisfy the equivariance relation $\varphi(ph)=\rho_G(C(h))\inv\cdotact\varphi(p)$. This is characteristic of a standard $H$-gauge field valued in the $H$-space $(V,\rho_G\circ C)$.

The correspondence space of this example is $\manS(M,H\times G)=\manP\times G$ together with the right action $\calR_{(h,g')}^\manS(p,g)=(ph,C(h)\inv gg')$. Similarly to the general case, the structure group acts on a $G$-space $(V,\rho_G)$ as $\rho_{H\times G}(h,g')\cdotact v \defeq \rho_G(g')\cdotact v$, and the action of $H$ on $V$ is trivial. Note, however, that one can define the action $\rho_H=\rho_G\circ C$ of $H$ on $V$ and induce a new action of the structure group $H$ on $V$ as $\tilde{\rho}_{H\times G}(h,g')\cdotact v \defeq \rho_H(h)\cdotact v=\rho_G(C(h))\cdotact v$. This new action should be identified with the $(H\times G)$-action induced on $H$-spaces $W$ by the standard action of $H$.

In this example, $\manQ \defeq \manP\times_{C(H)}G$ and the standard $H$-fields on $\manP$ with representation $(V,\rho_G\circ C)$ correspond to standard $G$-fields with representation $(V,\rho_G)$ on $\manQ$. This example offers an intuition for the more general case. Let $C(H) \defeq \{C(h),\ h\in H\}\subset G$. Because $C$ is an homomorphism, $C(H)$ is a subgroup of $G$. Note that $C(H)$ may be a proper subgroup. At first glance, the correspondence between twisted fields on $\manP$ and standards fields on $\manQ$ may appear counter-intuitive. Indeed, $G$-fields have internal degrees of freedom valued in a larger group. However, the $G$-equivariance satisfied by these fields also accounts for more constraints. Altogether, both kind of fields carry the same information, \textit{i.e.} the information of a field on the base manifold.

%%%%%%%%%%%%%%%%
\subsection{Conformal Cartan Geometry}
\label{subsec Conformal Cartan Geometry}
%%%%%%%%%%%%%%%%

Let $\Pz$ be a principal bundle over a smooth manifold $M$ with structure group\footnote{The notation $H\ltimes K$ refers to semidirect groups and should not get mixed up with the action groupoid notation $H\ltimes \manP$ used in Sec.~\ref{sec Groupoids}.} $G=H\ltimes K$, where both $H$ and $K$ are Lie groups. Since $G=H\ltimes K$, one has that $K$ is a normal subgroup of $G$ and $H\simeq G/K$ is embedded in $G$. Thus $\manP=\Pz/K$ is endowed with an action of $H$ making it a principal $H$-bundle over $M$. The projection on orbits of the action of $K$ is denoted $\pi_K:\Pz\to\manP$.

Assume that there exists a partial dressing field $u:\Pz\to G$ w.r.t. $K$ satisfying
\begin{align*}
u(p_0k)&=k\inv u(p_0),\\
u(p_0h)&=h\inv u(p_0) \widehat{C}(p_0,h),
\end{align*}
for any $(p_0,k,h)\in\Pz\times K\times H$, where $\widehat{C}:\Pz\times G\to G$ satisfies
\begin{align*}
\widehat{C}(p_0,g_1g_2)=\widehat{C}(p_0,g_1)\widehat{C}(p_0g_1,g_2),\hspace{1cm}\mathrm{and}\hspace{1cm}\widehat{C}(p_0,k)&=e
\end{align*}
for any $p_0\in\Pz$, $(g_1,g_2)\in G^2$ and $k\in K$. These two conditions imply that $\widehat{C}(p_0,gk)=\widehat{C}(p_0,g)$ and $\widehat{C}(p_0k,g)=\widehat{C}(p_0,g)$ for any $(p_0,k,g)\in\Pz\times K\times G$. Thus $\widehat{C}$ induces a map $C:\manP\times H\to G,\ (p,h)\mapsto \widehat{C}(p_0,h)$, where $p_0\in\pi_K\inv(p)$, and $C$ satisfies the cocycle relation \eqref{cocycle relation}.

This setup appears in the symmetry reduction of the conformal Cartan bundle to the tractor bundle \cite{attard_tractors_2017-1} or the local twistor bundle \cite{attard_tractors_2017}. In this example, $H\defeq SO(r,s)\times W$ is the direct product of the Lorentz group and the group of Weyl dilations $W\defeq \Rl_+\backslash\{0\}$ and $K=K_1=\Rl^{m*}$ (with $r+s=m$) is the group of conformal boosts. Let $(V,\rho)$ be a representation of $G$ and consider $\varphi_0\in\calF_\eqv(\Pz,V)$, $\alpha_0\in\Omega_\tens^\bullet(\Pz,V)$ and $\omega_0\in \calA(\Pz)$ 
be standard gauge fields on $\Pz$. These gauge fields are pulled back by $\Sigma:\Pz\to\Pz,\ p_0\mapsto p_0u(p_0)$ to define dressed fields
\begin{align*}
&\varphi_0^u(p_0)\defeq \Sigma^*\varphi_0(p_0)=\rho(u(p)\inv)\varphi_0(p_0)\\
&\alpha^u_{0|p_0}\defeq\Sigma^*\alpha_{0|p_0}=\rho(u(p_0)\inv)\alpha_{0|p_0}\\
&\omega^u_{0|p_0}\defeq\Sigma^*\omega_{0|p_0}=u(p_0)\inv\omega_{0|p_0}u(p_0)+u(p_0)\inv\dr u_{|p_0}.
\end{align*}
One checks that dressed fields are invariants wrt $K$. Thus, they induce fields on the quotient bundle $\manP(M,H)$.
\begin{align*}
&\varphi:\manP\to V,\ p\mapsto \varphi^u_0(p_0),\\
&\alpha\in\Omega^\bullet(\manP,V),\ \alpha_{|p}\defeq\alpha_{0|p_0}\\
&\omega\in\Omega^1(\manP,\kg),\ \omega_{|p}\defeq\omega_{0|p_0}
\end{align*}
where $p_0\in\pi_K\inv(p)$. The latter are twisted\footnote{Note that in this example, $C$ satisfies $C(p,h)=h$ whenever $h\in SO(r,s)\subset H$. Thus the fields are trivially twisted wrt the subgroup of Lorentz transformations. This does not impact the general argument.} gauge fields on $\manP$ wrt the cocycle $C$. Applying the untwisting method, one defines $\manS(M,H\times G)=\manP\times G$ and $\manQ(M,G)=\manS/H$. The twisted fields $\varphi$, $\alpha$ and $\omega$ are seen as pullbacks of standard fields on $\manQ$.

In the current example, the bundle $\manQ$ happens to have the same structure group as the initial bundle $\Pz$. An even stronger result holds.
\begin{proposition}
$\manQ$ is isomorphic to $\Pz$.
\end{proposition}
\begin{proof}
The isomorphism is provided by the map $\Phi:\Pz\to\manQ,\ p_0\mapsto [\pi_K(p_0),u(p_0)\inv]$. On the one hand, for any $p_0\in\Pz$,
\begin{align*}
\pi_\manQ\circ\Phi(p_0)=\pi_\manQ([\pi_K(p_0),u(p_0)\inv])=\pi_\manP(\pi_K(p_0))=\pi_{\Pz}(p_0).
\end{align*}
Thus $\Phi$ covers $\id_M$. On the other hand, for any $p_0\in\Pz$ and $g\in G$ written $g=hk$ for a unique $(h,k)\in H\ltimes K$,
\begin{align*}
\Phi(p_0hk)&=[\pi_K(p_0hk),u(p_0hk)\inv]\\
&=[\pi_K(p_0)h,\widehat{C}(p_0,hk)\inv u(p_0)\inv hk]\\
&=[\pi_K(p_0)h,C(\pi_K(p_0),h)\inv u(p_0)\inv hk]\\
&=[\pi_K(p_0), u(p_0)\inv ]hk\\
&=\Phi(p_0)hk
\end{align*}
Thus $\Phi$ is a morphism of $G$-principal bundles covering $\id_M$, hence an isomorphism.
\end{proof}
This particular property does not appear in general, as shown in the following.

Consider the principal bundle $\Pz(M,G)$ where $G=H'\ltimes K'$ and a Lie group $\bar{G}$ such that $G\subset \bar{G}$. In the second symmetry reduction of the conformal Cartan bundle \cite{attard_tractors_2017-1, attard_tractors_2017}, $H'=W$, $K'=SO(r,s)\ltimes \Rl^{(m)*}$, $G$ is unchanged and $\bar{G}=GL_{m+2}(\Rl)$. Assume that the representation $(V,\rho)$ of $G$ is induced by a representation $(V,\bar{\rho})$ of $\bar{G}$.

Assume that there exists a partial dressing field $u':\Pz\to \bar{G}$ wrt $K'$ satisfying
\begin{align*}
u'(p_0k')&=k^{\prime-1} u'(p_0),\\
u'(p_0h')&=h^{\prime-1} u'(p_0) \widehat{C}'(p_0,h'),
\end{align*}
for any $(p_0,k',h')\in\Pz\times K'\times H'$, where $\widehat{C}':\Pz\times G\to \bar{G}$ satisfies
\begin{align*}
\widehat{C}'(p_0,g_1g_2)=\widehat{C}'(p_0,g_1)\widehat{C}'(p_0g_1,g_2),\hspace{1cm}\mathrm{and}\hspace{1cm}\widehat{C}'(p_0,k')=e
\end{align*}
for any $p_0\in\Pz$, $(g_1,g_2)\in G^2$ and $k'\in K'$. Thus $\widehat{C}'(p_0,gk')=\widehat{C}'(p_0,g)$ and $\widehat{C}'(p_0k',g)=\widehat{C}'(p_0,g)$ for any $(p_0,k',g)\in\Pz\times K'\times G$. Thus $\widehat{C}'$ induces a map $C':\manP'\times H'\to \bar{G},\ (p',h')\mapsto \widehat{C}'(p_0,h')$, where $\manP'=\Pz/H'$, $p_0\in\pi_K^{\prime-1}(p')$, and $C'$ satisfies the cocycle relation \eqref{cocycle relation}.

One can define dressed fields in an algebraic way, using the appropriate action of $\bar{G}$ on $V$:
\begin{align*}
&\varphi_0^{u'}(p_0)\defeq\bar{\rho}(u'(p)\inv)\varphi_0(p_0)\\
&\alpha^{u'}_{0|p_0}\defeq\bar{\rho}(u'(p_0)\inv)\alpha_{0|p_0}\\
&\omega^{u'}_{0|p_0}\defeq u'(p_0)\inv\omega_{0|p_0}u'(p_0)+u'(p_0)\inv\dr u'_{|p_0}.
\end{align*}
These fields cannot be interpreted as pullbacks of standard fields by a map ``$\Sigma':p_0\mapsto p_0u'(p_0)$'' because $\Pz$ is not a priori\footnote{In \cite{attard_tractors_2017-1, attard_tractors_2017}, the dressed fields considered are only the local version (on $M$) and the algebraic approach is the only one remaining.} endowed with an action of $\bar{G}$. However, they are $K'$-invariant and induce twisted gauge fields on $\manP'$ with cocycle $C'$. These fields are untwisted via the bundles $\manS'(M,H'\times \bar{G})=\manP'\times\bar{G}$ and $\manQ'(M,\bar{G})=\manS'/H'$. Thus, the twisted gauge fields considered here are pullbacks of standard gauge fields on the bundle $\manQ(M,\bar{G})=\manQ(M,GL_{m+2}(\Rl))$.

%%%%%%%%%%%%%%%%
\section{Local Point of View}
\label{sec Local POV}
%%%%%%%%%%%%%%%%

Since $\manP$, $\manS$ and $\manQ$ are principal bundles over $M$, the fields defined on them locally give rise to fields on $M$ by pullback through local sections. The aim of this section is to relate local versions of twisted fields on $\manP$ and the local versions of their corresponding (standard) fields on $\manS$ and $\manQ$.

Therefore, all dressing fields are assumed to be local (\textit{i.e.} defined on $\manP_U$, $\manS_U$ or $\manQ_U$ over $U\subset M$ open and small enough) unless stated explicitly.

%%%%%%%%%%%%%%%%
\subsection{\texorpdfstring{Local Dressing Fields on $\manS$, $\manP$ and $\manQ$}{Local Dressing Fields on S, P and Q}}
\label{subsec Dressing S,P,Q}
%%%%%%%%%%%%%%%%

According to Sec.~\ref{subsec The DFM}, any local section of a principal bundle is in 1:1 correspondence with a local dressing field. It proves useful to draw the link between local fields on $\manS$, $\manP$ and $\manQ$ through the relation between their local dressing fields.

\begin{proposition}\label{prop decomposition u_S}
Any dressing field $u_\manS\in\calF_\eqv(\manS_U,H\times G)$ reads $(\pi_G^*u_\manP,\pi_H^*u_\manQ)$ for some dressing fields $u_\manP\in\calF_\eqv(\manP_U,H)$ and $u_\manQ\in\calF_\eqv(\manQ_U,G)$.
\end{proposition}

\begin{proof}
Let $u_\manS\in\calF_\eqv(\manS_U,H\times G)$ be a dressing field on $\manS_U$ over $U$. This map reads $u_\manS=(u_H,u_G)$ for some $u_H:\manS_U\rightarrow H$ and $u_G:\manS_U\rightarrow G$. The equivariance property of $u_\manS$ implies that
%\begin{align*}
\begin{multline*}
\calR_{(h,g')}^{\manS*}u_\manS(p,g)
= (h\inv,g^{\prime-1})u_\manS(p,g)
\\
\begin{aligned}
&\Longleftrightarrow (u_H(ph,C(p,h)\inv gg'),u_G(ph,C(p,h)\inv gg'))
= (h\inv u_H(p,g),g^{\prime-1}u_G(p,g))
\\
&\Longleftrightarrow \left\{ 
\begin{aligned}
u_H(ph,C(p,h)\inv gg') &= h\inv u_H(p,g),\\
u_G(ph,C(p,h)\inv gg') &= g^{\prime-1}u_G(p,g).
\end{aligned} \right.
\end{aligned}
\end{multline*}
%\end{align*}
This being valid for any $(h,g')\in H\times G$, the first equation implies that $u_H$ is independent of its second variable and $u_H(ph,\bullet)=h\inv u_H(p,\bullet)$. Then, $u_H$ induces a well-defined map $u_\manP:\manP_U\rightarrow H,\ p\mapsto u_H(p,g)$ for any $g\in G$, and $u_\manP(ph)=h\inv u_\manP(p)$. Thus $u_\manP$ is a $H$-dressing field on $\manP_U$ and $u_H=\pi_G^*u_\manP$. Similarly, the equation on $u_G$ implies that $u_G(p,gg')=g^{\prime-1}u_G(p,g)$ and that $u_G$ only depends on $[p,g]\in\manQ_U$. Then, $u_G$ induces a well-defined map $u_\manQ:\manQ_U\rightarrow G,\ q=[p,g]\mapsto u_G(p,g)$ for any $(p,g)\in\pi_H\inv(q)$, and $u_\manQ(qg')=g^{\prime-1} u_\manQ(q)$. Thus $u_\manQ$ is a $G$-dressing field on $\manQ_U$ and $u_G=\pi_H^*u_\manQ$.
\end{proof}
\begin{remark}
One straightforwardly checks that the converse is true: given $u_\manP\in\calF_\eqv(\manP_U,H)$ and $u_\manQ\in\calF_\eqv(\manQ_U,G)$ two dressing fields, $(\pi_G^*u_\manP,\pi_H^*u_\manQ)$ is a dressing field on $\manS_U$ over $U$.
\end{remark}
This result means that trivializing $\manS$ is equivalent to trivializing both $\manP$ and $\manQ$ independently. This link shows up in the usual approach in terms of local sections.

%%%%%%%%%%%%%%%%
\subsection{\texorpdfstring{Local Sections of $\manS$, $\manP$ and $\manQ$}{Local Sections of S, P and Q}}
\label{subsec Local Sect S,P,Q}
%%%%%%%%%%%%%%%%

Let $\sigma_\manS\in\Gamma(\manS_U)$ be a local trivializing section and let $u_\manS\in\calF_\eqv(\manS_U,H\times G)$ be the dressing field canonically associated to $\sigma_\manS$. For any $x\in U$, $\sigma_\manS(x)=s\cdotact u_\manS(s)=\calR_{u_\manS(p,g)}^\manS(p,g)$ for any $s=(p,g)\in\pi_\manS\inv(x)$. Using Prop.~\ref{prop decomposition u_S}, there exist dressing fields $u_\manP,\ u_\manQ$ such that $u_\manS=(\pi_G^*u_\manP,\pi_H^*u_\manQ)$. Then
\begin{align*}
\sigma_\manS(x)
&= \calR_{u_\manS(p,g)}^\manS(p,g)
= \calR_{(\pi_G^*u_\manP(p,g),\pi_H^*u_\manQ(p,g))}^\manS(p,g)
\\
&
= (p\pi_G^*u_\manP(p,g),C(p,\pi_G^*u_\manP(p,g))\inv g\pi_H^*u_\manQ(p,g))
\\
&= (pu_\manP(p),C(p,u_\manP(p))\inv g u_\manQ([p,g]))
%\\
%&
= (\sigma_\manP(x),C(p,u_\manP(p))\inv g u_\manQ([p,g]))
\end{align*}
where $\sigma_\manP(x) \defeq pu_\manP(p)$ is independent of the choice of $p\in\pi_\manP\inv(x)$. The last equation implies that $\pi_G(\sigma_\manS(x))=\sigma_\manP(x)$. Similarly,
\begin{align*}
\pi_H(\sigma_\manS(x))
&= [p u_\manP(p), C(p,u_\manP(p))\inv g u_\manQ([p,g])]
= [p, g u_\manQ([p,g])]
= [p, g] u_\manQ([p,g])
\\
&
= q u_\manQ(q) 
\rdefeq \sigma_\manQ(x).
\end{align*}
Then, a section $\sigma_\manS$ of $\manS_U$ over $U$ provides a section $\sigma_\manP=\pi_G\circ\sigma_\manS$ of $\manP_U$ and a section $\sigma_\manQ=\pi_H\circ\sigma_\manS$ of $\manQ_U$ over $U$. Note that, following Prop.~\ref{prop decomposition u_S}, the converse is also true. Given $\sigma_\manP,\ \sigma_\manQ$ two sections, one can construct a section $\sigma_\manS$. The explicit process is less straightforward. Either one has to construct $u_\manP,\ u_\manQ$ associated to $\sigma_\manP,\ \sigma_\manQ$, then to define $u_\manS \defeq (\pi_G^*u_\manP,\pi_H^*u_\manQ)$ and to define $\sigma_\manS$ from $u_\manS$. Or one can define $\sigma_\manS$ implicitly from $\sigma_\manP,\ \sigma_\manQ$ as follows. Let $x\in U$. Since $\sigma_\manQ(x)\in\manQ_U$, there exists $(p,g)\in\manS$ such that $\sigma_\manQ(x)=[p,g]$. Then there exists $h\in H$ such that $p=\sigma_\manP(x)h$. Thus one can rewrite $\sigma_\manQ(x)=[p,g]=[\sigma_\manP(x)h,g]= [\sigma_\manP(x),C(\sigma_\manP(x),h)g]$. One defines $\sigma_\manS(x) \defeq  (\sigma_\manP(x),C(\sigma_\manP(x),h)g)$. In other words, $\sigma_\manS(x)$ is “the representative of the equivalence class $\sigma_\manQ(x)$ whose first component is $\sigma_\manP(x)$”.

However, there exists another parametrization of $\sigma_\manS$ given $\sigma_\manP,\ \sigma_\manQ$. Following Corollary~\ref{corollary isom eq functions}, any (local) dressing field $u_\manQ\in\calF_\eqv(\manQ_U,G)$ corresponds to a so-called (local) “twisted dressing field” $f_u=\tilde{\chi}_{\manP\manQ}(u_\manQ)=i_e^*\pi_H^*u_\manQ\in\calF_\Ceq(\manP_U,G)$.

\begin{remark}
There are two ways to link a twisted dressing field $f_u\in\calF_\Ceq(\manP_U,G)$ to a section $\sigma_\manQ\in\Gamma(\manQ_U)$.
\begin{itemize}
\item On the one hand, there is a 1:1 correspondence $f_u\leftrightarrow u_\manQ \in \calF_\eqv(\manQ_U, G)$ given by $f_u = i_e^* \pi_H^* u_\manQ$ and there is a 1:1 correspondence $u_\manQ \leftrightarrow \sigma_\manQ$ given by $\sigma_\manQ(x)=qu_\manQ(q)$. This comes from the fact that $\manQ$ is associated to itself: $\manQ \simeq \manQ \times_G G$.
\item On the other hand, there is a 1:1 correspondence $f_u\leftrightarrow \sigma_\manQ\in\Gamma(\manP_U\times_{C(H)}G)$ given by $\sigma_\manQ(x)=[p, f_u(p)]=[ph, C(p,h) \inv f_u(p)]=[p h, f_u(ph)]$. This comes from the fact that $\manQ$ is (twisted) associated to $\manP$: $\manQ = \manP \times_{C(H)} G$.
\end{itemize}
Both perspectives yield equivalent geometric formulations depending on the chosen framework.
\end{remark}

Using the equivariance property of $u_\manQ$ in the previous calculations,
\begin{align*}
\sigma_\manS(x)
&= (\sigma_\manP(x), C(p,u_\manP(p))\inv g u_\manQ([p,g]))
= (\sigma_\manP(x),C(p,u_\manP(p))\inv g u_\manQ([p,e]g))
\\
&= (\sigma_\manP(x),C(p,u_\manP(p))\inv u_\manQ([p,e]))
= (\sigma_\manP(x), u_\manQ([p,e]C(p,u_\manP(p))))
\\
&= (\sigma_\manP(x), u_\manQ([p,C(p,u_\manP(p))]))
= (\sigma_\manP(x), u_\manQ([pu_\manP(p),e]))
\\
&= (\sigma_\manP(x), u_\manQ([\sigma_\manP(x),e]))
= (\sigma_\manP(x), i_e^*\pi_H^*u_\manQ(\sigma_\manP(x)))
\\
&= (\sigma_\manP(x), f_u(\sigma_\manP(x)))
= i_{f_u}\circ\sigma_\manP(x)
\end{align*}
This way, any local section of $\manS_U$ appears as being a section of $\manP_U$ embedded in $\manS_U$ through a twisted dressing field (which carries the same information as a local section of $\manQ_U$). Note that this specific parametrization requires that $\manP$ is a submanifold of $\manS$. In particular, one cannot define a similar parametrization exchanging the roles of $\manP$ and $\manQ$ since the latter is merely a homogeneous space.
\begin{remark}
It is worth insisting on the interpretation of $f$ and $f_u$, since both share common features and related uses. On the one hand, $f\in\calF(\manP,G)$ parametrizes cocycles within an equivalence class. Thus, it is interpreted as a coboundary. The same $f$ provides an embedding of $\manP$ in $\manS$ through $i_f$. These two roles of $f$ are the reason why the cohomology class of $C$ is linked to the isomorphism class of $\manS$ in Theorem~\ref{theorem eq C equiv isom S}. \textit{A priori}, $f$ satisfies no particular equivariance relation. On the other hand, $f_u\in\calF_\Ceq(\manP_U,G)$ is the twisted version of a local dressing field on $\manQ_U$. It parametrizes local trivializing sections of $\manQ_U$. $f_u$ is \textit{a priori} local. The following property highlight the consequence of the very restrictive case where $f$ and $f_u$ are identified.
\end{remark}
\begin{proposition}\label{prop C trivial iff Q trivial}
$\manQ$ is a trivial principal bundle if and only if $C$ is a trivial cocycle, \textit{i.e.}
\begin{align*}
C \sim e \Longleftrightarrow \manQ\simeq M\times G.
\end{align*}
\end{proposition}
\begin{proof}
The proof is as follows:
\begin{align*}
\manQ \simeq M\times G
&\Longleftrightarrow \text{There is a global section }\sigma_\manQ\in\Gamma(\manQ) = \Gamma(\manP\times_{C(H)}G)
\\
&\Longleftrightarrow \text{There is a global twisted dressing }f_u\in\calF_\Ceq(\manP,G)
\\
&\Longleftrightarrow \exists f\in\calF(\manP,G),\ f(ph)=C(p,h)\inv f(p)
\\
&\Longleftrightarrow \exists f\in\calF(\manP,G),\ f(p)\inv C(p,h)f(ph)=e
\\
&\Longleftrightarrow C\sim e.
\end{align*}
\end{proof}
\begin{remark}
This result gives the reciprocal implication of the particular case stated in Corollary~\ref{corollary eq C implies isom Q}. Contrary to Theorem~\ref{theorem eq C equiv isom S}, this proposition does not generalize to arbitrary equivalent cocycles $C\sim C'$.
\end{remark}

%%%%%%%%%%%%%%%%
\subsection{Local Fields}
\label{subsec Local Fields}
%%%%%%%%%%%%%%%%

Standard and twisted gauge fields on $\manS$, $\manP$ and $\manQ$ induce local fields on a generic open set $U$ by pullback through local trivializing sections. Several relations between sections of these principal fiber bundles were displayed in Sec.~\ref{subsec Local Sect S,P,Q}: given a $\sigma_\manS\in\Gamma(\manS_U)$, one defines $\sigma_\manP \defeq \pi_G\circ\sigma_\manS\in\Gamma(\manP_U)$ and $\sigma_\manQ \defeq \pi_H\circ\sigma_\manS\in\Gamma(\manQ_U)$. Conversely, given $\sigma_\manP$ and $\sigma_\manQ$, one defines $f_u\in\calF_\Ceq(\manP_U,G)$ such that $\sigma_\manQ(x) \rdefeq [p,f_u(p)]$ and $\sigma_\manS \defeq i_{f_u}\circ\sigma_\manP$. Together with these links between sections were proven relations between (twisted or standard) equivariant fields on $\manS$, $\manP$ and $\manQ$ in Sec.~\ref{sec Isomorphism} and \ref{sec Connections GT on S}. These relations between both global fields and sections induce relations between local fields.
\begin{proposition}\label{prop local fields SQ}
The standard gauge fields on $\manQ$ and their corresponding fields on $\manS$ (in the sense of Prop.~\ref{prop isom tensorial forms}, \ref{prop isom connections SQ} and \ref{prop isom GT SQ}) induce the same local fields on $U$.
\end{proposition}
Note that, since these are local versions of standard $G$-fields on $\manQ$, they support a standard action of $\calG_U$ as described in \eqref{eq standard GT}.
\begin{proof}
Let $\sigma_\manS$ be a section of $\manS_U$ over $U$ and set $\sigma_\manQ \defeq \pi_H\circ\sigma_\manS$, a section of $\manQ_U$ over $U$. Consider $\alpha^\manQ\in\Omega_\tens^\bullet(\manQ,V)$, $\omega^\manQ\in\calA(\manQ)$ and $\gamma^\manQ\in\calG^\manQ$. As seen in Prop.~\ref{prop isom tensorial forms}, \ref{prop isom connections SQ} and \ref{prop isom GT SQ}, these fields are in 1:1 correspondence with $\alpha^\manS \defeq \pi_H^*\alpha^\manQ\in\Omega_\tens^\bullet(\manS,V)$, $\omega^\kg \defeq \pi_H^*\omega^\manQ$ which is the $\kg$-valued part of a connection $\omega^\manS\in\calA(\manS)$, and $\gamma^G \defeq \pi_H^*\gamma^\manQ\in\calG^\manS$. The local version of the fields on $\manQ$ induced through $\sigma_\manQ$ are
\begin{align*}
\phya
&\defeq \sigma_\manQ^*\alpha^\manQ
= (\pi_H\circ\sigma_\manS)^*\alpha^\manQ 
= \sigma_\manS^*\pi_H^*\alpha^\manQ
= \sigma_\manS^*\alpha^\manS,
\\
\phyA
&\defeq \sigma_\manQ^*\omega^\manQ
= (\pi_H\circ\sigma_\manS)^*\omega^\manQ
= \sigma_\manS^*\pi_H^*\omega^\manQ
= \sigma_\manS^*\omega^\kg,
\\
\locgamma
&\defeq \sigma_\manQ^*\gamma^\manQ
= (\pi_H\circ\sigma_\manS)^*\gamma^\manQ 
= \sigma_\manS^*\pi_H^*\gamma^\manQ
= \sigma_\manS^*\gamma^G.
\end{align*}
\end{proof}
A similar link exists between the local versions of twisted fields on $\manP$ and of standard fields on $\manQ$.
\begin{proposition}\label{prop local fields PQ}
The twisted gauge fields on $\manP$ and their corresponding fields on $\manQ$ induce the same local fields on $U$ for a local trivialization of $\manQ$ induced by a local trivialization of $\manP$, as an associated bundle.
\end{proposition}
\begin{proof}
This proof relies on the use of a section $\sigma_\manQ$ induced by a section $\sigma_\manP$. Let $\sigma_\manP\in\Gamma(\manP_U)$. Using the fact that $\manQ$ is a bundle associated to $\manP$, one defines $\sigma_\manQ:x\mapsto [\sigma_\manP(x),e]$, \textit{i.e.} $\sigma_\manQ \defeq \pi_H\circ i_e\circ\sigma_\manP$.

Consider $\alpha^\manQ\in\Omega_\tens^\bullet(\manQ,V)$, $\omega^\manQ\in\calA(\manQ)$ and $\gamma^\manQ\in\calG^\manQ$. As seen in Prop.~\ref{prop isom tensorial forms}, \ref{prop isom connections PQ} and \ref{prop isom GT PQ}, these fields are in 1:1 correspondence with $\alpha^\manP \defeq i_e^*\pi_H^*\alpha^\manQ\in\Omega_\Ctens^\bullet(\manP,V)$, $\omega^\manP \defeq i_e^*\pi_H^*\omega^\manQ\in\calA^C(\manP)$ and $\zeta^\manP \defeq i_e^*\pi_H^*\gamma^\manQ\in\calG^\manP$. The local version of the fields on $\manQ$ induced through $\sigma_\manQ$ are
\begin{align*}
\phya
&\defeq \sigma_\manQ^*\alpha^\manQ
= (\pi_H\circ i_e\circ\sigma_\manP)^*\alpha^\manQ 
= \sigma_\manP^*i_e^*\pi_H^*\alpha^\manQ
= \sigma_\manP^*\alpha^\manP,
\\
\phyA
&\defeq \sigma_\manQ^*\omega^\manQ
= (\pi_H\circ i_e\circ\sigma_\manP)^*\omega^\manQ 
= \sigma_\manP^*i_e^*\pi_H^*\omega^\manQ
= \sigma_\manP^*\omega^\manP,
\\
\locgamma
&\defeq \sigma_\manQ^*\gamma^\manQ
= (\pi_H\circ i_e\circ\sigma_\manP)^*\gamma^\manQ 
= \sigma_\manP^*i_e^*\pi_H^*\gamma^\manQ=\sigma_\manP^*\zeta^\manP.
\end{align*}
\end{proof}
Let us insist on the main difference of the last two proofs. In the former, the (global) gauge fields on $\manQ$ are naturally induced by the $H$-invariant fields on $\manS$. Then, at fixed $\sigma_\manQ$, one could perform the proof with any $\sigma_\manS$ such that $\pi_H\circ\sigma_\manS=\sigma_\manQ$ (the $\manP$ part of $\sigma_\manS$ is arbitrary.) On the contrary, the proof of Prop.~\ref{prop local fields PQ} involves a specific section, which is compatible with the isomorphism $i_e^*\pi_H^*$. Given $\sigma_\manP$, the others sections must be $\sigma_\manS:x\mapsto (\sigma_\manP(x),e)$ and $\sigma_\manQ:x\mapsto [\sigma_\manP(x),e]$.

This last constraint echoes the distinction made at the end of Sec.~\ref{subsec Isom Connection PQ} between gauge transformation associated to the geometric structure group $H$ and the algebraic one $G$. Indeed, these active gauge transformations locally correspond to passive ones, \textit{i.e.} gluing functions for $\sigma_\manP$ and $\sigma_\manQ$. Without constraint, one could relate two sections $\sigma_\manQ,\ \sigma_\manQ'$ by any $g^\manQ:U\rightarrow G$. However, assuming that $\sigma_\manQ,\ \sigma_\manQ'$ are induced by $\sigma_\manP,\ \sigma_\manP'$ and denoting $g^H:U\to H$ such that $\sigma_\manP'(x)=\sigma_\manP(x)g^H(x)$, one has
\begin{align*}
\sigma_\manQ'(x)
&= [\sigma_\manP'(x),e]
= [\sigma_\manP(x)g^H(x),e] 
= [\sigma_\manP(x),C(\sigma_\manP(x),g^H(x))] 
= [\sigma_\manP(x),e]C(\sigma_\manP(x),g^H(x))
\\ 
&= \sigma_\manQ(x)C(\sigma_\manP(x),g^H(x)).
\end{align*}
Then, $g^G$ must read $g^G(x)=C(\sigma_\manP(x),g^H(x))$ for some $\sigma_\manP$ and $g^H$. With no constraint, $g^G$ is the local representative of an active gauge transformation in $\calG^\manP$, while with constraint, $g^H$ is the local representative of an active gauge transformation $\gamma^\manP\in\calH^\manP$ and $g^G(x)=C(\sigma_\manP(x),g^H(x))$ corresponds to $C(p,\gamma(p))$.\footnote{There is no effective constraint if, at any $p\in\manP$, any $g\in G$ reads $g=C(p,h)$ for some $h\in H$.}

Given Prop.~\ref{prop local fields PQ}, one can interpret any model on $M$ involving twisted fields as a standard model on $M$ with local gauge transformations valued in $G$. For instance, one can define a manifestly gauge-invariant Lagrangian (w.r.t. the local gauge group $\calG_U$) involving the previously-defined fields on $U$. Examples of Yang-Mills Lagrangian and Gravity Gauge Lagrangians dealing with local twisted fields were introduced in \cite{francois_twisted_2021}. The $\calG_U$-invariance of such Lagrangian is enough to work with local twisted fields and get gauge-independent predictions.

\appendix
%%%%%%%%%%%%%%%%
\section{Proof of Lemma~\ref{lemma Technical Lemma}}
\label{sec proof technical lemma}
%%%%%%%%%%%%%%%%

\begin{proof}
\textbf{Item \textit{(i)}}

Let $\phi_X^\manP$ (resp. $\phi_X^G$) be the flow of $X_p^\manP$ at $p$ (resp. of $X_g^G$ at $g$).
\begin{align*}
\Tg_{(p,g)} \calR_{(h,e)}^\manS(X_p^\manP\oplus X_g^G)
&= \Tg_{(p,g)}\calR_{(h,e)}^\manS\left( \dert{(\phi_X^\manP(t),\phi_X^G(t))} \right)
\\
&= \dert{(\phi_X^\manP(t)h,C(\phi_X^\manP(t),h)\inv\phi_X^G(t))}
\\
&= 
\begin{multlined}[t]
\dert{(\phi_X^\manP(t)h,C(p,h)\inv g)} 
%\\
+ \dert{(ph,C(\phi_X^\manP(t),h)\inv g)} 
\\
+ \dert{(ph,C(p,h)\inv\phi_X^G(t))}
\end{multlined}
\\
&= 
\begin{multlined}[t]
\Tg_{ph}i_{C(p,h)\inv g}\dert{\phi_X^\manP(t)h} 
%\\
+ \Tg_{C(p,h)\inv g}i_{ph}\dert{C(\phi_X^\manP(t),h)\inv g} 
\\
+ \Tg_{C(p,h)\inv g}i_{ph}\dert{C(p,h)\inv\phi_X^G(t)}
\end{multlined}
\\
&= 
\begin{multlined}[t]
\Tg_{ph}i_{C(p,h)\inv g}\circ\Tg_p\calR_h^\manP(X_p^\manP) 
\\
+\ Tg_{C(p,h)\inv g}i_{ph}\circ\Tg_{C(p,h)\inv}\Rr_g\circ\dr_\manP C_{(p,h)}\inv(X_p^\manP) 
\\
+ \Tg_{C(p,h)\inv g}i_{ph}\circ\Tg_g\Lr_{C(p,h)\inv}(X_g^G)
\end{multlined}
\\
&= 
\begin{multlined}[t]
(\Tg_p\calR_h^\manP(X_p^\manP)\oplus 0_g) + (0_p\oplus \Tg_{C(p,h)\inv}\Rr_g\circ\dr_\manP C_{(p,h)}\inv(X_p^\manP)) 
\\
+ (0_g\oplus \Tg_g \Lr_{C(p,h)\inv}(X_g^G)) 
\end{multlined}
\\
&= \Tg_p\calR_h^\manP(X_p^\manP)\oplus (\Tg_{C(p,h)\inv} \Rr_g\circ\dr_\manP C_{(p,h)}\inv(X_p^\manP) +\Tg_g\Lr_{C(p,h)\inv}(X_g^G))
\end{align*}
\begin{align*}
\Tg_{(p,g)}\calR_{(e,g')}^\manS(X_p^\manP\oplus X_g^G)
&= \Tg_{(p,g)}\calR_{(e,g')}^\manS\left( \dert{(\phi_X^\manP(t),\phi_X^G(t))} \right)
\\
&= \dert{(\phi_X^\manP(t),\phi_X^G(t)g')}
\\
&= \dert{(\phi_X^\manP(t),gg')}+\dert{(p,\phi_X^G(t)g')}
\\
&= \Tg_pi_{gg'}\dert{\phi_X^\manP(t)}+\Tg_{gg'}i_p\dert{\phi_X^G(t)g'}
\\
&= (X_p^\manP\oplus 0_{gg'})+(0_p\oplus \Tg_g\Rr_{g'}(X_g^G))
\\
&= (X_p^\manP\oplus \Tg_g\Rr_{g'}(X_g^G))
\end{align*}

\medskip
\textbf{Item \textit{(ii)}}

Let $\phi_X^\manP$ be the flow of $X_p^\manP$ at $p$.
\begin{align*}
\Tg_{ph}i_e\circ\Tg_p\calR_h^\manP(X_p^\manP)
&= \dert{ (\phi_X^\manP(t)h,e) }
\\
&= \Tg_{(p,C(p,h))}\calR_{(h,e)}^\manS\left( \dert{ (\phi_X^\manP(t),C(\phi_X^\manP(t),h)) } \right)
\\
&= 
\begin{multlined}[t]
\Tg_{(p,C(p,h))}\calR_{(h,e)}^\manS\left( \dert{ (\phi_X^\manP(t),C(p,h)) } \right)
\\
+ \Tg_{(p,C(p,h))}\calR_{(h,e)}^\manS\left( \dert{ (p,C(\phi_X^\manP(t),h)) } \right)
\end{multlined}
\\
&= 
\begin{multlined}[t]
\Tg_{(p,C(p,h))}\calR_{(h,C(p,h))}^\manS\left( \dert{ (\phi_X^\manP(t),e) } \right)
\\
+ \dert{ (p,C(p,h)\inv C(\phi_X^\manP(t),h)) } 
\end{multlined}
\\
&= \Tg_{(p,C(p,h))}\calR_{(h,C(p,h))}^\manS\circ\Tg_pi_e(X_p^\manP)
+ \left[ C(p,h)\inv\dr_\manP C_{(p,h)}(X_p^\manP) \right]_{(ph,e)}^\ver
\end{align*}

\medskip
\textbf{Item \textit{(iii)}}

Let us proceed by double inclusion. First, consider $\ka\in\kh$. The flow of $\ka_p^\ver\oplus \Tg_e\Rr_g\circ(\dr_GC\inv)_{(p,e)}(\ka)$ at $(p,g)$ is given by $\phi_\ka(t)=(p\exp(t\ka),C(p,\exp(t\ka))\inv g)$.
\begin{align*}
\Tg_{(p,g)}\pi_H \left(\ka_p^\ver\oplus \Tg_e \Rr_g\circ(\dr_G C\inv)_{(p,e)}(\ka)\right)
&= \dert{ \pi_H\left((p\exp(t\ka),C(p,\exp(t\ka))\inv g)\right) }
\\
&= \dert{ \pi_H\circ \calR_{(\exp(t\ka),e)}^\manS((p,g)) }
\\
&= \dert{ \pi_H((p,g)) }
\\
&= 0,
\end{align*}
which proves one inclusion. Then, consider $X_{(p,g)}^\manS\in\ker \Tg_{(p,g)}\pi_H$ and let $\phi_X^\manS$ be the flow of $X_{(p,g)}^\manS$ at $(p,g)$. By definition, $\dert{\pi_H(\phi_X^\manS(t))}=0$. This means that $\pi_H(\phi_X^\manS(t))$ is locally constant around $t=0$. Thus, $\phi_X^\manS(t)$ is locally confined in an $H$-equivalence class. Around $t=0$, $\phi_X^\manS(t)$ is parametrized as $(ph(t),C(p,h(t))\inv g)$ for some curve $h:]-\varepsilon,\varepsilon[\rightarrow H$ satisfying $h(0)=e$. Let us call $\ka \defeq \dot{h}(0)\in \Tg_eH\simeq\kh$. Then the tangent vector to $\phi_X^\manS$ at $t=0$ reads
\begin{align*}
\dert{\phi_X^\manS(t)}
&= \dert{(p h(t),C(p,h(t))\inv g)}
\\
&= \dert{(p h(t),g)}+\dert{(p,C(p,h(t))\inv g)}
\\
&= \dert{p h(t)}\oplus 0+0\oplus\dert{C(p,h(t))\inv g}
\\
&= \left[\dert{h(t)}\right]_p^v\oplus 0+0\oplus \Tg_e \Rr_g\circ(\dr_H C\inv)_{(p,e)}\left(\dert{h(t)}\right)
\\
&= \ka_p^\ver\oplus \Tg_e \Rr_g\circ(\dr_HC\inv)_{(p,e)}(\ka),
\end{align*}
which proves the other inclusion.

\medskip
\textbf{Item \textit{(iv)}}

Let $\phi_X^G$ be the flow of $X^G$ at $g$.
\begin{align*}
\Tg_{(p,g)}\pi_H(0_p\oplus X_g^G)
= \dert{[p,\phi_X^G(t)]}
= \dert{[p,g]g\inv\phi_X^G(t)}
= \left[ \Tg_g\Lr_{g\inv}X_g^G \right]_{[p,g]}^\ver\in\Ver_{[p,g]}\manQ.
\end{align*}
So $\Tg_{(p,g)}\pi_H(0_p\oplus \Tg_gG)\subset\Ver_{[p,g]}\manQ$. Note that according to \textit{C}, $\ker \Tg\pi_H\cap(0\oplus \Tg_g G)=\{0\}$. Hence
\begin{align*}
\dim \Tg_{(p,g)} \pi_H(0_p \oplus \Tg_gG)
= \dim 0_p \oplus \Tg_g G
= \dim \kg
= \dim \Ver_{[p,g]}\manQ,
\end{align*}
which concludes the proof.

\medskip
\textbf{Item \textit{(v)}}

On the one hand, $\Tg_{(ph,C(p,h)\inv g)}\pi_H\circ\Tg_{(p,g)}\calR_{(h,e)}^\manS=\Tg_{(p,g)}\pi_H$ is a direct consequence of the definition of $\pi_H$ as projection onto equivalence classes generated by $\calR_{(h,e)}^\manS$. On the other hand, let $\phi_X^\manP$ (resp. $\phi_X^G$) be the flow of $X_p^\manP$ at $p$ (resp. of $X_g^G$ at $g$) and let $g'\in G$.
\begin{align*}
\Tg_{(p, g g')}\pi_H\circ\Tg_{(p,g)}\calR_{(e,g')}^\manS(X_p^\manP\oplus X_g^G)
&= \Tg_{(p,g g')}\pi_H \left( \dert{(\phi_X^\manP(t),\phi_X^G(t)g')} \right)
\\
&= \dert{[\phi_X^\manP(t),\phi_X^G(t)g']}
= \dert{[\phi_X^\manP(t),\phi_X^G(t)]g'}
\\
&= \Tg_{[p,g]}\calR_{g'}^\manQ\left(\dert{[\phi_X^\manP(t),\phi_X^G(t)]}\right)
\\
&= \Tg_{[p,g]}\calR_{g'}^\manQ\circ\Tg_{(p,g)}\pi_H(X_p^\manP,X_g^G).
\end{align*}

\medskip
\textbf{Item \textit{(vi)}}

This is a direct consequence of \textit{(ii)}, \textit{(iv)} and \textit{(v)}.
%\begin{multline*}
\begin{align*}
\Tg_{(ph,e)} \pi_H\circ\Tg_{ph} i_e \circ \Tg_p \calR_h^\manP(X_p^\manP)
%\\
&\overset{\mathclap{(ii)}}{=}
\begin{multlined}[t]
 \ \Tg_{(ph,e)}\pi_H\bigl( \Tg_{(p,C(p,h))}\calR_{(h,C(p,h))}^\manS\circ\Tg_pi_e(X_p^\manP)\\
+ \bigl[ C(p,h)\inv\dr_\manP C_{(p,h)}(X_p^\manP) \bigr]_{(ph,e)}^\ver \bigr)
\end{multlined}
\\
&\overset{\mathclap{(iv,v)}}{=}\ 
\begin{multlined}[t]
 \ \Tg_{[p,e]}\calR_{C(p,h)}^\manQ \circ \Tg_{(p,e)}\pi_H \circ \Tg_pi_e(X_p^\manP)\\
+ \left[ C(p,h)\inv\dr_\manP C_{(p,h)}(X_p^\manP) \right]_{[ph,e]}^\ver,
\end{multlined}
\end{align*}
where the use of \textit{(iv)} is possible because $\left[ C(p,h)\inv\dr_\manP C_{(p,h)}(X_p^\manP) \right]_{(ph,e)}^\ver$ is a vertical vector in $\Tg_{(ph,e)}\manS$ tangent to $G$.
\end{proof}

\bibliography{biblio.bib}

\end{document}